\documentclass[a4paper]{article}
\usepackage{enumitem}
\usepackage{calc}

\usepackage{lmodern} 

\usepackage{graphicx} 

\usepackage{amsmath, accents}
\usepackage{amssymb,tikz}
\usepackage{pict2e}
\usepackage{amsthm}
\usepackage{amscd}
\usepackage{mathrsfs}
\usepackage{todonotes}
\usetikzlibrary{calc} 

\usepackage{yfonts}

\usepackage{marvosym}
\usepackage[percent]{overpic}

\newtheorem{theorem}{Theorem} [section]
	
\newtheorem{corollary}[theorem]{Corollary}	
\newtheorem{lemma}[theorem]{Lemma}

\newtheorem{remark}[theorem]{Remark}
\theoremstyle{definition}

\newcommand{\C}{\mathbb{C}}
\newcommand{\R}{\mathbb{R}}

\newcommand{\im}{\text{\upshape Im\,}}

\tikzset{
	master/.style={
		execute at end picture={
			\coordinate (lower right) at (current bounding box.south east);
			\coordinate (upper left) at (current bounding box.north west);
		}
	},
	slave/.style={
		execute at end picture={
			\pgfresetboundingbox
			\path (upper left) rectangle (lower right);
		}
	}
}

\let\oldbibliography\thebibliography
\renewcommand{\thebibliography}[1]{\oldbibliography{#1}
\setlength{\itemsep}{-0.5pt}}

\def\XXint#1#2#3{{\setbox0=\hbox{$#1{#2#3}{\int}$}
\vcenter{\hbox{$#2#3$}}\kern-.5\wd0}}

\usepackage{tikz}
\usetikzlibrary{arrows}
\usetikzlibrary{decorations.pathmorphing}
\usetikzlibrary{decorations.markings}
\usetikzlibrary{patterns}
\usetikzlibrary{automata}
\usetikzlibrary{positioning}
\usepackage{tikz-cd}
\tikzset{->-/.style={decoration={
				markings,
				mark=at position #1 with {\arrow{latex}}},postaction={decorate}}}
	
	\tikzset{-<-/.style={decoration={
				markings,
				mark=at position #1 with {\arrowreversed{latex}}},postaction={decorate}}}

\usetikzlibrary{shapes.misc}\tikzset{cross/.style={cross out, draw, 
         minimum size=2*(#1-\pgflinewidth), 
         inner sep=0pt, outer sep=0pt}}

\usepackage{pgfplots}
	
\allowdisplaybreaks	

\numberwithin{equation}{section}

\def\bigO{{\cal O}}

\usepackage[colorlinks=true]{hyperref}
\hypersetup{urlcolor=blue, citecolor=red, linkcolor=blue}

\begin{document}
\title{\vspace*{-1.5cm} Disk counting statistics near hard edges of random normal matrices: the multi-component regime}
\author{Yacin Ameur$^*$, Christophe Charlier$^*$, Joakim Cronvall\footnote{Centre for Mathematical Sciences, Lund University, 22100 Lund, Sweden. e-mails: yacin.ameur@math.lu.se,  christophe.charlier@math.lu.se, joakim.cronvall@math.lu.se
} \, and Jonatan Lenells\footnote{Department of Mathematics, KTH Royal Institute of Technology, 10044 Stockholm, Sweden. e-mail: jlenells@kth.se}}

\maketitle

\begin{abstract}
We consider a two-dimensional point process whose points are separated into two disjoint components by a hard wall, and study the multivariate moment generating function of the corresponding disk counting statistics. We investigate the ``hard edge regime" where all disk boundaries are a distance of order $\frac{1}{n}$ away from the hard wall, where $n$ is the number of points. We prove that as $n \to + \infty$, the asymptotics of the moment generating function are of the form
\begin{align*}
& \exp \bigg(C_{1}n + C_{2}\ln n + C_{3} + \mathcal{F}_{n} + \frac{C_{4}}{\sqrt{n}} + \bigO(n^{-\frac{3}{5}})\bigg),
\end{align*}
and we determine the constants $C_{1},\dots,C_{4}$ explicitly. The oscillatory term $\mathcal{F}_{n}$ is of order $1$ and is given in terms of the Jacobi theta function. Our theorems allow us to derive various precise results on the disk counting function. For example, we prove that the asymptotic fluctuations of the number of points in one component are of order $1$ and are given by an oscillatory discrete Gaussian. Furthermore, the variance of this random variable enjoys asymptotics described by the Weierstrass $\wp$-function. 
\end{abstract}


\section{Introduction and statement of results}\label{section: introduction}

In recent years there have been a lot of works on counting statistics of two dimensional point processes, see e.g. \cite{CE2020, LMS2018, L et al 2019, SDMS2020, BKLL2021, Charlier 2d jumps, SDMS2021, AkemannSungsoo, ChLe2022, BC2022, ACCL2022} and references therein. The common feature of these works is that they all deal exclusively with models for which the points condensate on a single connected component (``the one-component regime"). In this paper we deviate from these earlier works in that we study the disk counting statistics of a Coulomb gas (at inverse temperature $\beta=2$) whose points are separated into two disjoint components by a hard wall. Let us now introduce the Coulomb gas model investigated in this work.

\medskip The Mittag-Leffler ensemble is the following joint probability density function
\begin{align}\label{def of point process}
\frac{1}{n!Z_{n}} \prod_{1 \leq j < k \leq n} |\zeta_{k} -\zeta_{j}|^{2} \prod_{j=1}^{n}|\zeta_{j}|^{2\alpha}e^{-n |\zeta_{j}|^{2b}}, \qquad \zeta_{1},\dots,\zeta_{n} \in \mathbb{C},
\end{align}
where $b>0$ and $\alpha > -1$ are fixed parameters and $Z_{n}$ is the normalization constant. As $n \to + \infty$, with high probability the random points $\zeta_{1},\dots,\zeta_{n}$ accumulate on the disk centered at $0$ of radius $b^{-\smash{\frac{1}{2b}}}$ according to the probability measure $\mu(d^{2}z) =  \smash{\frac{b^{\smash{2}}}{\pi}}|z|^{\smash{2b-2}}d^{\smash{2}}z$ \cite{HM2013, SaTo}. This determinantal point process generalizes the complex Ginibre process (which corresponds to $(b,\alpha)=(1,0)$) and has attracted a lot of attention over the years, see e.g. \cite{AV2003, AKS2018, BS2021}.

\medskip \medskip In this paper we focus on the Mittag-Leffler ensemble with a hard wall that separates the random points into two disjoint components. To be precise, let $0<\rho_{1}<\rho_{2}<b^{-\frac{1}{2b}}$. We consider the probability density
\begin{align}\label{def of point process hard}
\frac{1}{n!\mathcal{Z}_{n}} \prod_{1 \leq j < k \leq n} |z_{k} -z_{j}|^{2} \prod_{j=1}^{n} e^{-nQ(z_{j})}, \qquad z_{1},\dots,z_{n} \in \mathbb{C},
\end{align}
where $\mathcal{Z}_{n}$ is the normalization constant and
\begin{align}\label{def of Q}
Q(z) = \begin{cases}
|z|^{2b} - \frac{2\alpha}{n}\ln |z|, & \mbox{if } |z| \in [0,\rho_{1}]\cup [\rho_{2},+\infty), \\
+\infty, & \mbox{otherwise.}
\end{cases}
\end{align}
Because $\rho_{1},\rho_{2} < b^{-\frac{1}{2b}}$, the macroscopic behavior of \eqref{def of point process hard} is different from that of \eqref{def of point process}; it is described by a probability measure $\mu_{h}$ which is supported on $\{z \in \mathbb{C}:|z| \in [0,\rho_{1}]\cup [\rho_{2},b^{-\frac{1}{2b}}] \}$ and has a singular component on the circles of radii $\rho_{1}$ and $\rho_{2}$. This measure can be computed using standard balayage techniques \cite{SaTo} (we provide the details of this computation in Appendix \ref{appendix:eq measure}) and is given by
\begin{align}\label{def of muh}
\mu_{h}(d^{2}z) = 2b^{2}r^{2b-1}\chi_{[0,\rho_{1}]\cup [\rho_{2},b^{-\frac{1}{2b}}]}(r) dr\frac{d\theta}{2\pi} + \sigma_{1} \delta_{\rho_{1}}(r) dr \frac{d\theta}{2\pi} + \sigma_{2} \delta_{\rho_{2}}(r) dr \frac{d\theta}{2\pi}, 
\end{align} 
where $z=re^{i\theta}$, $r>0$, $\theta \in (-\pi,\pi]$ and 
\begin{align}\label{def of taustar}
\sigma_{1} = \sigma_{\star}-b\rho_{1}^{2b}, \qquad \sigma_{2} = b\rho_{2}^{2b}-\sigma_{\star}, \qquad \sigma_{\star} := \frac{\rho_{2}^{2b}-\rho_{1}^{2b}}{2\ln(\frac{\rho_{2}}{\rho_{1}})}.
\end{align}
The assumption $0<\rho_{1}<\rho_{2}<b^{-\frac{1}{2b}}$ implies that $b\rho_1^{2b} < \sigma_{\star} < b\rho_2^{2b}$ and hence $\sigma_{1}, \sigma_{2} >0$. The quantity $\sigma_{\star}$ is the mass of $\mu_{h}$ on $\{|z| \leq \rho_{1}\}$. Indeed, straightforward calculations show that
\begin{align}\label{interpretation of sigma star}
\int_{|z| \leq \rho_{1}} \mu_{h}(d^{2}z) = \sigma_{\star}, \qquad \int_{|z| \geq \rho_{2}} \mu_{h}(d^{2}z) = 1 - \sigma_{\star},
\end{align}
which means that for large $n$, the number of $z_{j}$'s on  $\{|z| \leq \rho_{1}\}$ is roughly $\sigma_{\star}n$ with high probability (see also Corollary \ref{coro:prob interpretation of Xn} and the asymptotics \eqref{asymptotics E in the first droplet} and \eqref{asymptotics Var in the first droplet} below).
The point process \eqref{def of point process hard} is an example of a two-dimensional Coulomb gas that is rotation invariant (meaning that the density \eqref{def of point process hard} remains unchanged if all $z_{j}$'s are multiplied by $e^{i\beta}$, $\beta \in \mathbb{R}$). This ensemble can be seen as a conditional process where the points from \eqref{def of point process} are conditioned on the hole event $\mathcal{H}$ that no $\zeta_{j}$'s lie in the annulus centered at 0 of radii $\rho_{1}$ and $\rho_{2}$. The partition function $\mathcal{Z}_{n}$ of \eqref{def of point process hard} is precisely equal to $Z_{n}\mathbb{P}(\mathcal{H})$, and its large $n$ asymptotics were investigated in \cite{Adhikari, AR Infinite Ginibre, Charlier 2d gap}; see also \cite{GHS1988, ForresterHoleProba, JLM1993, APS gap 2009, AK hole 2013, AIE gap 2014, GN2018, L et al 2019} for related works on the hole event. The process \eqref{def of point process hard} can also be realized as the eigenvalues of an $n \times n$ random normal matrix $M$ taken at random according to the probability density proportional to $e^{-n\, \mathrm{tr}Q(M)}dM$, where ``$\mathrm{tr}$" is the trace and $dM$ is the measure on the set of $n \times n$ normal matrices induced by the flat Euclidian metric of $\mathbb{C}^{n\times n}$  \cite{Mehta, CZ1998, ElbauFedler}. Correlation kernels near hard edges have been studied in \cite{ZS2000, NAKP2020, Seo}.

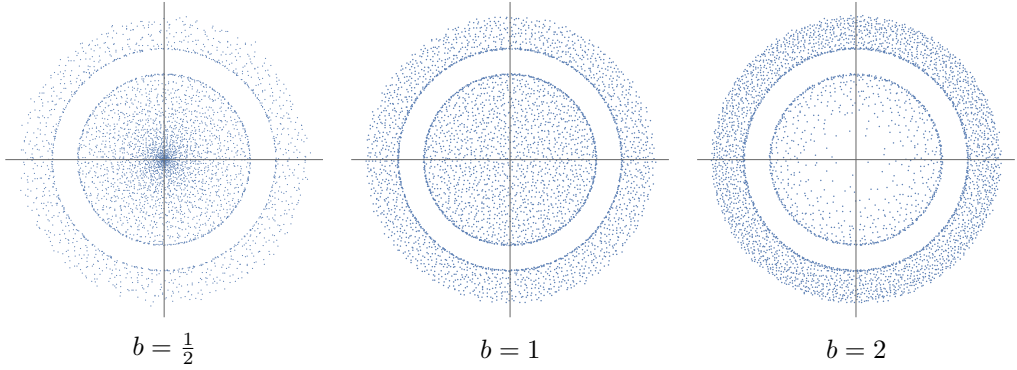
\begin{figure}
\begin{center}
\begin{tikzpicture}[master]
\node at (0,0) {\includegraphics[width=4.2cm]{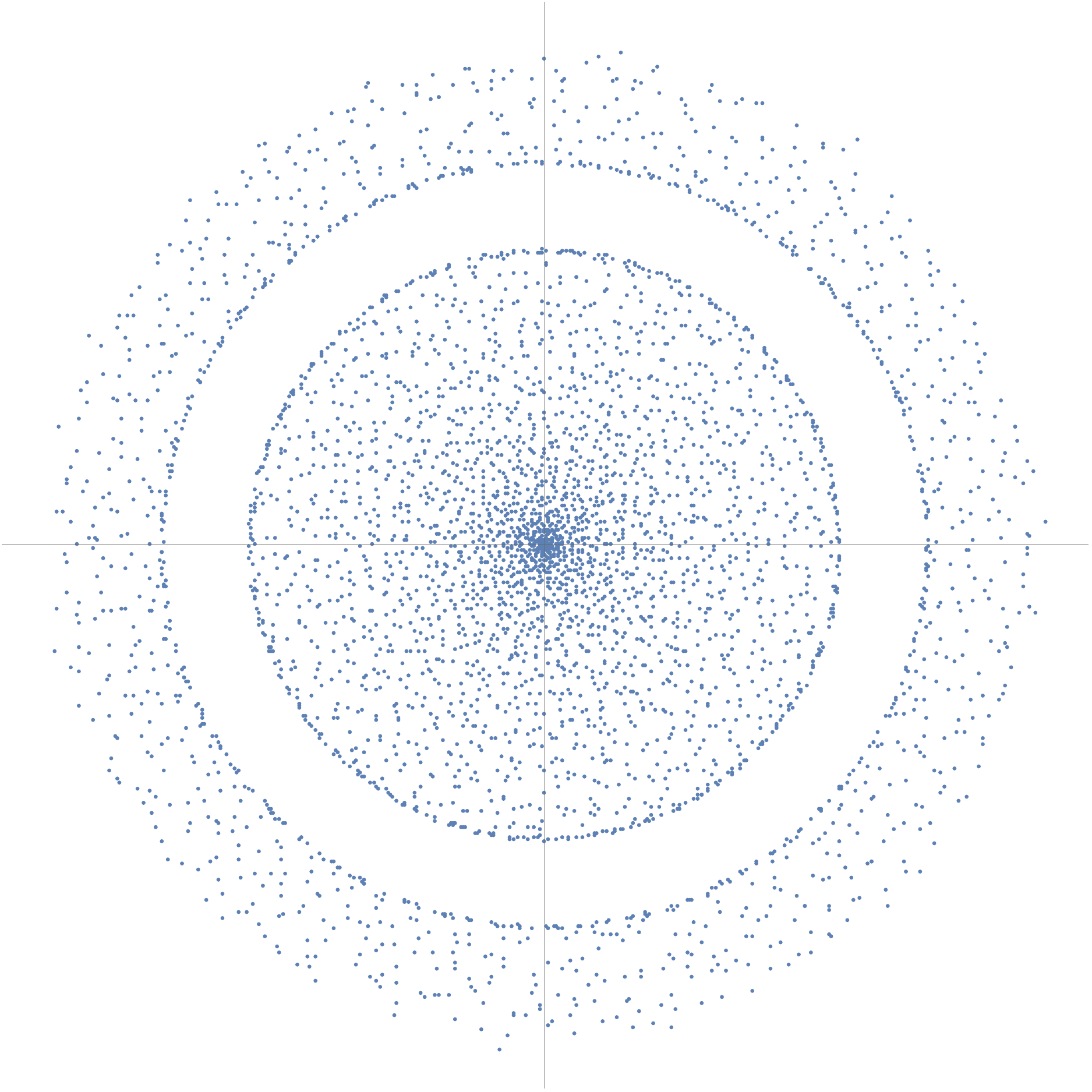}};
\node at (0,-2.5) {$b=\frac{1}{2}$};
\end{tikzpicture}
\begin{tikzpicture}[slave]
\node at (0,0) {\includegraphics[width=4.2cm]{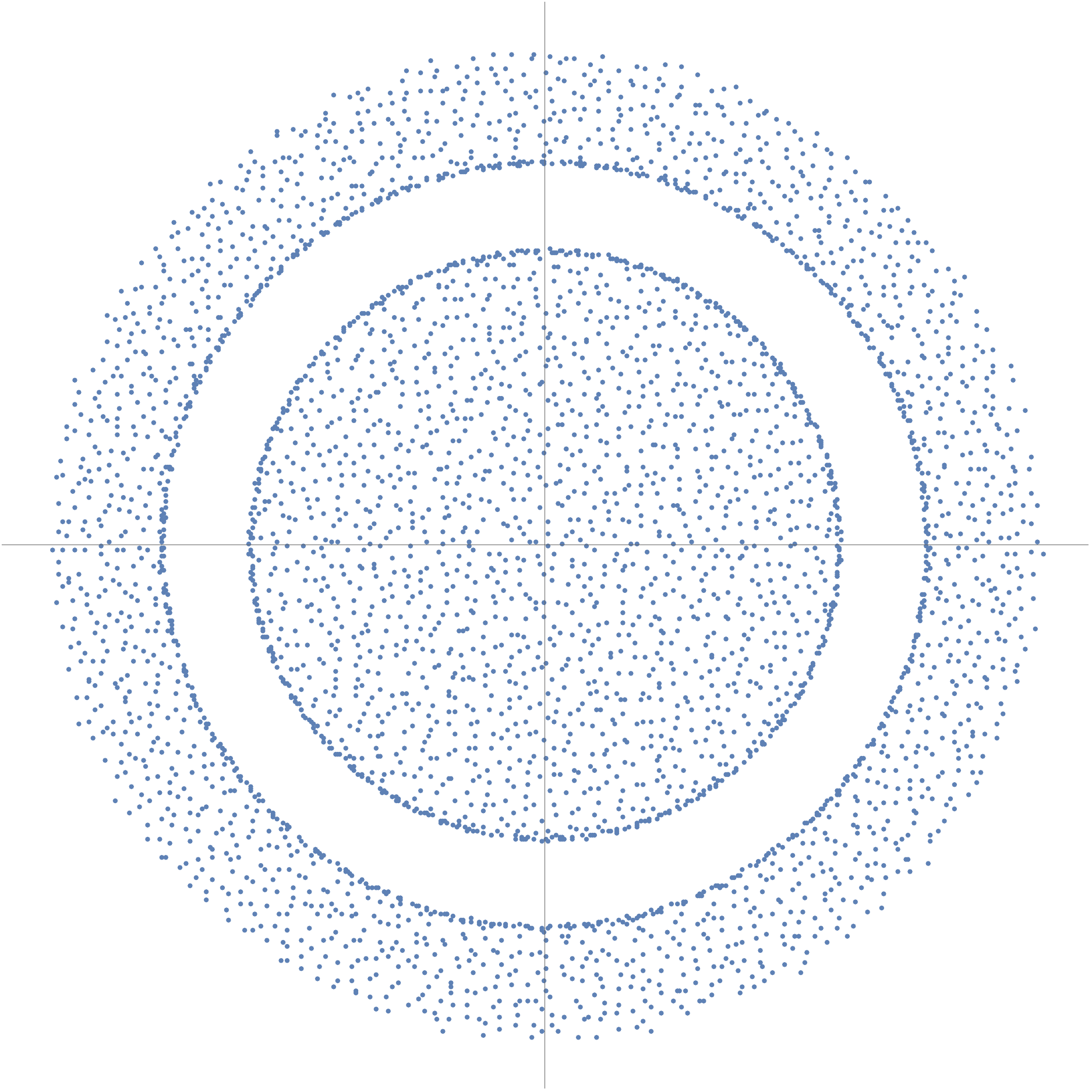}};
\node at (0,-2.5) {$b=1$};
\end{tikzpicture}
\begin{tikzpicture}[slave]
\node at (0,0) {\includegraphics[width=4.2cm]{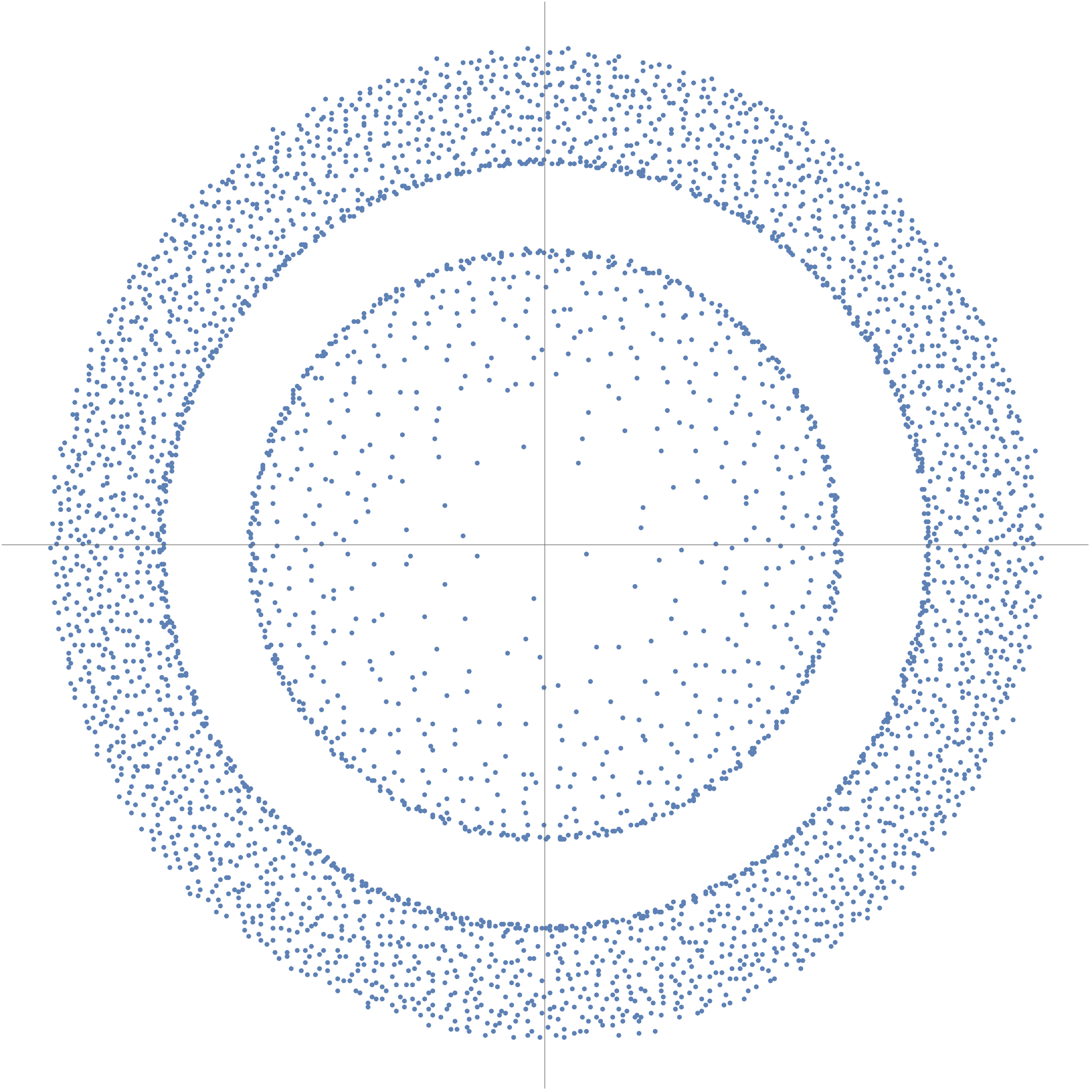}};
\node at (0,-2.5) {$b=2$};
\end{tikzpicture}
\end{center}
\caption{\label{fig: ML with hard wall} Illustration of the point processes corresponding to \eqref{def of point process hard} with $n=4096$, $\rho_{1}=\frac{3}{5}b^{-\frac{1}{2b}}$, $\rho_{2}=\frac{4}{5}b^{-\frac{1}{2b}}$, $\alpha=0$ and the indicated values of $b$.}
\end{figure}

\medskip We will focus on the ``the hard edge regime", i.e. when all disk boundaries are a distance of order $\frac{1}{n}$ away from the hard edges $\{|z|=\rho_{1}\}$ and $\{|z|=\rho_{2}\}$. (Further away from the hard edges, it turns out that there are no oscillations and the situation gets similar to the one-component semi-hard edge regime considered in \cite{ACCL2022}.) Let us now be more specific. Let $\mathrm{N}(y):=\#\{z_{j}: |z_{j}| < y\}$ be the random variable that counts the number of points of \eqref{def of point process hard} in the disk centered at $0$ of radius $y$. Our main result is a precise asymptotic formula as $n \to + \infty$ for the multivariate moment generating function (MGF)
\begin{align}\label{moment generating function intro}
\mathbb{E}\bigg[ \prod_{j=1}^{2m} e^{u_{j}\mathrm{N}(r_{j})} \bigg]
\end{align}
where $m \in \mathbb{N}_{>0}$ is arbitrary (but fixed), $u_{1},\dots,u_{2m} \in \mathbb{R}$, and the radii $r_{1},\dots,r_{2m}$ satisfy $r_{1} < \dots <r_{2m}$ and are merging at a critical speed in the following way
\begin{align}
& r_{\ell} = \rho_{1} \bigg( 1-\frac{t_{\ell}}{n} \bigg)^{\frac{1}{2b}}, & & t_{\ell}\geq 0, \; \ell=1,\dots,m, \label{def of rell part1} \\
& r_{\ell} = \rho_{2} \bigg( 1+\frac{t_{\ell}}{n} \bigg)^{\frac{1}{2b}}, & & t_{\ell}\geq 0, \; \ell=m+1,\dots,2m. \label{def of rell part2}
\end{align}
As $n \to + \infty$, we prove that $\mathbb{E}\big[ \prod_{j=1}^{2m} e^{u_{j}\mathrm{N}(r_{j})} \big]$ enjoys asymptotics of the form
\begin{align}
& \exp \bigg(C_{1}n + C_{2}\ln n + C_{3} + \mathcal{F}_{n} + \frac{C_{4}}{\sqrt{n}} + \bigO(n^{-\frac{3}{5}})\bigg), \label{shape of asymp in hard edge}
\end{align}
and we determine $C_{1},\dots,C_{4},\mathcal{F}_{n}$ explicitly. As corollaries of our various results on the generating function \eqref{moment generating function intro}, we also provide a central limit theorem for the joint fluctuations of $\mathrm{N}(r_{1}),\dots,\mathrm{N}(r_{2m})$, and precise asymptotic formulas for all cumulants of these random variables. Even for $m=1$ our results are new. 

We now introduce the necessary material to present our results. For $j \geq 0$, define
\begin{align}
& \mathsf{T}_{j}(x;\vec{t},\vec{u}) = \sum_{\ell=1}^{m} \omega_{\ell}t_{\ell}^{j}e^{-\frac{t_{\ell}}{b}(x-b\rho_{1}^{2b})}, & & \hat{\mathsf{T}}_{j}(x;\vec{t},\vec{u}) = \sum_{\ell=m+1}^{2m} \omega_{\ell}t_{\ell}^{j}e^{-\frac{t_{\ell}}{b}(b\rho_{2}^{2b}-x)},\label{def of T and That}
\end{align}
where
\begin{align}\label{def of Omega j intro}
\omega_{\ell} = \begin{cases}
e^{u_{\ell}+\dots+u_{2m}}-e^{u_{\ell+1}+\dots+u_{2m}}, & \mbox{if } \ell < 2m, \\
e^{u_{2m}}-1, & \mbox{if } \ell=2m, \\
1, & \mbox{if } \ell=2m+1,
\end{cases}
\end{align}
and $\vec{t}=(t_{1},\dots,t_{2m})$, $\vec{u} = (u_{1},\dots,u_{2m})$. Since the quantities $\mathsf{T}_{0}(b\rho_{1}^{2b};\vec{t},\vec{u})$ and $\hat{\mathsf{T}}_{0}(b\rho_{2}^{2b};\vec{t},\vec{u})$ are independent of $\vec{t}$, we will simply write $\mathsf{T}_{0}(b\rho_{1}^{2b};\vec{u})$ and $\hat{\mathsf{T}}_{0}(b\rho_{2}^{2b};\vec{u})$ instead.
Define
\begin{align}
& f(x;\vec{t},\vec{u}) = \frac{- \big(\frac{b\rho_{1}^{2b}}{x-b\rho_{1}^{2b}} + \frac{\alpha}{b} \big)\mathsf{T}_{1}(x;\vec{t},\vec{u})-\frac{x}{2b}\mathsf{T}_{2}(x;\vec{t},\vec{u})}{1 + \mathsf{T}_{0}(x;\vec{t},\vec{u}) + \hat{\mathsf{T}}_{0}(b\rho_{2}^{2b};\vec{u})}, \label{def of f hard} \\
& \hat{f}(x;\vec{t},\vec{u}) = \frac{ \big(\frac{b\rho_{2}^{2b}}{b\rho_{2}^{2b}-x} - \frac{\alpha}{b} \big)\hat{\mathsf{T}}_{1}(x;\vec{t},\vec{u})+\frac{x}{2b}\hat{\mathsf{T}}_{2}(x;\vec{t},\vec{u})}{1 - \hat{\mathsf{T}}_{0}(x;\vec{t},\vec{u}) + \hat{\mathsf{T}}_{0}(b\rho_{2}^{2b};\vec{u})}. \label{def of fh hard}
\end{align}
Let $\Omega := e^{u_{1}+\dots+u_{2m}}$ and let
\begin{align}\label{mathsfQdef}
\mathsf{Q}(\vec{t}, \vec{u}) := \frac{1 + \mathsf{T}_{0}(\sigma_{\star};\vec{t},\vec{u}) + \hat{\mathsf{T}}_{0}(b\rho_{2}^{2b};\vec{u})}{1 - \hat{\mathsf{T}}_{0}(\sigma_{\star};\vec{t},\vec{u}) + \hat{\mathsf{T}}_{0}(b\rho_{2}^{2b};\vec{u})}.
\end{align}
Our main result involves $\ln \mathsf{Q}(\vec{t}, \vec{u})$ and the next lemma, whose proof is given in Appendix \ref{T0lemmaapp}, shows that this logarithm is well-defined. It also shows that $f$ and $\hat{f}$ are smooth functions of $x \in (b\rho_1^{2b}, b\rho_2^{2b})$. 

\begin{lemma}\label{T0lemma}
Suppose $\vec{u} \in \mathbb{R}^{2m}$, $t_{1}>\dots>t_{m} \geq 0$, and $0 \leq t_{m+1} < \dots < t_{2m}$.
For $x \in [b\rho_1^{2b}, b\rho_2^{2b}]$, it holds that
$$1 + \mathsf{T}_{0}(x;\vec{t},\vec{u}) + \hat{\mathsf{T}}_{0}(b\rho_{2}^{2b};\vec{u}) > 0, \qquad 1 - \hat{\mathsf{T}}_{0}(x;\vec{t},\vec{u}) + \hat{\mathsf{T}}_{0}(b\rho_{2}^{2b};\vec{u}) > 0.$$
\end{lemma}
The Jacobi theta function $\theta(z; \tau)$ is defined for $z \in \mathbb{C}$ and $\im \tau > 0$ by
\begin{align*}
\theta(z; \tau) = \sum_{\ell =-\infty}^\infty e^{\pi i \ell^2 \tau + 2\pi i \ell z}.
\end{align*}
This function satisfies 
\begin{align}\label{theta prop}
\theta(z+1;\tau) = \theta(z;\tau), \qquad \theta(z+\tau;\tau)=e^{-2\pi i z}e^{-\pi i \tau}\theta(z), \qquad \theta(-z) = \theta(z), \qquad \mbox{for all } z \in \mathbb{C},
\end{align}
see also \cite[Chapter 20]{NIST} for further properties. We are now ready to state our main result.

\begin{theorem}[Merging radii at the hard edge: the multi-component regime]\label{thm:main thm hard}
Let $m \in \mathbb{N}_{>0}$, $b>0$, $0 < \rho_1 < \rho_2 < b^{-\frac{1}{2b}}$, $t_{1},\dots,t_{2m} \geq 0$, and $\alpha > -1$ be fixed parameters such that $t_{1}>\dots>t_{m} \geq 0$ and $0 \leq t_{m+1} < \dots < t_{2m}$. For $n \in \mathbb{N}_{>0}$, define
\begin{align}\label{rellhardedge}
& r_{\ell} = \begin{cases}
\rho_{1} \big( 1-\frac{t_{\ell}}{n} \big)^{\frac{1}{2b}}, & \ell=1,\dots,m, \\ 
\rho_{2} \big( 1+\frac{t_{\ell}}{n} \big)^{\frac{1}{2b}}, & \ell=m+1,\dots,2m.
\end{cases}
\end{align}
For any fixed $x_{1},\dots,x_{2m} \in \mathbb{R}$, there exists $\delta > 0$ such that 
\begin{align}\label{asymp in main thm hard}
\mathbb{E}\bigg[ \prod_{j=1}^{2m} e^{u_{j}\mathrm{N}(r_{j})} \bigg] = \exp \bigg( C_{1} n + C_{2} \ln n + C_{3} + \mathcal{F}_{n} +  \frac{C_{4}}{\sqrt{n}} + \bigO\big(n^{-\frac{3}{5}}\big)\bigg), \qquad \mbox{as } n \to + \infty
\end{align}
uniformly for $u_{1} \in \{z \in \mathbb{C}: |z-x_{1}|\leq \delta\},\dots,u_{2m} \in \{z \in \mathbb{C}: |z-x_{2m}|\leq \delta\}$, where
\begin{align}\nonumber
 C_{1} =&\; b \rho_{1}^{2b} \sum_{j=1}^{2m}u_{j} + \int_{b\rho_{1}^{2b}}^{\sigma_{\star}} \ln(1+\mathsf{T}_{0}(x;\vec{t},\vec{u})+\hat{\mathsf{T}}_{0}(b\rho_{2}^{2b};\vec{u}))dx + \int_{\sigma_{\star}}^{b\rho_{2}^{2b}} \ln(1-\hat{\mathsf{T}}_{0}(x;\vec{t},\vec{u})+\hat{\mathsf{T}}_{0}(b\rho_{2}^{2b};\vec{u}))dx,
	\\ \nonumber
 C_{2} = & - \frac{b\rho_{1}^{2b}}{2} \frac{\mathsf{T}_{1}(b\rho_{1}^{2b};\vec{t},\vec{u})}{\Omega} + \frac{b\rho_{2}^{2b}}{2}\hat{\mathsf{T}}_{1}(b\rho_{2}^{2b};\vec{t},\vec{u}), 
	\\ \nonumber
 C_{3} = & - \frac{1}{2} \sum_{j=1}^{2m}u_{j} -  \bigg(\alpha
- \frac{2\ln(\sigma_{2}/\sigma_{1}) + \ln{\mathsf{Q}(\vec{t}, \vec{u})}}{4 \ln(\rho_{2}/\rho_{1})}\bigg) \ln{\mathsf{Q}(\vec{t}, \vec{u})}
	\\ \nonumber
& + \int_{b\rho_{1}^{2b}}^{\sigma_{\star}} \bigg\{ f(x;\vec{t},\vec{u}) + \frac{b \rho_{1}^{2b} \mathsf{T}_{1}(b\rho_{1}^{2b};\vec{t},\vec{u})}{\Omega (x-b\rho_{1}^{2b})} \bigg\}dx + \int_{\sigma_{\star}}^{b\rho_{2}^{2b}} \bigg\{ \hat{f}(x;\vec{t},\vec{u}) - \frac{b \rho_{2}^{2b} \hat{\mathsf{T}}_{1}(b\rho_{2}^{2b};\vec{t},\vec{u})}{b\rho_{2}^{2b}-x} \bigg\}dx 
	\\ \nonumber
& + b\rho_{1}^{2b} \frac{\mathsf{T}_{1}(b\rho_{1}^{2b};\vec{t},\vec{u})}{\Omega} \ln \bigg( \frac{b\rho_{1}^{b}}{\sqrt{2\pi}(\sigma_{\star}-b\rho_{1}^{2b})} \bigg)
- b\rho_{2}^{2b} \hat{\mathsf{T}}_{1}(b\rho_{2}^{2b};\vec{t},\vec{u}) \ln \bigg( \frac{b\rho_{2}^{b}}{\sqrt{2\pi}(b\rho_{2}^{2b}-\sigma_{\star})} \bigg), 
	\\ \nonumber
 C_{4} = &\; \sqrt{2} \,\mathcal{I} b\bigg( \rho_1^{3b} \frac{ \mathsf{T}_{2}(b\rho_1^{2b};\vec{t},\vec{u}) }{ \Omega } - \rho_1^{b}\frac{\mathsf{T}_{1}(b\rho_1^{2b};\vec{t},\vec{u}) }{ \Omega } - \rho_1^{3b} \frac{ \mathsf{T}_{1}(b\rho_1^{2b};\vec{t},\vec{u})^{2} }{ \Omega^{2} } 
	\\\nonumber
&  - \rho_2^{3b} \hat{\mathsf{T}}_2(b \rho_2^{2b}; \vec{t}, \vec{u}) - \rho_2^{b}  \hat{\mathsf{T}}_{1}(b\rho_2^{2b};\vec{t},\vec{u})  - \rho_2^{3b} \hat{\mathsf{T}}_{1}(b\rho_2^{2b};\vec{t},\vec{u})^2  \bigg), 
	\\ \label{mathcalFndef}
\mathcal{F}_{n} = &\; \ln \frac{\theta( \sigma_{\star} n  + \frac{1}{2} - \alpha + \frac{\ln(\sigma_{2}/\sigma_{1}) + \ln{\mathsf{Q}(\vec{t}, \vec{u})} }{2 \ln(\rho_2/\rho_1)}; \frac{\pi i}{\ln(\rho_2/\rho_1)})}{ \theta( \sigma_{\star} n  + \frac{1}{2} - \alpha + \frac{\ln(\sigma_{2}/\sigma_{1})}{2 \ln(\rho_2/\rho_1)}; \frac{\pi i}{\ln(\rho_2/\rho_1)})},
\end{align}
and the constant $\mathcal{I} \in \mathbb{R}$ is given by
\begin{align}
& \mathcal{I} = \int_{-\infty}^{+\infty} \bigg\{ \frac{y\, e^{-y^{2}}}{\sqrt{\pi}\, \mathrm{erfc}(y)} - \chi_{(0,+\infty)}(y) \bigg[ y^{2}+\frac{1}{2} \bigg] \bigg\}dy \approx -0.81367. \label{def of I}
\end{align}
In particular, since $\mathbb{E}\big[ \prod_{j=1}^{2m} e^{u_{j}\mathrm{N}(r_{j})} \big]$ is analytic in $u_{1},\dots,u_{2m} \in \mathbb{C}$ and is positive for $u_{1},\dots,u_{2m} \in \mathbb{R}$, the asymptotic formula \eqref{asymp in main thm hard} together with Cauchy's formula shows that
\begin{align}\label{der of main result hard}
\partial_{u_{1}}^{k_{1}}\dots \partial_{u_{2m}}^{k_{2m}} \bigg\{ \ln \mathbb{E}\bigg[ \prod_{j=1}^{2m} e^{u_{j}\mathrm{N}(r_{j})} \bigg] - \bigg( C_{1} n + C_{2} \ln n + C_{3} +  \frac{C_{4}}{\sqrt{n}} \bigg) \bigg\} = \bigO\big(n^{-\frac{3}{5}}\big), \quad \mbox{as } n \to + \infty,
\end{align}
for any $k_{1},\dots,k_{2m}\in \mathbb{N}$, and $u_{1},\dots,u_{2m}\in \mathbb{R}$. 
\end{theorem}

\begin{remark}
It is well-known that the theta function is a universal object of one-dimensional point processes in the multi-cut regime, see e.g. \cite{DIZ1997, DKMVZ1999, BDE2000, G2006, Shcherbina, ClaeysGravaMcLaughlin, BG2, BCL1, BCL2, FK2020, BCL3, KM2021, CFWW2021}. In dimension two, the emergence of this function was conjectured in \cite[Section 1.5]{LS2018} and proved in \cite{Charlier 2d gap} in the context of large gap problems (or equivalently, partition function asymptotics with hard edges). This function also appears in \cite{ACC2022} in the study of microscopic correlations and smooth macroscopic statistics of two-dimensional rotation-invariant ensembles with soft edges. To our knowledge, Theorem \ref{thm:main thm hard} is the first result on counting statistics of a two-dimensional point process involving the $\theta$-function.
\end{remark}
\begin{remark}(Periodicity of $\mathcal{F}_{n}$.)
In the multi-cut regime of one-dimensional point processes, asymptotic formulas are, in general, only quasiperiodic in $n$, see e.g. \cite{Widom1995, DKMVZ1999, Shcherbina, BG2}. However, in the special case where the mass of the equilibrium measure on each interval of the support is a rational number, these asymptotic formulas become periodic, see e.g. \cite{Marchal} and \cite[Corollary 2.2]{CFWW2021}. 

Interestingly, Theorem \ref{thm:main thm hard} shows that an analogous phenomenon holds in our two-dimensional setting. To see this, recall from \eqref{theta prop} that $\theta$ is periodic of period $1$. Since $n$ runs over the integers, the function $n \mapsto \mathcal{F}_{n}$ is, in general, only quasiperiodic in $n$. However, it follows from \eqref{mathcalFndef} and \eqref{interpretation of sigma star} that if the mass $\sigma_{\star}=\int_{|z| \leq \rho_{1}} \mu_{h}(d^{2}z)$ is rational, then $n \mapsto \mathcal{F}_{n}$ is periodic in $n$. 

More generally, asymptotic formulas related to a given two-dimensional point process are expected to be periodic in $n$ whenever the masses of the components of the equilibrium measure are all rational. This holds true in the setting of this paper, as well as in the two-dimensional soft edge setting of \cite{ACC2022}. Recent works \cite{DeanoSimm, BY2022} on the so-called lemniscate ensemble also support this belief. This model has a $d$-fold rotational symmetry, where $d$ is a parameter that determines the number of components. The mass of the equilibrium measure on each component is $1/d$, and therefore one expects asymptotic formulas in this model to be periodic in $n$ of period $d$. 
The formulas in \cite{DeanoSimm, BY2022} are consistent with this expectation: these formulas involve $n=Nd$ points and are not oscillatory as $N \to +\infty$.
\end{remark}

\begin{remark}(Probabilistic interpretation of $\mathcal{F}_{n}$.)
Using the well-known relation (see \cite[eq 21.5.8]{NIST})
\begin{align*}
\theta(\tau^{-1}z ; -\tau^{-1}) = e^{\pi i z^{2} \tau^{-1}} \sqrt{-i\tau}\theta(z;\tau),
\end{align*}
we can rewrite $\mathcal{F}_{n}$ as
\begin{align}\label{e Fn prob}
e^{\mathcal{F}_{n}} = \exp \bigg( -\frac{(\ln \mathsf{Q}(\vec{t}, \vec{u}))^{2}}{4 \ln (\frac{\rho_{2}}{\rho_{1}})} \bigg) \frac{\sum_{\ell \in \mathbb{Z}} (\frac{\rho_{1}}{\rho_{2}})^{(\ell-\langle\Lambda_{n}\rangle)^{2}}\mathsf{Q}(\vec{t}, \vec{u})^{\ell-\langle\Lambda_{n}\rangle}}{\sum_{\ell \in \mathbb{Z}} (\frac{\rho_{1}}{\rho_{2}})^{(\ell-\langle\Lambda_{n}\rangle)^{2}}},
\end{align}
where $\langle\Lambda_{n}\rangle := \Lambda_{n}-\lfloor \Lambda_{n} \rfloor$ is the fractional part of $\Lambda_{n}$ and
\begin{align*}
\Lambda_{n} := \sigma_{\star}n -\frac{1}{2}-\alpha + \frac{\ln(\sigma_{2}/\sigma_{1})}{2 \ln(\rho_{2}/\rho_{1})}.
\end{align*}
The right-most fraction in \eqref{e Fn prob} is related to the random variable $v_{n}(\vec{t}, \vec{u}):=(\mathcal{X}_{n}-\langle\Lambda_{n}\rangle) \ln \mathsf{Q}(\vec{t}, \vec{u})$ via
\begin{align}\label{Evn}
\frac{\sum_{\ell \in \mathbb{Z}} (\frac{\rho_{1}}{\rho_{2}})^{(\ell-\langle\Lambda_{n}\rangle)^{2}}\mathsf{Q}(\vec{t}, \vec{u})^{\ell-\langle\Lambda_{n}\rangle}}{\sum_{\ell \in \mathbb{Z}} (\frac{\rho_{1}}{\rho_{2}})^{(\ell-\langle\Lambda_{n}\rangle)^{2}}} = \mathbb{E}[e^{v_{n}(\vec{t}, \vec{u})}],
\end{align}
where $\mathcal{X}_{n}$ is the discrete Gaussian random variable on $\mathbb{Z}$ defined by
\begin{align}\label{def of mathcalXn}
\mathbb{P}(\mathcal{X}_{n}=x) = \frac{(\frac{\rho_{1}}{\rho_{2}})^{(x-\langle\Lambda_{n}\rangle)^{2}}}{\sum_{\ell \in \mathbb{Z}} (\frac{\rho_{1}}{\rho_{2}})^{(\ell-\langle\Lambda_{n}\rangle)^{2}}}, \qquad x \in \mathbb{Z}.
\end{align}
In the multi-cut regime of one-dimensional point processes, fluctuation formulas analogous to \eqref{e Fn prob} involving an oscillatory discrete Gaussian can be found in \cite[Section 8]{BG2} and \cite{CFWW2021}.
\end{remark}
The following corollary gives a probabilistic interpretation of $\mathcal{X}_{n}$ in terms of counting statistics and is analogous to the earlier result \cite[Corollary 1.4]{CFWW2021} obtained in dimension one.

\begin{corollary}\label{coro:prob interpretation of Xn}
Let $x \in \mathbb{Z}$ be fixed. As $n \to + \infty$,
\begin{align}\label{discrete gaussian in the limit}
\mathbb{P}\big(\mathrm{N}(\rho_{1})=\lfloor \Lambda_{n} \rfloor +x\big) = \mathbb{P}(\mathcal{X}_{n}=x) +o(1).
\end{align}
\end{corollary}
\begin{proof}
The proof is inspired by the proof of \cite[Corollary 1.4]{CFWW2021}. Using Theorem \ref{thm:main thm hard} with $m=1$, $u_{2}=0$ and $t_{1}=0=t_{2}$, we have $C_{1}=u_{1}\sigma_{\star}$, $C_{2}=0$, $\ln{\mathsf{Q}(\vec{0}, \vec{u})}=u_{1}$, $C_{3} = -\frac{u_{1}}{2}-u_{1}(\alpha
- \frac{2\ln(\sigma_{2}/\sigma_{1}) + u_{1}}{4 \ln(\rho_{2}/\rho_{1})})$, $C_{4}=0$, and thus
\begin{align}\label{Eeu1rho1}
\mathbb{E}[e^{u_{1}\mathrm{N}(\rho_{1})}] = e^{u_{1}\Lambda_{n}}\mathbb{E}[e^{u_{1}(\mathcal{X}_{n}-\langle\Lambda_{n}\rangle)}] \big( 1+\bigO(n^{-\frac{3}{5}}) \big), \qquad \mbox{as } n \to + \infty,
\end{align}
uniformly for $u_{1}\in K$, where $K \subset \mathbb{R}$ is a compact subset containing $0$. The rest of the proof proceeds by contradiction. Fix $x \in \mathbb{Z}$ and suppose that \eqref{discrete gaussian in the limit} does not hold. Then there exists a sequence $\{n_{k}\}_{k=1}^{+\infty} \subset \mathbb{N}$ such that $\mathbb{P}\big(\mathrm{N}(\rho_{1})=\lfloor \Lambda_{n_{k}} \rfloor +x\big)$ remains bounded away from $\mathbb{P}(\mathcal{X}_{n_{k}}=x)$ for all $k$. Since $\langle\Lambda_{n_{k}}\rangle \in [0,1)$ for all $k$, there exists a subsequence $\{n_{k_{j}}\}_{j=1}^{+\infty}$ such that $\langle\Lambda_{n_{k_{j}}}\rangle \to y_{*} \in [0,1]$ as $j \to + \infty$. Let $\mathcal{X}_{*}$ be the random variable defined as in \eqref{def of mathcalXn} but with $\langle\Lambda_{n}\rangle$ replaced by $y_{*}$. Then \eqref{Eeu1rho1} together with the fact that 
\begin{align*}
\mathbb{E}[e^{u_{1}\mathcal{X}_{n_{k_{j}}}}] = \mathbb{E}[e^{u_{1}\mathcal{X}_{*}}]+o(1), \qquad \mbox{as } j \to + \infty \mbox{ uniformly for } u_{1} \in K
\end{align*}
implies that
\begin{align*}
\mathbb{E}[e^{u_{1}(\mathrm{N}(\rho_{1})-\lfloor\Lambda_{n_{k_{j}}}\rfloor)}] = \mathbb{E}[e^{u_{1}\mathcal{X}_{*}}]+o(1), \qquad \mbox{as } j\to + \infty \mbox{ uniformly for } u_{1} \in K.
\end{align*}
By \cite[top of page 415]{Bill}, this implies that $\mathrm{N}(\rho_{1})-\lfloor \Lambda_{n_{k_{j}}} \rfloor$ converges in distribution to $\mathcal{X}_{*}$ as $j \to + \infty$. Since $\mathbb{P}(\mathcal{X}_{n_{k_{j}}}=x) = \mathbb{P}(\mathcal{X}_{*}=x)+o(1)$ as $j \to + \infty$, we thus have $\mathbb{P}\big(\mathrm{N}(\rho_{1})=\lfloor \Lambda_{n_{k_{j}}} \rfloor +x\big) = \mathbb{P}(\mathcal{X}_{n_{k_{j}}}=x)+o(1)$ as $j \to + \infty$. We have obtained our contradiction.
\end{proof}

 For $\vec{j} \in (\mathbb{N}^{2m})_{>0} := \{\vec{j}=(j_{1},\dots,j_{2m}) \in \mathbb{N}: j_{1}+\dots+j_{2m}\geq 1\}$, the joint cumulant $\kappa_{\vec{j}}=\kappa_{\vec{j}}(r_{1},\dots,r_{2m};n,b,\alpha)$ of $\mathrm{N}(r_{1}), \dots, \mathrm{N}(r_{2m})$  is defined by
\begin{align}\label{joint cumulant}
\kappa_{\vec{j}}=\kappa_{j_{1},\dots,j_{2m}}:=\partial_{\vec{u}}^{\vec{j}} \ln \mathbb{E}[e^{u_{1}\mathrm{N}(r_{1})+\dots + u_{2m}\mathrm{N}(r_{2m})}] \Big|_{\vec{u}=\vec{0}},
\end{align}
where $\partial_{\vec{u}}^{\vec{j}}:=\partial_{u_{1}}^{j_{1}}\dots \partial_{u_{2m}}^{j_{2m}}$. As an immediate corollary of Theorem \ref{thm:main thm hard}, we obtain the large $n$ behavior of any cumulant $\kappa_{\vec{j}}$.

\begin{corollary}[Asymptotics of cumulants]\label{cumulantscor}
Let $m \in \mathbb{N}_{>0}$, $b>0$, $0 < \rho_1 < \rho_2 < b^{-\frac{1}{2b}}$, $\vec{j} \in (\mathbb{N}^{2m})_{>0}$, $\alpha > -1$, $t_{1}>\dots>t_{m} \geq 0$ and $0 \leq t_{m+1} < \dots < t_{2m}$ be fixed. For $n \in \mathbb{N}_{>0}$, define $\{r_\ell\}_{\ell =1}^{2m}$ by \eqref{rellhardedge}. As $n \to +\infty$, the joint cumulant $\kappa_{\vec{j}}$ satisfies
\begin{align}\label{asymp cumulant hard edge}
\kappa_{\vec{j}} = \partial_{\vec{u}}^{\vec{j}}C_{1}\big|_{\vec{u}=\vec{0}} \; n + \partial_{\vec{u}}^{\vec{j}}C_{2}\big|_{\vec{u}=\vec{0}}  \;\ln{n} + \partial_{\vec{u}}^{\vec{j}}C_{3}\big|_{\vec{u}=\vec{0}} + \partial_{\vec{u}}^{\vec{j}}\mathcal{F}_{n}\big|_{\vec{u}=\vec{0}}  +  \frac{\partial_{\vec{u}}^{\vec{j}}C_{4}\big|_{\vec{u}=\vec{0}}}{\sqrt{n}} + \bigO\big(n^{-\frac{3}{5}}\big),
\end{align}
where $C_{1},\dots,C_{4}, \mathcal{F}_n$ are as in Theorem \ref{thm:main thm hard}. 
\end{corollary}

Since $\mathbb{E}[\mathrm{N}(r)] = \kappa_{1}(r)$, we obtain the following result after setting $\vec{j} = 1$ in Corollary \ref{cumulantscor} and performing some long but straightforward calculations. 
 
\begin{corollary}[{Asymptotics of $\mathbb{E}[\mathrm{N}(r_{\ell})]$)}]
Under the assumptions of Corollary \ref{cumulantscor}, the expectation value $\mathbb{E}[\mathrm{N}(r_{\ell})]$ obeys the following formula for any $1 \leq \ell \leq 2m$:
\begin{align*}
& \mathbb{E}[\mathrm{N}(r_{\ell})] = b_1(t_\ell) n + c_{1}(t_\ell) \ln{n} + d_1(t_\ell) + \mathsf{f}_{1}(n, t_\ell) + e_1(t_\ell) n^{-\frac{1}{2}} + \bigO\big(n^{-\frac{3}{5}}\big)
\end{align*}
as $n \to + \infty$, where
\begin{align}\nonumber
 b_1(t_\ell) =&\; \begin{cases}  b\rho_1^{2b} + b\frac{1 - e^{-\frac{t_{\ell}}{b}(\sigma_{\star} - b\rho_1^{2b})} }{t_{\ell}}, & \ell = 1, \dots, m \mbox{ and } t_{\ell}>0, \\ 
 b\rho_2^{2b} - b\frac{1 - e^{-\frac{t_{\ell}}{b}(b\rho_2^{2b} - \sigma_{\star})} }{t_{\ell}}, & \ell = m+1, \dots, 2m \mbox{ and } t_{\ell}>0, \\
 \sigma_{\star}, & \ell = m,m+1 \mbox{ and } t_{\ell}=0,
  \end{cases} 
	\\\nonumber
 c_{1}(t_\ell) = &\begin{cases}   - \frac{b\rho_1^{2b} t_\ell}{2}, & \ell = 1, \dots, m, \\ 
  \frac{b\rho_2^{2b} t_\ell}{2}, & \ell = m+1, \dots, 2m,
  \end{cases}
	\\ \nonumber
d_1(t_\ell) = & -\frac{1}{2} - e^{-\frac{t_{\ell}}{b} (\sigma_{\star} -b\rho_1^{2 b})}\bigg(\alpha - \frac{\ln(\sigma_{2}/\sigma_{1})}{2 \ln(\rho_{2}/\rho_{1})}\bigg) + b \rho_1^{2b} t_\ell \ln \bigg( \frac{b \rho_1^b}{\sqrt{2\pi}(\sigma_{\star} - b \rho_1^{2b})}\bigg)	
	\\\nonumber
& + \int_{b \rho_1^{2b}}^{\sigma_{\star}} \bigg(\frac{b \rho_1^{2b} t_\ell (1 - e^{-\frac{t_{\ell}}{b} (x-b\rho_1^{2 b})})}{x - b\rho_1^{2b}} - e^{-\frac{t_{\ell}}{b} (x-b\rho_1^{2 b})} \frac{t_\ell (2\alpha + x t_\ell)}{2b}\bigg) dx, \qquad \ell = 1, \dots, m, 
	\\ \nonumber
d_1(t_\ell) = & -\frac{1}{2} - e^{-\frac{t_{\ell}}{b} (b\rho_2^{2 b}-\sigma_{\star})}\bigg(\alpha - \frac{\ln(\sigma_{2}/\sigma_{1})}{2 \ln(\rho_{2}/\rho_{1})}\bigg) - b \rho_2^{2b} t_\ell \ln \bigg( \frac{b \rho_2^b}{\sqrt{2\pi}(b \rho_2^{2b} - \sigma_{\star})}\bigg)	
	\\\nonumber
& - \int_{\sigma_{\star}}^{b \rho_2^{2b}} \bigg(\frac{b \rho_2^{2b} t_\ell (1 - e^{-\frac{t_{\ell}}{b} (b\rho_2^{2 b}-x)})}{b\rho_2^{2b} - x} + e^{-\frac{t_{\ell}}{b} (b\rho_2^{2 b}-x)} \frac{t_\ell (2\alpha -xt_{\ell})}{2b}\bigg) dx, \qquad \ell = m+1, \dots, 2m,
	\\ \nonumber
\mathsf{f}_{1}(n, t_\ell) = &\; \begin{cases}
\frac{e^{-\frac{t_{\ell}}{b} (\sigma_{\star}-b\rho_1^{2 b})}}{2 \ln(\rho_{2}/\rho_{1})} (\ln \theta)'( n \sigma_{\star}  + \frac{1}{2} - \alpha+\frac{\ln(\sigma_{2}/\sigma_{1}) }{2 \ln(\rho_2/\rho_1)}; \frac{\pi i}{\ln(\rho_2/\rho_1)}), & \ell = 1, \dots, m,
	\\ \nonumber
\frac{e^{-\frac{t_{\ell}}{b} (b\rho_2^{2 b}- \sigma_{\star})}}{2 \ln(\rho_{2}/\rho_{1})} (\ln \theta)'(n \sigma_{\star}  + \frac{1}{2} - \alpha+\frac{\ln(\sigma_{2}/\sigma_{1}) }{2 \ln(\rho_2/\rho_1)}; \frac{\pi i}{\ln(\rho_2/\rho_1)}), & \ell = m+1, \dots, 2m,
\end{cases}
	\\ \nonumber
 e_1(t_\ell) =&\; \begin{cases}  \sqrt{2} \, \mathcal{I} b \rho_1^{b} t_{\ell} \big( \rho_1^{2b}t_{\ell} - 1 \big), & \ell = 1, \dots, m, \\ 
-\sqrt{2} \, \mathcal{I} b \rho_2^{b} t_{\ell} \big( \rho_2^{2b}t_{\ell} +1 \big), & \ell = m+1, \dots, 2m.
  \end{cases}
\end{align}
In particular, taking $m=1$ and $t_{1}=0$, as $n \to + \infty$ we infer that
\begin{align}\label{asymptotics E in the first droplet}
& \mathbb{E}[N(\rho_{1})] = \Lambda_{n} + \frac{(\ln \theta)'( n \sigma_{\star}  + \frac{1}{2} - \alpha+\frac{\ln(\sigma_{2}/\sigma_{1}) }{2 \ln(\rho_2/\rho_1)}; \frac{\pi i}{\ln(\rho_2/\rho_1)})}{2 \ln(\rho_{2}/\rho_{1})} + \bigO(n^{-\frac{3}{5}}) \nonumber \\
& = \sigma_{\star}n -\frac{1}{2}-\alpha + \frac{\ln(\sigma_{2}/\sigma_{1})}{2 \ln(\rho_{2}/\rho_{1})} + \frac{(\ln \theta)'( n \sigma_{\star}  + \frac{1}{2} - \alpha+\frac{\ln(\sigma_{2}/\sigma_{1}) }{2 \ln(\rho_2/\rho_1)}; \frac{\pi i}{\ln(\rho_2/\rho_1)})}{2 \ln(\rho_{2}/\rho_{1})} + \bigO(n^{-\frac{3}{5}}).
\end{align}
\end{corollary}

Using that
\begin{align*}
 \mbox{Cov}[\mathrm{N}(r_{\ell}),\mathrm{N}(r_{k})] = \kappa_{(1,1)}(r_{\ell},r_{k}), \qquad \mbox{Var}[\mathrm{N}(r_\ell)] = \mbox{Cov}[\mathrm{N}(r_{\ell}),\mathrm{N}(r_{\ell})],
\end{align*}
we can also obtain asymptotic formulas for the covariance and variance of the disk counting function by setting $\vec{j} = (1,1)$ in Corollary \ref{cumulantscor}. We distinguish three cases depending the values of $\ell$ and $k$.
\begin{enumerate}
\item If $1 \leq \ell \leq k \leq m$, then $\mathrm{N}(r_{\ell})$ and $\mathrm{N}(r_{k})$ count the number of points in disks whose radii $r_{\ell}$ and $r_k$ are close to $\rho_1$.

\item If $m+1 \leq \ell \leq k \leq 2m$, then $\mathrm{N}(r_{\ell})$ and $\mathrm{N}(r_{k})$ count the number of points in disks whose radii $r_{\ell}$ and $r_k$ are close to $\rho_2$.

\item If $1 \leq \ell \leq m$ and $m+1\leq k \leq 2m$, then $\mathrm{N}(r_{\ell})$ counts the number of points in a disk of radius $r_\ell \approx \rho_1$, whereas $\mathrm{N}(r_{k})$ counts the number of points in a disk of radius $r_k \approx \rho_2$. 
\end{enumerate}

The formulas for the covariance are naturally expressed in terms of the Weierstrass $\wp$-function. Indeed, these formulas involve the second derivative of $\ln \theta$. By \cite[formula 23.6.14]{NIST} ($\theta_{1}(\pi z)$ and $\theta_{3}(\pi z)$ in \cite{NIST} correspond here to $\theta_{1}(z)$ and $\theta(z)$, respectively)
$$\wp(z; \tau) = \frac{\theta_1'''(0; \tau)}{3\theta_1'(0; \tau)} - \frac{d^2}{dz^2}\ln \theta_1(z; \tau), \quad \text{where} \quad \theta_1(z; \tau) = e^{i\pi z + \frac{\pi i \tau}{4} - \frac{\pi i}{2}} \theta\Big(z + \frac{1 + \tau}{2}; \tau\Big),$$
so we can express this second derivative in terms of the Weierstrass $\wp$-function as
$$-\frac{d^2}{dz^2}\ln \theta(z; \tau) = \wp\Big(z  - \frac{1+\tau}{2} ; \tau\Big) - c,$$
where $c \in \C$ is defined by
\begin{align}\label{cdef}
c := \frac{\theta_1'''(0; \tau)}{3 \theta_1'(0; \tau)}.
\end{align}
The function $z\mapsto \wp(z;\tau)$ is doubly periodic in the complex plane, of periods $1$ and $\tau$.
In the following corollaries, we let $c = c(\rho_1, \rho_2)$ be given by (\ref{cdef}) with $\tau$ defined by
$$\tau := \frac{\pi i}{\ln(\rho_2/\rho_1)}.$$

\begin{corollary}[{Asymptotics of the covariance for $1 \leq \ell \leq k \leq m$}]\label{coro:cov case 1}
Under the assumptions of Corollary \ref{cumulantscor}, the covariance $\mathrm{Cov}(\mathrm{N}(r_{\ell}),\mathrm{N}(r_{k}))$ obeys the following formula for any $1 \leq \ell \leq k \leq m$:
\begin{align*}
\mathrm{Cov}(\mathrm{N}(r_{\ell}),\mathrm{N}(r_{k})) = & \; b_{(1,1)}(t_{\ell},t_{k})n + c_{(1,1)}(t_{\ell},t_{k})\ln{n} + d_{(1,1)}(t_{\ell},t_{k}) \\
& + \mathsf{f}_{(1,1)}(n, t_\ell, t_k) + e_{(1,1)}(t_{\ell},t_{k})n^{-\frac{1}{2}} + \bigO\big(n^{-\frac{3}{5}}\big) \nonumber
\end{align*}
as $n \to +\infty$, where 
\begin{align}\label{def of b11 hard edge}
& b_{(1,1)}(t_{\ell},t_{k}) =  b\frac{1 - e^{-\frac{t_{\ell}}{b}(\sigma_{\star} - b\rho_1^{2b}) } }{t_{\ell}} - b\frac{1 - e^{-\frac{t_{\ell} + t_k}{b}(\sigma_{\star} - b\rho_1^{2b})} }{t_{\ell} + t_k} \quad \mbox{if } t_{k}>0, \\
& b_{(1,1)}(t_{\ell},0) = 0, \qquad c_{(1,1)}(t_{\ell},t_{k}) = \frac{b\rho_1^{2b} t_k}{2}, \nonumber
	 \\ \nonumber
& d_{(1,1)}(t_{\ell},t_{k}) = 
 -e^{-\frac{t_{\ell}}{b}(\sigma_{\star} - b\rho_1^{2b}) } \bigg(\alpha
   - \frac{\ln (\sigma_{2}/\sigma_{1})}{2 \ln (\rho_{2}/\rho_{1})}\bigg)
   +e^{-\frac{t_\ell+t_k}{b}(\sigma_{\star} - b \rho_1^{2 b})}  \bigg(\alpha
   -\frac{\ln (\sigma_{2}/\sigma_{1}) - 1}{2 \ln (\rho_{2}/\rho_{1})}\bigg) 
  	\\ \nonumber
&   +
   \int_{b\rho_1^{2b}}^{\sigma_{\star}}
    \bigg\{- b \rho_{1}^{2b} t_{k} \frac{1-e^{-\frac{t_{\ell}+t_{k}}{b}(x-b\rho_{1}^{2b})}}{x-b\rho_{1}^{2b}}  + \frac{x t_{k}^{2}+2\alpha t_{k}}{2b} e^{-\frac{t_{\ell}+t_{k}}{b}(x-b\rho_{1}^{2b})}       
     	\\ \nonumber
& -t_\ell (e^{-\frac{t_{\ell}}{b}(x-b\rho_{1}^{2b})}-e^{-\frac{t_{\ell}+t_{k}}{b}(x-b\rho_{1}^{2b})}) \bigg( \frac{b\rho_{1}^{2b}}{x-b\rho_{1}^{2b}} + \frac{2\alpha + x t_{\ell}}{2b} \bigg) \bigg\}  dx -b  \rho_1^{2 b} t_k \ln\bigg(\frac{b \rho_1^b}{\sqrt{2 \pi } (\sigma_{\star}-b \rho_1^{2 b})}\bigg),
	\\ \nonumber
& \mathsf{f}_{(1,1)}(n, t_\ell, t_k) = -\frac{e^{-\frac{t_\ell+t_k}{b} (\sigma_{\star} - b \rho_1^{2 b})}}{4\ln(\rho_{2}/\rho_{1})^2}\bigg\{\wp\bigg( n \sigma_{\star}  + \frac{\tau}{2} - \alpha+\frac{\ln(\sigma_{2}/\sigma_{1}) }{2 \ln(\rho_2/\rho_1)}; \tau\bigg) - c
	\\\nonumber
& - 2(e^{\frac{t_k}{b}(\sigma_{\star} - b\rho_1^{2b})} -1) \ln(\rho_{2}/\rho_{1}) (\ln \theta)'\bigg(n \sigma_{\star}  + \frac{1}{2} - \alpha+\frac{\ln(\sigma_{2}/\sigma_{1}) }{2 \ln(\rho_2/\rho_1)}; \tau\bigg) \bigg\},
	\\\nonumber	
& e_{(1,1)}(t_{\ell},t_{k}) = \sqrt{2} \, \mathcal{I} b \rho_1^{b} t_{k} \big( 1-\rho_1^{2b}(2t_{\ell}+t_{k}) \big).
\end{align}
In particular, with $m=1$, $\ell=k=1$ and $t_{1}=0$, as $n \to + \infty$ we obtain
\begin{align}\label{asymptotics Var in the first droplet}
& \mathrm{Var}[N(\rho_{1})] = \frac{1}{2 \ln (\rho_{2}/\rho_{1})} - \frac{\wp\big( n \sigma_{\star}  + \frac{\tau}{2} - \alpha+\frac{\ln(\sigma_{2}/\sigma_{1}) }{2 \ln(\rho_2/\rho_1)}; \tau\big) - c}{4\ln(\rho_{2}/\rho_{1})^2} + \bigO(n^{-\frac{3}{5}}).
\end{align}
\end{corollary}

\begin{remark}
It should be emphasized that the case $t_{k}>0$ of Corollary \ref{coro:cov case 1} above is drastically different from the case $t_{k}=0$. Indeed, if $t_{k}>0$, then $\mathrm{Cov}(\mathrm{N}(r_{\ell}),\mathrm{N}(r_{k}))$ is of order $n$, while if $t_{k}=0$, then $\mathrm{Cov}(\mathrm{N}(r_{\ell}),\mathrm{N}(r_{k}))=\mathrm{Cov}(\mathrm{N}(r_{\ell}),\mathrm{N}(\rho_{1}))$ is of order $1$. 
\end{remark}

\begin{corollary}[{Asymptotics of the covariance for $m+1 \leq \ell \leq k \leq 2m$}]
Under the assumptions of Corollary \ref{cumulantscor}, the covariance $\mathrm{Cov}(\mathrm{N}(r_{\ell}),\mathrm{N}(r_{k}))$ obeys the following formula for any $m+1 \leq \ell \leq k \leq 2m$:
\begin{align*}
\mathrm{Cov}(\mathrm{N}(r_{\ell}),\mathrm{N}(r_{k})) = & \; b_{(1,1)}(t_{\ell},t_{k})n + c_{(1,1)}(t_{\ell},t_{k})\ln{n} + d_{(1,1)}(t_{\ell},t_{k}) \\
& + \mathsf{f}_{(1,1)}(n, t_\ell, t_k) + e_{(1,1)}(t_{\ell},t_{k})n^{-\frac{1}{2}} + \bigO\big(n^{-\frac{3}{5}}\big) \nonumber
\end{align*}
as $n \to +\infty$, where 
\begin{align}\label{def of b11 hard edge 2}
& b_{(1,1)}(t_{\ell},t_{k}) = b\frac{1 - e^{-\frac{t_{k}}{b}(b\rho_2^{2b} - \sigma_{\star}) } }{t_{k}} - b\frac{1 - e^{-\frac{t_{\ell} + t_k}{b}(b\rho_2^{2b} - \sigma_{\star}) } }{t_{\ell} + t_k} \quad \mbox{if } t_{\ell}>0,
\\
& b_{(1,1)}(0,t_{k}) =0, \qquad c_{(1,1)}(t_{\ell},t_{k}) = \frac{b\rho_2^{2b} t_\ell}{2}, \nonumber
	 \\\nonumber
& d_{(1,1)}(t_{\ell},t_{k}) = e^{-\frac{t_{k}}{b}(b\rho_2^{2b} - \sigma_{\star}) } \bigg(\alpha
   - \frac{\ln (\sigma_{2}/\sigma_{1})}{2 \ln (\rho_{2}/\rho_{1})}\bigg)
   - e^{-\frac{t_\ell+t_k}{b}(b \rho_2^{2 b} - \sigma_{\star})} \bigg(\alpha
   - \frac{\ln (\sigma_{2}/\sigma_{1}) + 1}{2 \ln (\rho_{2}/\rho_{1})}\bigg) 
  	\\\nonumber
&   +
   \int_{\sigma_{\star}}^{b\rho_2^{2b}}
   \bigg\{- b \rho_{2}^{2b} t_{\ell} \frac{1-e^{-\frac{t_{\ell}+t_{k}}{b}(b\rho_{2}^{2b}-x)}}{b\rho_{2}^{2b}-x}  + \frac{x t_{\ell}^{2}-2\alpha t_{\ell}}{2b} e^{-\frac{t_{\ell}+t_{k}}{b}(b\rho_{2}^{2b}-x)}     
     	\\\nonumber
& +t_k (e^{-\frac{t_{k}}{b}(b\rho_{2}^{2b}-x)}-e^{-\frac{t_{\ell}+t_{k}}{b}(b\rho_{2}^{2b}-x)}) \bigg( \frac{-b\rho_{2}^{2b}}{b\rho_{2}^{2b}-x} + \frac{2\alpha - x t_{k}}{2b} \bigg) \bigg\}  dx
  -b  \rho_{2}^{2 b} t_{\ell} \ln\bigg(\frac{b \rho_{2}^b}{\sqrt{2 \pi } (b \rho_2^{2 b}-\sigma_{\star})}\bigg),
	\\\nonumber
& \mathsf{f}_{(1,1)}(n, t_\ell, t_k) = -\frac{e^{-\frac{t_\ell+t_k}{b} (b \rho_2^{2 b} - \sigma_{\star} )}}{4\ln(\rho_{2}/\rho_{1})^2}\bigg\{\wp\bigg( n \sigma_{\star}  + \frac{\tau}{2} - \alpha+\frac{\ln(\sigma_{2}/\sigma_{1}) }{2 \ln(\rho_2/\rho_1)}; \tau\bigg) - c
	\\\nonumber
& + 2(e^{\frac{t_\ell}{b}(b \rho_2^{2 b} - \sigma_{\star} )} -1) \ln(\rho_{2}/\rho_{1}) (\ln \theta)'\bigg(n \sigma_{\star}  + \frac{1}{2} - \alpha+\frac{\ln(\sigma_{2}/\sigma_{1}) }{2 \ln(\rho_2/\rho_1)}; \tau\bigg) \bigg\},
	\\\nonumber
& e_{(1,1)}(t_{\ell},t_{k}) = - \sqrt{2} \, \mathcal{I} b \rho_2^{b} t_{\ell} \big( 1+\rho_2^{2b}(t_{\ell}+2t_{k}) \big).
\end{align}
\end{corollary}

\begin{corollary}[{Asymptotics of the covariance for $1 \leq \ell \leq m$ and $m+1 \leq k \leq 2m$}]
Under the assumptions of Corollary \ref{cumulantscor}, the covariance $\mathrm{Cov}(\mathrm{N}(r_{\ell}),\mathrm{N}(r_{k}))$ obeys the following formula for any $1 \leq \ell \leq m$ and $m+1 \leq k \leq 2m$:
\begin{align*}
& \mathrm{Cov}(\mathrm{N}(r_{\ell}),\mathrm{N}(r_{k})) = d_{(1,1)}(t_{\ell},t_{k}) + \mathsf{f}_{(1,1)}(n, t_\ell, t_k) + \bigO\big(n^{-\frac{3}{5}}\big) \nonumber
\end{align*}
as $n \to +\infty$, where 
\begin{align}\nonumber
& d_{(1,1)}(t_{\ell},t_{k}) = 
 \frac{e^{-\frac{t_\ell}{b}(\sigma_{\star} - b\rho_1^{2b})} e^{-\frac{t_k}{b}(b\rho_2^{2b} - \sigma_{\star})}}{2 \ln(\rho_{2}/\rho_{1})},
	\\\nonumber
& \mathsf{f}_{(1,1)}(n, t_\ell, t_k) = \frac{e^{-\frac{t_\ell}{b}(\sigma_{\star} - b\rho_1^{2b})}e^{-\frac{t_k}{b}(b \rho_2^{2 b} - \sigma_{\star})}}{4\ln(\rho_{2}/\rho_{1})^2}\bigg\{ c - \wp\bigg( n \sigma_{\star}  + \frac{\tau}{2} - \alpha+\frac{\ln(\sigma_{2}/\sigma_{1}) }{2 \ln(\rho_2/\rho_1)}; \tau\bigg) \bigg\}.
\end{align}
\end{corollary}

It is crucial for the next corollary that $t_{m},t_{m+1}>0$, so that $b_{(1,1)}(t_\ell,t_\ell)>0$ for all $\ell \in \{1,\ldots,2m\}$. In this case the asymptotic joint fluctuations of $\mathrm{N}(r_{1}),\ldots,\mathrm{N}(r_{2m})$ are of order $\sqrt{n}$ and are continuous Gaussians. This should be compared to Corollary \ref{coro:prob interpretation of Xn} (corresponding to the case $m=1$ and $t_{m}=0$), which shows that the asymptotic fluctuations of $\mathrm{N}(\rho_{1})$ are of order $1$ and are described by a discrete Gaussian.

\begin{corollary}
Suppose that the assumptions of Corollary \ref{cumulantscor} hold and that $t_{m},t_{m+1}>0$. As $n \to + \infty$, the random variable $(\mathcal{N}_{1},\dots,\mathcal{N}_{2m})$, where
\begin{align}
& \mathcal{N}_{\ell} := \frac{\mathrm{N}(r_{\ell})-(b_1(t_\ell) n+ c_{1}(t_\ell) \ln{n})}{\sqrt{b_{(1,1)}(t_\ell,t_\ell) n}}, \qquad \ell=1,\dots,2m,
\label{Nj hard edge}
\end{align}
convergences in distribution to a multivariate normal random variable of mean $(0,\dots,0)$ whose covariance matrix $\Sigma$ is given by
\begin{align*}
\Sigma_{\ell,k} = \Sigma_{k, \ell} = \begin{cases}
\frac{b_{(1,1)}(t_{\ell},t_{k})}{\sqrt{b_{(1,1)}(t_{\ell},t_{\ell})b_{(1,1)}(t_{k},t_{k})}}, & \text{$1 \leq \ell \leq k \leq m$ or $m+1 \leq \ell \leq k \leq 2m$}, \\
0, & \text{$1 \leq \ell \leq m$ and $m+1 \leq k \leq 2m$},
\end{cases}
\end{align*}
where $b_{(1,1)}(t_{\ell},t_{k})$ is given by \eqref{def of b11 hard edge} for $1 \leq \ell \leq k \leq m$ and by \eqref{def of b11 hard edge 2} for $m+1 \leq \ell \leq k \leq 2m$.
\end{corollary}
\begin{proof}
By L\'evy's continuity theorem, the assertion will follow if we can show that the characteristic function $\mathbb{E}[e^{i \sum_{\ell = 1}^{2m} v_\ell \mathcal{N}_\ell}]$ converges pointwise to $e^{-\frac{1}{2}\sum_{\ell, k=1}^{2m} v_\ell \Sigma_{\ell,k} v_k}$ for every $v_\ell \in \mathbb{R}^{2m}$ as $n \to +\infty$. Letting $u_\ell = \frac{i v_\ell }{\sqrt{b_{(1,1)}(t_\ell,t_\ell) n}}$, (\ref{Nj hard edge}) and (\ref{asymp in main thm hard}) show that
\begin{align*}
\mathbb{E}[e^{i \sum_{\ell = 1}^{2m} v_\ell \mathcal{N}_\ell}]
& = \mathbb{E}[e^{\sum_{\ell = 1}^{2m} u_\ell \mathrm{N}(r_{\ell})}]
e^{- \sum_{\ell = 1}^{2m} u_\ell (b_1(t_\ell) n+ c_{1}(t_\ell) \ln{n})}
	\\
& = e^{C_{1}(\vec{u}) n + C_{2}(\vec{u}) \ln n + C_{3}(\vec{u}) + \mathcal{F}_n(\vec{u}) + \bigO(n^{-\frac{1}{2}})}
e^{- \sum_{\ell = 1}^{2m} u_\ell (\partial_{u_\ell} C_1|_{\vec{u}=\vec{0}} n + \partial_{u_\ell} C_2|_{\vec{u}=\vec{0}} \ln{n})}
\end{align*}
as $n \to +\infty$ for any fixed $v_\ell \in \mathbb{R}^{2m}$. Since $C_j|_{\vec{u}=\vec{0}} = 0$ for $j = 1,2,3$, $\mathcal{F}_n(\vec{u}) = \sum_{\ell=1}^{2m} \mathsf{f}_1(n, t_\ell) u_{\ell} + \bigO(|\vec{u}|^2) = \bigO(|\vec{u}|)$, and $|\vec{u}| = \bigO(n^{-1/2})$, we obtain
\begin{align*}
\mathbb{E}[&e^{i \sum_{\ell = 1}^{2m} v_\ell \mathcal{N}_\ell}]
 =  e^{\frac{1}{2}\sum_{\ell,k= 1}^{2m} u_\ell u_k \partial_{u_\ell}\partial_{u_k} C_1|_{\vec{u}=\vec{0}} n
 + \bigO(|\vec{u}|^3 n + |\vec{u}|^2 \ln{n} + |\vec{u}| + n^{-1/2})}
	\\
& =  e^{\frac{1}{2}\sum_{\ell,k = 1}^m \frac{iv_\ell}{\sqrt{b_{(1,1)}(t_\ell,t_\ell)}} \frac{iv_k}{\sqrt{b_{(1,1)}(t_k,t_k)}} b_{(1,1)}(t_{\min(\ell,k)}, t_{\max(\ell,k)}) }
	\\
&\times 
e^{\frac{1}{2}\sum_{\ell,k = m+1}^{2m} \frac{iv_\ell}{\sqrt{b_{(1,1)}(t_\ell,t_\ell)}} \frac{iv_k}{\sqrt{b_{(1,1)}(t_k,t_k)}} b_{(1,1)}(t_{\min(\ell,k)}, t_{\max(\ell,k)}) + \bigO( n^{-1/2})}
\to  e^{-\frac{1}{2}\sum_{\ell, k=1}^{2m} v_\ell \Sigma_{\ell,k} v_k}
\end{align*}
as $n \to +\infty$, which completes the proof.
\end{proof}

\medskip \textbf{Outline of proof.} Using that $\prod_{1 \leq j < k \leq n} |z_{k} -z_{j}|^{2}$ is the product of two Vandermonde determinants, we obtain after standard manipulations that
\begin{align}
\mathcal{E}_{n} & := \mathbb{E}\bigg[ \prod_{\ell=1}^{2m} e^{u_{\ell}\mathrm{N}(r_{\ell})} \bigg] = \frac{1}{n!Z_{n}} \int_{\mathbb{C}}\dots \int_{\mathbb{C}} \prod_{1 \leq j < k \leq n} |z_{k} -z_{j}|^{2} \prod_{j=1}^{n}w(z_{j}) d^{2}z_{j} \label{partition function} 
	\\
& = \frac{1}{Z_{n}} \det \left( \int_{\mathbb{C}} z^{j} \overline{z}^{k} w(z) d^{2}z \right)_{j,k=0}^{n-1} \label{def of Dn as n fold integral} \\
& = \frac{1}{Z_{n}}(2\pi)^{n}\prod_{j=0}^{n-1}\bigg(\int_{0}^{\rho_{1}}+\int_{\rho_{2}}^{+\infty}\bigg)u^{2j+1}w(u)du, \label{simplified determinant}
\end{align}
where the weight $w$ is defined by
\begin{align}\label{def of w and omega}
w(z):= e^{-n Q(z)} \omega(|z|), \qquad \omega(x) := \prod_{\ell=1}^{2m} \begin{cases}
e^{u_{\ell}}, & \mbox{if } x < r_{\ell}, \\
1, & \mbox{if } x \geq r_{\ell}.
\end{cases}
\end{align}
Formula \eqref{simplified determinant} directly follows from \eqref{def of Dn as n fold integral} and the fact that $w$ is rotation-invariant. Indeed, since $w(z)=w(|z|)$, the integral $\int_{\mathbb{C}} z^{j} \overline{z}^{k} w(z) d^{2}z$ is $0$ for $j \neq k$ and is $2\pi  \int_{0}^{+\infty}u^{2j+1}w(u)du$ for $j=k$. So only the main diagonal contributes for the determinants in \eqref{def of Dn as n fold integral}.

Shifting $j \to j-1$ and then performing the change of variables $v = nu^{2b}$ in \eqref{simplified determinant}, we obtain
\begin{align}\nonumber
 \mathcal{E}_{n} & = \frac{\prod_{j=1}^{n}\big(\int_{0}^{\rho_{1}}+\int_{\rho_{2}}^{+\infty}\big)u^{2j+2\alpha-1} e^{-n u^{2b}} \omega(u) du}{\prod_{j=1}^{n}\big(\int_{0}^{\rho_{1}}+\int_{\rho_{2}}^{+\infty}\big)u^{2j+2\alpha-1} e^{-n u^{2b}} du}
	\\ \label{lnEn}
& = \prod_{j=1}^{n} \frac{
\big(\int_{0}^{(\rho_1/n)^{\frac{1}{2b}}}+\int_{(\rho_2/n)^{\frac{1}{2b}}}^{+\infty}\big)v^{\frac{j+ \alpha}{b} - 1} e^{-v} \omega((v/n)^{\frac{1}{2b}})dv}{\big(\int_{0}^{(\rho_1/n)^{\frac{1}{2b}}}+\int_{(\rho_2/n)^{\frac{1}{2b}}}^{+\infty}\big)v^{\frac{j+ \alpha}{b} - 1} e^{-v} dv}.
\end{align}
At this stage it is convenient to write
\begin{align}\label{def of omegaell}
\omega(x) = \sum_{\ell=1}^{2m+1}\omega_{\ell} \mathbf{1}_{[0,r_{\ell})}(x),
\end{align}
where $\{\omega_{\ell}\}_1^{2m+1}$ are given by (\ref{def of Omega j intro}) and $r_{2m+1}:=+\infty$.
Using this representation for $\omega(x)$ in (\ref{lnEn}), we arrive at the following representation for $\ln \mathcal{E}_n$:
\begin{align}
& \ln \mathcal{E}_{n} = \sum_{j=1}^{n} \ln \bigg(1  +\sum_{\ell=1}^{2m} \omega_{\ell} F_{n,j,\ell} \bigg), \label{main exact formula} 
	\\ \label{def of Fnjell}
& F_{n,j,\ell} := \begin{cases} \frac{\gamma(\tfrac{j+\alpha}{b},n r_{\ell}^{2b})}{\Gamma(\tfrac{j+\alpha}{b})-\gamma(\tfrac{j+\alpha}{b},n \rho_{2}^{2b})+\gamma(\tfrac{j+\alpha}{b},n \rho_{1}^{2b})}, & \ell=1,\dots,m, 
	\vspace{.1cm} \\
\frac{\gamma(\tfrac{j+\alpha}{b},n r_{\ell}^{2b})-\gamma(\tfrac{j+\alpha}{b},n \rho_{2}^{2b})+\gamma(\tfrac{j+\alpha}{b},n \rho_{1}^{2b})}{\Gamma(\tfrac{j+\alpha}{b})-\gamma(\tfrac{j+\alpha}{b},n \rho_{2}^{2b})+\gamma(\tfrac{j+\alpha}{b},n \rho_{1}^{2b})}, & \ell=m+1,\dots,2m, 
\end{cases}
\end{align}
where $\gamma(a,z)$ is the incomplete gamma function defined by
\begin{align*}
\gamma(a,z) = \int_{0}^{z}t^{a-1}e^{-t}dt.
\end{align*}
We infer from \eqref{main exact formula} that the asymptotics of $\gamma(a,z)$ as $z \to +\infty$ uniformly for $a\in [\frac{1+\alpha}{b},\frac{z}{b r_{1}^{2b}}+\frac{\alpha}{b}]$ are needed to obtain large $n$ asymptotics for $\mathcal{E}_{n}$ --- we recall these asymptotics in Appendix \ref{section:uniform asymp gamma}.

\medskip Note from \eqref{partition function} that $\mathcal{E}_{n}$ can also be viewed as a partition function with discontinuities approaching hard edges; see also e.g. \cite{ZW2006, LS2017, BBNY2, WebbWong, DeanoSimm, Charlier 2d gap, BC2022, BKS2022} for other works on partition functions.

\medskip In \eqref{main exact formula} and below, $\ln$ always denotes the principal branch of the logarithm.

\section{Proof of Theorem \ref{thm:main thm hard}}\label{section:proof edge}

Our proof strategy uses ideas from \cite{Charlier 2d jumps, Charlier 2d gap, ChLe2022, ACCL2022}. Define
\begin{align*}
& j_{k,-}:=\lceil \tfrac{bn\rho_{k}^{2b}}{1+\epsilon} - \alpha \rceil, \qquad j_{k,+} := \lfloor  \tfrac{bn\rho_{k}^{2b}}{1-\epsilon} - \alpha \rfloor, \qquad k=1,2,
\end{align*}
where $\epsilon > 0$ is independent of $n$ and sufficiently small such that
\begin{align}\label{rhoinequalities}
\frac{b\rho_{1}^{2b}}{1-\epsilon} < \frac{b\rho_{2}^{2b}}{1+\epsilon}, \qquad \frac{b\rho_{2}^{2b}}{1-\epsilon} < 1.
\end{align}
It is convenient to split the sum \eqref{main exact formula} into six parts:
\begin{align}\label{log Dn as a sum of sums hard}
\ln \mathcal{E}_{n} = S_{0} + S_{1} + S_{2} + S_{3} + S_{4} + S_{5},
\end{align}
where
\begin{align}
& S_{0} = \sum_{j=1}^{M'} \ln \bigg( 1+\sum_{\ell=1}^{2m} \omega_{\ell} F_{n,j,\ell} \bigg), & & S_{1} = \sum_{j=M'+1}^{j_{1,-}-1} \ln \bigg( 1 + \sum_{\ell=1}^{2m} \omega_{\ell} F_{n,j,\ell} \bigg), \label{def of S0 and S1 hard} \\
& S_{2} = \sum_{j=j_{1,-}}^{j_{1,+}} \ln \bigg( 1+\sum_{\ell=1}^{2m} \omega_{\ell} F_{n,j,\ell} \bigg), & & S_{3}=\sum_{j=j_{1,+}+1}^{j_{2,-}-1} \ln \bigg( 1 + \sum_{\ell=1}^{2m} \omega_{\ell} F_{n,j,\ell} \bigg), \label{def of S2 and S3 hard} \\
& S_{4} = \sum_{j=j_{2,-}}^{j_{2,+}} \ln \bigg( 1+\sum_{\ell=1}^{2m} \omega_{\ell} F_{n,j,\ell} \bigg), & & S_{5}=\sum_{j=j_{2,+}+1}^{n} \ln \bigg( 1 + \sum_{\ell=1}^{2m} \omega_{\ell} F_{n,j,\ell} \bigg), \label{def of S4 and S5 hard}
\end{align}
and where $M'>0$ is an integer independent of $n$. For $j=1,\dots,n$, we define $a_{j}:=\frac{j+\alpha}{b}$ and
\begin{align}\label{def etajl hard}
& \hat\lambda_{j,\ell} := \frac{bnr_{\ell}^{2b}}{j+\alpha}, & & \hat\eta_{j,\ell} := (\hat\lambda_{j,\ell}-1)\sqrt{\frac{2 (\hat\lambda_{j,\ell}-1-\ln \hat\lambda_{j,\ell})}{(\hat\lambda_{j,\ell}-1)^{2}}}, & & \ell =1,\dots,2m, \\
& \lambda_{j,k} := \frac{bn\rho_{k}^{2b}}{j+\alpha}, & & \eta_{j,k} := (\lambda_{j,k}-1)\sqrt{\frac{2 (\lambda_{j,k} -1-\ln \lambda_{j,k})}{(\lambda_{j,k} -1)^{2}}}, & & k =1,2.
\end{align}
With this notation, we can write
\begin{align}\label{Fnjellexpression}
& F_{n,j,\ell} = \begin{cases} \frac{\gamma(a_j, a_j \hat{\lambda}_{j,\ell})}{\Gamma(a_j)-\gamma(a_j,a_j \lambda_{j,2})+\gamma(a_j,a_j \lambda_{j,1})}, & \ell=1,\dots,m, 
	\vspace{.1cm} \\
\frac{\gamma(a_j, a_j \hat{\lambda}_{j,\ell})-\gamma(a_j, a_j \lambda_{j,2})+\gamma(a_j,a_j \lambda_{j,1})}{\Gamma(a_j)-\gamma(a_j,a_j \lambda_{j,2})+\gamma(a_j, a_j \lambda_{j,1})}, & \ell=m+1,\dots,2m. 
\end{cases}
\end{align}

\begin{lemma}\label{lemma: S0 hard}
Let $x_{1},\dots,x_{2m} \in \mathbb{R}$ be fixed. There exists $\delta > 0$ such that
\begin{align}\label{asymp of S0 hard}
S_{0} = M' \ln \Omega + \bigO(e^{-cn}), \qquad \mbox{as } n \to + \infty,
\end{align}
uniformly for $u_{1} \in \{z \in \mathbb{C}: |z-x_{1}|\leq \delta\},\dots,u_{2m} \in \{z \in \mathbb{C}: |z-x_{2m}|\leq \delta\}$.
\end{lemma}
\begin{proof}
Using \eqref{def of Fnjell} and Lemma \ref{lemma:various regime of gamma} we infer that
\begin{align*}
F_{n,j,\ell} = 1+\bigO(e^{-cn}), \qquad \mbox{as } n \to + \infty
\end{align*}
uniformly for $j \in \{1,\dots,M'\}$ and $\ell\in\{1,\dots,2m\}$. Hence, since $1+\sum_{\ell=1}^{2m} \omega_{\ell} = e^{u_{1}+\dots+u_{2m}} = \Omega$,
\begin{align*}
S_{0} & = \sum_{j=1}^{M'} \ln \bigg( 1+\sum_{\ell=1}^{2m} \omega_{\ell} \big[1 + \bigO(e^{-cn}) \big] \bigg) = \sum_{j=1}^{M'} \ln \Omega + \bigO(e^{-cn}), \quad \mbox{as } n \to +\infty.
\end{align*}
Since the above error terms on the left of the second equality are independent of $u_{1},\dots,u_{2m}$, the claim follows.
\end{proof}

\begin{lemma}\label{lemma: S2km1 hard}
We can choose $M'$ sufficiently large such that the following holds. For any fixed $x_{1},\dots,x_{2m} \in \mathbb{R}$, there exists $\delta > 0$ such that 
\begin{align*}
& S_{1} = (j_{1,-}-M'-1) \ln \Omega + \bigO(e^{-cn}),
\end{align*}
as $n \to +\infty$ uniformly for $u_{1} \in \{z \in \mathbb{C}: |z-x_{1}|\leq \delta\},\dots,u_{2m} \in \{z \in \mathbb{C}: |z-x_{2m}|\leq \delta\}$.
\end{lemma}
\begin{proof}
By Lemma \ref{lemma: uniform} (i), for any $\tilde{\delta}>0$ there exist $A=A(\tilde{\delta}),C=C(\tilde{\delta})>0$ such that $|\frac{\gamma(a,z)}{\Gamma(a)}-1| \leq Ce^{-\frac{a\eta^{2}}{2}}$ for all $a \geq A$ and all $\lambda=\frac{z}{a} \geq 1+\tilde{\delta}$, where $\eta$ is defined by \eqref{lol8}. For $n$ large enough, we have $\lambda_{j,k}, \hat{\lambda}_{j,\ell} \geq 1 + \frac{\epsilon}{2}$ for all $j \in \{M'+1,\dots,j_{1,-}-1\}$ and all $k$ and $\ell$. Thus, let us take $\tilde{\delta}=\frac{\epsilon}{2}$ and choose $M'$ large enough so that $a_{j} = \frac{j+\alpha}{b} \geq A(\frac{\epsilon}{2})$ for all $j \in \{M'+1,\dots,j_{1,-}-1\}$. 
From (\ref{Fnjellexpression}), we infer that, as $n \to + \infty$,
\begin{align*}
& F_{n,j,\ell} = \frac{1+\bigO(e^{-\frac{a_{j}\hat\eta_{j,\ell}^{2}}{2}})}{1+\bigO(e^{-\frac{a_{j}\eta_{j,2}^{2}}{2}}+e^{-\frac{a_{j}\eta_{j,1}^{2}}{2}})} = 1+\bigO(e^{-cn}), & & \ell=1,\dots,m, \\
& F_{n,j,\ell} = \frac{1+\bigO(e^{-\frac{a_{j}\hat\eta_{j,\ell}^{2}}{2}}+e^{-\frac{a_{j}\eta_{j,2}^{2}}{2}}+e^{-\frac{a_{j}\eta_{j,1}^{2}}{2}})}{1+\bigO(e^{-\frac{a_{j}\eta_{j,2}^{2}}{2}}+e^{-\frac{a_{j}\eta_{j,1}^{2}}{2}})} = 1+\bigO(e^{-cn}), & & \ell=m+1,\dots,2m,
\end{align*}
uniformly for $j \in \{M'+1,\dots,j_{1,-}-1\}$, and the claim follows.
\end{proof}
\begin{lemma}\label{lemma: S5 hard}
For any fixed $x_{1},\dots,x_{2m} \in \mathbb{R}$, there exists $\delta > 0$ such that
\begin{align}\label{asymp of S5 hard}
S_{5} = \bigO(e^{-cn}), \qquad \mbox{as } n \to + \infty,
\end{align}
uniformly for $u_{1} \in \{z \in \mathbb{C}: |z-x_{1}|\leq \delta\},\dots,u_{2m} \in \{z \in \mathbb{C}: |z-x_{2m}|\leq \delta\}$.
\end{lemma}
\begin{proof}
We infer from \eqref{def of Fnjell} and Lemma \ref{lemma: uniform} (ii) that
\begin{align*}
F_{n,j,\ell} = \bigO(e^{-cn}), \qquad \mbox{as } n \to + \infty
\end{align*}
uniformly for $j \in \{j_{2,+}+1,\dots,n\}$ and $\ell\in\{1,\dots,2m\}$. 
Hence
\begin{align*}
S_{5} & = \sum_{j=j_{2,+}+1}^{n} \ln \bigg( 1+\sum_{\ell=1}^{2m} \omega_{\ell} \bigO(e^{-cn}) \bigg) = \bigO(e^{-cn}), \quad \mbox{as } n \to +\infty.
\end{align*}
Since the above error terms on the left of the second equality are independent of $u_{1},\dots,u_{2m}$, the claim follows.
\end{proof}

The following lemma will be used to obtain the large $n$ asymptotics of $S_{3}$.

\begin{lemma}\label{lemma:Riemann sum NEW}[Taken from \cite[Lemma 3.4]{Charlier 2d gap}]
Let $A = A(n), a_{0}=a_0(n), B = B(n), b_{0} = b_0(n)$ be bounded functions of $n \in \{1,2,\dots\}$, such that 
\begin{align*}
& a_{n} := An + a_{0} \qquad \mbox{ and } \qquad b_{n} := Bn + b_{0}
\end{align*}
are integers. Assume also that $B-A$ is positive and remains bounded away from $0$. Let $f$ be a function independent of $n$, and which is $C^{4}([\min\{\frac{a_{n}}{n},A\},\max\{\frac{b_{n}}{n},B\}])$ for all $n\in \{1,2,\dots\}$. Then as $n \to + \infty$, we have
\begin{align}
&  \sum_{j=a_{n}}^{b_{n}}f(\tfrac{j}{n}) = n \int_{A}^{B}f(x)dx + \frac{(1-2a_{0})f(A)+(1+2b_{0})f(B)}{2}  \nonumber \\
& + \frac{(-1+6a_{0}-6a_{0}^{2})f'(A)+(1+6b_{0}+6b_{0}^{2})f'(B)}{12n}+ \frac{(-a_{0}+3a_{0}^{2}-2a_{0}^{3})f''(A)+(b_{0}+3b_{0}^{2}+2b_{0}^{3})f''(B)}{12n^{2}} \nonumber \\
& + \bigO \bigg( \frac{\mathfrak{m}_{A,n}(f''')+\mathfrak{m}_{B,n}(f''')}{n^{3}} + \sum_{j=a_{n}}^{b_{n}-1} \frac{\mathfrak{m}_{j,n}(f'''')}{n^{4}} \bigg), \label{sum f asymp gap NEW}
\end{align}
where, for a given function $g$ continuous on $[\min\{\frac{a_{n}}{n},A\},\max\{\frac{b_{n}}{n},B\}]$,
\begin{align*}
\mathfrak{m}_{A,n}(g) := \max_{x \in [\min\{\frac{a_{n}}{n},A\},\max\{\frac{a_{n}}{n},A\}]}|g(x)|, \quad \mathfrak{m}_{B,n}(g) := \max_{x \in [\min\{\frac{b_{n}}{n},B\},\max\{\frac{b_{n}}{n},B\}]}|g(x)|,
\end{align*}
and for $j \in \{a_{n},\dots,b_{n}-1\}$, $\mathfrak{m}_{j,n}(g) := \max_{x \in [\frac{j}{n},\frac{j+1}{n}]}|g(x)|$.
\end{lemma}
Following the method of \cite{Charlier 2d gap}, we define
\begin{align} \label{def of theta n eps hard}
\theta_{k,+}^{(n,\epsilon)} = \bigg( \frac{b n \rho_{k}^{2b}}{1-\epsilon}-\alpha \bigg)-\bigg\lfloor \frac{b n \rho_{k}^{2b}}{1-\epsilon}-\alpha \bigg\rfloor, \quad \theta_{k,-}^{(n,\epsilon)} = \bigg\lceil \frac{b n \rho_{k}^{2b}}{1+\epsilon}-\alpha \bigg\rceil-\bigg( \frac{b n \rho_{k}^{2b}}{1+\epsilon}-\alpha \bigg), \quad k=1,2,
\end{align}
and we split $S_{3}$ in two parts
\begin{align}\label{lol30}
S_{3} = S_{3}^{(1)}+S_{3}^{(2)},
\end{align}
where 
\begin{align}\label{def of S3p1p and S3p2p}
& S_{3}^{(1)} = \sum_{j=j_{1,+}+1}^{\lfloor j_{\star} \rfloor} \ln \bigg(1+ \sum_{\ell=1}^{2m} \omega_{\ell} F_{n,j,\ell} \bigg), \qquad S_{3}^{(2)} = \sum_{j=\lfloor j_{\star} \rfloor+1}^{j_{2,-}-1} \ln \bigg(1+ \sum_{\ell=1}^{2m} \omega_{\ell} F_{n,j,\ell} \bigg),
\end{align}
with
\begin{align}\label{def of jkstar and tkstar}
j_{\star} := n \sigma_{\star} -\alpha,
\end{align}
and where $\sigma_{\star}$ is defined in \eqref{def of taustar}. Define also
\begin{align}\label{def of theta star}
\theta_{\star} = j_{\star}- \lfloor j_{\star} \rfloor.
\end{align}
 The identity
\begin{align}\label{asymp etajk-etajkm1}
\frac{a_{j}(\eta_{j,2}^{2}-\eta_{j,1}^{2})}{2} = 2(j_{\star}-j) \ln\bigg( \frac{\rho_{2}}{\rho_{1}} \bigg),
\end{align}
implies that $\eta_{j,2}^{2}-\eta_{j,1}^{2}$ is positive for $j\in\{j_{1,+}+1,\dots,\lfloor j_{\star} \rfloor\}$ and negative for $j \in \{\lfloor j_{\star} \rfloor+1,\dots,j_{2,-}-1\}$.
\begin{lemma}\label{lemma:S3p1p asymp hard}
We can choose $M'$ sufficiently large such that the following holds. For any fixed $x_{1},\dots,x_{2m} \in \mathbb{R}$, there exists $\delta > 0$ such that 
\begin{align}\nonumber
& S_{3}^{(1)} = n \int_{\frac{b\rho_{1}^{2b}}{1-\epsilon}}^{\sigma_{\star}} f_{1}(x)dx + \bigg( \alpha - \frac{1}{2} + \theta_{1,+}^{(n,\epsilon)} \bigg) f_{1}\Big( \frac{b \rho_{1}^{2b}}{1-\epsilon} \Big) + \bigg( \frac{1}{2} - \alpha - \theta_{\star} \bigg) f_{1}(\sigma_{\star}) + \int_{\frac{b\rho_{1}^{2b}}{1-\epsilon}}^{\sigma_{\star}}f(x)dx 
	\\ \label{S31expansion}
& + \sum_{j=0}^{+\infty} \ln \bigg\{ 1 - \frac{\mathsf{T}_{0}(\sigma_{\star})+\hat{\mathsf{T}}_{0}(\sigma_{\star})}{1 + \mathsf{T}_{0}(\sigma_{\star}) + \hat{\mathsf{T}}_{0}(b\rho_{2}^{2b})} \frac{\frac{\sigma_{\star}-b\rho_{1}^{2b}}{b\rho_{2}^{2b}-\sigma_{\star}}(\frac{\rho_{2}}{\rho_{1}})^{-2(\theta_{\star}+j)}}{1+\frac{\sigma_{\star}-b\rho_{1}^{2b}}{b\rho_{2}^{2b}-\sigma_{\star}}(\frac{\rho_{2}}{\rho_{1}})^{-2(\theta_{\star}+j)}} \bigg\} + \bigO\bigg(\frac{(\ln n)^{2}}{n}\bigg)
\end{align}
as $n \to +\infty$ uniformly for $u_{1} \in \{z \in \mathbb{C}: |z-x_{1}|\leq \delta\},\dots,u_{2m} \in \{z \in \mathbb{C}: |z-x_{2m}|\leq \delta\}$, where $f_{1}(x) := \ln \big( 1+\mathsf{T}_{0}(x) + \hat{\mathsf{T}}_{0}(b\rho_{2}^{2b}) \big)$ and $f$ and $\mathsf{T}_{j}, \hat{\mathsf{T}}_{j}$ are defined in \eqref{def of f hard} and \eqref{def of T and That}.
\end{lemma} 
\begin{proof}
Using (\ref{rhoinequalities}), we see that for $j \in \{j_{1,+}+1,\ldots,j_{2,-}-1\}$ and $\ell \in \{1,\ldots,m\}$, $1-\hat\lambda_{j,\ell}$ and $1-\lambda_{j,1}$ are positive and remain bounded away from $0$, while for $j \in \{j_{1,+}+1,\ldots,j_{2,-}-1\}$ and $\ell \in \{m+1,\ldots,2m\}$, $\hat\lambda_{j,\ell}-1$ and $\lambda_{j,2}-1$ are positive and remain bounded away from $0$. 

Hence, using Lemma \ref{lemma: asymp of gamma for lambda one over sqrt away from 1} $(i)$ and $(ii)$, as $n \to + \infty$ we have
\begin{align}
& \frac{\gamma(a_j,a_j \hat{\lambda}_{j,\ell})}{\Gamma(a_j)} = \frac{e^{-\frac{a_{j}\hat\eta_{j,\ell}^{2}}{2}}}{\sqrt{2\pi}} \bigg\{ \frac{1}{1-\hat{\lambda}_{j,\ell}}\frac{1}{\sqrt{a_{j}}} + \frac{1+10\hat{\lambda}_{j,\ell}+\hat{\lambda}_{j,\ell}^{2}}{12(\hat{\lambda}_{j,\ell}-1)^{3}}\frac{1}{a_{j}^{3/2}} + \bigO(n^{-5/2}) \bigg\}, & & \ell=1,\ldots,m, \nonumber \\
& \frac{\gamma(a_j,a_j \hat{\lambda}_{j,\ell})}{\Gamma(a_j)} = 1+\frac{e^{-\frac{a_{j}\hat\eta_{j,\ell}^{2}}{2}}}{\sqrt{2\pi}} \bigg\{ \frac{-1}{\hat{\lambda}_{j,\ell}-1}\frac{1}{\sqrt{a_{j}}} + \frac{1+10\hat{\lambda}_{j,\ell}+\hat{\lambda}_{j,\ell}^{2}}{12(\hat{\lambda}_{j,\ell}-1)^{3}}\frac{1}{a_{j}^{3/2}} + \bigO(n^{-5/2}) \bigg\}, & & \ell=m+1,\ldots,2m, 
	\nonumber \\
& \frac{\gamma(a_j,a_j \lambda_{j,1})}{\Gamma(a_j)} = \frac{e^{-\frac{a_{j}\eta_{j,1}^{2}}{2}}}{\sqrt{2\pi}} \bigg\{ \frac{1}{1-\lambda_{j,1}}\frac{1}{\sqrt{a_{j}}} + \frac{1+10\lambda_{j,1}+\lambda_{j,1}^{2}}{12(\lambda_{j,1}-1)^{3}}\frac{1}{a_{j}^{3/2}} + \bigO(n^{-5/2}) \bigg\}, \nonumber \\
& \frac{\gamma(a_j,a_j \lambda_{j,2})}{\Gamma(a_j)} = 1+\frac{e^{-\frac{a_{j}\eta_{j,2}^{2}}{2}}}{\sqrt{2\pi}} \bigg\{ \frac{-1}{\lambda_{j,2}-1}\frac{1}{\sqrt{a_{j}}} + \frac{1+10\lambda_{j,2}+\lambda_{j,2}^{2}}{12(\lambda_{j,2}-1)^{3}}\frac{1}{a_{j}^{3/2}} + \bigO(n^{-5/2}) \bigg\}, \label{lol2}
\end{align}
uniformly for $j \in \{j_{1,+}+1,\dots,j_{2,-}-1\}$. Moreover, as $n \to + \infty$,
\begin{align}
& e^{-\frac{a_{j} \hat{\eta}_{j,\ell}^{2}}{2}} = e^{-\frac{a_{j}\eta_{j,1}^{2}}{2}}e^{-\frac{t_{\ell}}{b}(\frac{j}{n}-b\rho_{1}^{2b})} \bigg\{ 1 - \frac{t_{\ell}^{2}j/n+2t_{\ell}\alpha}{2bn} + \bigO(n^{-2}) \bigg\}, & & \ell=1,\dots,m, \nonumber \\
& e^{-\frac{a_{j} \hat{\eta}_{j,\ell}^{2}}{2}} = e^{-\frac{a_{j}\eta_{j,2}^{2}}{2}}e^{-\frac{t_{\ell}}{b}(b\rho_{2}^{2b}-\frac{j}{n})} \bigg\{ 1 - \frac{t_{\ell}^{2}j/n-2t_{\ell}\alpha}{2bn} + \bigO(n^{-2}) \bigg\}, & & \ell=m+1,\dots,2m, \label{lol3}
\end{align}
uniformly for $j \in \{j_{1,+}+1,\dots,j_{2,-}-1\}$. Hence, after a long but direct computation,  we obtain
\begin{align*}
& F_{n,j,l} = \frac{e^{-\frac{t_{\ell}}{b}(j/n-b\rho_{1}^{2b})}}{1+\frac{j/n-b\rho_{1}^{2b}}{b\rho_{2}^{2b}-j/n}\big(\frac{\rho_{2}}{\rho_{1}}\big)^{-2(j_{\star}-j)}} + t_{\ell}e^{-\frac{t_{\ell}}{b}(j/n-b\rho_{1}^{2b})} \frac{(j/n)^{2}t_{\ell}+2b^{2}\rho_{1}^{2b} - t_{\ell}b\rho_{1}^{2b}j/n+2\alpha(j/n-b\rho_{1}^{2b})}{-2b(j/n-b\rho_{1}^{2b})n} , \\
& + \bigO\bigg( \frac{\big(\frac{\rho_{2}}{\rho_{1}}\big)^{-2(j_{\star}-j)}}{n} + \frac{1}{n^{2}} \bigg), \hspace{7.55cm} \ell=1,\dots,m, \\
& F_{n,j,l} = 1 - \frac{e^{-\frac{t_{\ell}}{b}(b\rho_{2}^{2b} - j/n)} \frac{j/n-b\rho_{1}^{2b}}{b\rho_{2}^{2b}-j/n}\big(\frac{\rho_{2}}{\rho_{1}}\big)^{-2(j_{\star}-j)}}{1+\frac{j/n-b\rho_{1}^{2b}}{b\rho_{2}^{2b}-j/n}\big(\frac{\rho_{2}}{\rho_{1}}\big)^{-2(j_{\star}-j)}} + \bigO\bigg( \frac{\big(\frac{\rho_{2}}{\rho_{1}}\big)^{-2(j_{\star}-j)}}{n} + \frac{1}{n^{2}} \bigg), \quad \ell=m+1,\dots,2m,
\end{align*}
as $n \to + \infty$ uniformly for $j \in \{j_{1,+}+1,\dots,\lfloor j_{\star} \rfloor\}$. Substituting the above into \eqref{def of S3p1p and S3p2p} yields
\begin{align*}
& S_{3}^{(1)} = \sum_{j=j_{1,+}+1}^{\lfloor j_{\star} \rfloor} \ln \bigg\{ 1 + \frac{\mathsf{T}_{0}(j/n)}{1+\frac{j/n-b\rho_{1}^{2b}}{b\rho_{2}^{2b}-j/n}(\frac{\rho_{2}}{\rho_{1}})^{-2(j_{\star}-j)}} + \hat{\mathsf{T}}_{0}(b\rho_{2}^{2b}) - \frac{\hat{\mathsf{T}}_{0}(j/n) \frac{j/n-b\rho_{1}^{2b}}{b\rho_{2}^{2b}-j/n}(\frac{\rho_{2}}{\rho_{1}})^{-2(j_{\star}-j)}}{1+\frac{j/n-b\rho_{1}^{2b}}{b\rho_{2}^{2b}-j/n}(\frac{\rho_{2}}{\rho_{1}})^{-2(j_{\star}-j)}} \\
& + \frac{(j/n)^{2} \mathsf{T}_{2}(j/n) + 2b^{2}\rho_{1}^{2b}\mathsf{T}_{1}(j/n) - b\rho_{1}^{2b} j/n \; \mathsf{T}_{2}(j/n)+2\alpha(j/n-b\rho_{1}^{2b})\mathsf{T}_{1}(j/n)}{-2b(j/n-b\rho_{1}^{2b})n} \\
& + \sum_{\ell=1}^{2m} \omega_{\ell} \bigO\bigg( \frac{\big(\frac{\rho_{2}}{\rho_{1}}\big)^{-2(j_{\star}-j)}}{n} + \frac{1}{n^{2}} \bigg)  \bigg\}
\end{align*}
as $n \to + \infty$, where the above error terms are independent of $u_{1},\dots,u_{2m}$ and uniform for $j \in \{j_{1,+}+1,\dots,\lfloor j_{\star} \rfloor\}$. Expanding further, we obtain
\begin{align*}
S_{3}^{(1)} = \mathcal{S}_{1} + \mathcal{S}_{2} + \frac{1}{n}\mathcal{S}_{3} + \sum_{j=j_{1,+}+1}^{\lfloor j_{\star} \rfloor} \bigO\bigg( \frac{\big(\frac{\rho_{2}}{\rho_{1}}\big)^{-2(j_{\star}-j)}}{n} + \frac{1}{n^{2}} \bigg), \qquad \mbox{as } n \to + \infty,
\end{align*}
where
\begin{align*}
& \mathcal{S}_{1} = \sum_{j=j_{1,+}+1}^{\lfloor j_{\star} \rfloor} \ln \Big\{ 1 + \mathsf{T}_{0}(j/n) + \hat{\mathsf{T}}_{0}(b\rho_{2}^{2b}) \Big\}, \\
& \mathcal{S}_{2} = \sum_{j=j_{1,+}+1}^{\lfloor j_{\star} \rfloor} \ln \bigg\{ 1 - \frac{\mathsf{T}_{0}(j/n)+\hat{\mathsf{T}}_{0}(j/n)}{1 + \mathsf{T}_{0}(j/n) + \hat{\mathsf{T}}_{0}(b\rho_{2}^{2b})} \frac{\frac{j/n-b\rho_{1}^{2b}}{b\rho_{2}^{2b}-j/n}(\frac{\rho_{2}}{\rho_{1}})^{-2(j_{\star}-j)}}{1+\frac{j/n-b\rho_{1}^{2b}}{b\rho_{2}^{2b}-j/n}(\frac{\rho_{2}}{\rho_{1}})^{-2(j_{\star}-j)}} \bigg\}, \\
& \mathcal{S}_{3} = \sum_{j=j_{1,+}+1}^{\lfloor j_{\star} \rfloor} \frac{(j/n)^{2} \mathsf{T}_{2}(j/n) + 2b^{2}\rho_{1}^{2b}\mathsf{T}_{1}(j/n) - b\rho_{1}^{2b} j/n \; \mathsf{T}_{2}(j/n)+2\alpha(j/n-b\rho_{1}^{2b})\mathsf{T}_{1}(j/n)}{-2b(j/n-b\rho_{1}^{2b})(1 + \mathsf{T}_{0}(j/n) + \hat{\mathsf{T}}_{0}(b\rho_{2}^{2b}))},
\end{align*}
and we have used that
\begin{align*}
\bigg(1 - \frac{\mathsf{T}_{0}(j/n)+\hat{\mathsf{T}}_{0}(j/n)}{1 + \mathsf{T}_{0}(j/n) + \hat{\mathsf{T}}_{0}(b\rho_{2}^{2b})} \frac{\frac{j/n-b\rho_{1}^{2b}}{b\rho_{2}^{2b}-j/n}(\frac{\rho_{2}}{\rho_{1}})^{-2(j_{\star}-j)}}{1+\frac{j/n-b\rho_{1}^{2b}}{b\rho_{2}^{2b}-j/n}(\frac{\rho_{2}}{\rho_{1}})^{-2(j_{\star}-j)}}\bigg)^{-1} = 1 + \bigO\bigg( \Big(\frac{\rho_{2}}{\rho_{1}}\Big)^{-2(j_{\star}-j)} \bigg)
\end{align*}
uniformly for $j \in \{j_{1,+}+1,\dots,\lfloor j_{\star} \rfloor\}$ to obtain the expression for $\mathcal{S}_{3}$. Clearly, $M'$ can be chosen large enough such that
\begin{align*}
& \sum_{j=j_{1,+}+1}^{\lfloor j_{\star} \rfloor} \bigO\bigg( \frac{\big(\frac{\rho_{2}}{\rho_{1}}\big)^{-2(j_{\star}-j)}}{n} + \frac{1}{n^{2}} \bigg) = \bigO (n^{-1}) + \sum_{j=j_{1,+}+1}^{\lfloor j_{\star} \rfloor} \bigO\bigg( \frac{\big(\frac{\rho_{2}}{\rho_{1}}\big)^{-2(j_{\star}-j)}}{n} \bigg) \\
& = \bigO (n^{-1}) + \bigO(n^{-100}) + \sum_{j=\lfloor j_{\star}-M' \ln n \rfloor}^{\lfloor j_{\star} \rfloor} \bigO\bigg( \frac{\big(\frac{\rho_{2}}{\rho_{1}}\big)^{-2(j_{\star}-j)}}{n} \bigg) = \bigO\bigg( \frac{\ln n}{n} \bigg).
\end{align*}
Also, we can express $\mathcal{S}_1$ and $\mathcal{S}_3$ in terms of the functions $f_1$ and $f$ as follows:
$$\mathcal{S}_{1} = \sum_{j=j_{1,+}+1}^{\lfloor j_{\star} \rfloor} f_1(j/n), \qquad
\mathcal{S}_{3} = \sum_{j=j_{1,+}+1}^{\lfloor j_{\star} \rfloor} f(j/n).$$
These sums can be expanded using Lemma \ref{lemma:Riemann sum NEW} (with $A=\frac{b\rho_{1}^{2b}}{1-\epsilon}$, $a_{0} = 1-\alpha-\theta_{1,+}^{(n,\epsilon)}$, $B=\sigma_{\star}$ and $b_{0} = -\alpha-\theta_{\star}$); this gives the terms involving $f_1$ and $f$ on the right-hand side of (\ref{S31expansion}). Thus it only remains to expand $\mathcal{S}_{2}$ (the analysis required for $\mathcal{S}_{2}$ is similar but different from \cite[Lemma 3.6]{Charlier 2d gap}). Let $s_{2,j}$ be the summand of $\mathcal{S}_{2}$. Since $s_{2,j} = \bigO((\frac{\rho_{2}}{\rho_{1}})^{-2(j_{\star}-j)})$ as $n \to + \infty$ uniformly for $j \in \{j_{1,+}+1,\dots,\lfloor j_{\star} \rfloor\}$, we have
\begin{align*}
\mathcal{S}_{2} = \sum_{j=\lfloor j_{\star}-M' \ln n \rfloor}^{\lfloor j_{\star} \rfloor} s_{2,j} + \bigO(n^{-100}), \qquad \mbox{as } n \to + \infty,
\end{align*}
provided $M'$ is chosen large enough. Furthermore, since $\mathsf{T}_{0}$ and $\hat{\mathsf{T}}_{0}$ are analytic at $\sigma_{\star}$,
\begin{align*}
\mathcal{S}_{2} & = \sum_{j=\lfloor j_{\star}-M' \ln n \rfloor}^{\lfloor j_{\star} \rfloor} \bigg[ \ln \bigg\{ 1 - \frac{\mathsf{T}_{0}(\sigma_{\star})+\hat{\mathsf{T}}_{0}(\sigma_{\star})}{1 + \mathsf{T}_{0}(\sigma_{\star}) + \hat{\mathsf{T}}_{0}(b\rho_{2}^{2b})} \frac{\frac{\sigma_{\star}-b\rho_{1}^{2b}}{b\rho_{2}^{2b}-\sigma_{\star}}(\frac{\rho_{2}}{\rho_{1}})^{-2(j_{\star}-j)}}{1+\frac{\sigma_{\star}-b\rho_{1}^{2b}}{b\rho_{2}^{2b}-\sigma_{\star}}(\frac{\rho_{2}}{\rho_{1}})^{-2(j_{\star}-j)}} \bigg\} + \bigO\big(\sigma_{\star}-j/n\big) \bigg] + \bigO(n^{-100})
	\\
& = \sum_{j=\lfloor j_{\star}-M' \ln n \rfloor}^{\lfloor j_{\star} \rfloor} \ln \bigg\{ 1 - \frac{\mathsf{T}_{0}(\sigma_{\star})+\hat{\mathsf{T}}_{0}(\sigma_{\star})}{1 + \mathsf{T}_{0}(\sigma_{\star}) + \hat{\mathsf{T}}_{0}(b\rho_{2}^{2b})} \frac{\frac{\sigma_{\star}-b\rho_{1}^{2b}}{b\rho_{2}^{2b}-\sigma_{\star}}(\frac{\rho_{2}}{\rho_{1}})^{-2(j_{\star}-j)}}{1+\frac{\sigma_{\star}-b\rho_{1}^{2b}}{b\rho_{2}^{2b}-\sigma_{\star}}(\frac{\rho_{2}}{\rho_{1}})^{-2(j_{\star}-j)}} \bigg\} + \bigO\bigg(\frac{(\ln n)^{2}}{n}\bigg) \\
& = \sum_{j=-\infty}^{\lfloor j_{\star} \rfloor} \ln \bigg\{ 1 - \frac{\mathsf{T}_{0}(\sigma_{\star})+\hat{\mathsf{T}}_{0}(\sigma_{\star})}{1 + \mathsf{T}_{0}(\sigma_{\star}) + \hat{\mathsf{T}}_{0}(b\rho_{2}^{2b})} \frac{\frac{\sigma_{\star}-b\rho_{1}^{2b}}{b\rho_{2}^{2b}-\sigma_{\star}}(\frac{\rho_{2}}{\rho_{1}})^{-2(j_{\star}-j)}}{1+\frac{\sigma_{\star}-b\rho_{1}^{2b}}{b\rho_{2}^{2b}-\sigma_{\star}}(\frac{\rho_{2}}{\rho_{1}})^{-2(j_{\star}-j)}} \bigg\} + \bigO\bigg(\frac{(\ln n)^{2}}{n}\bigg)
\end{align*}
as $n \to + \infty$. Changing now indices, we find
\begin{align*}
\mathcal{S}_{2} = \sum_{j=0}^{+\infty} \ln \bigg\{ 1 - \frac{\mathsf{T}_{0}(\sigma_{\star})+\hat{\mathsf{T}}_{0}(\sigma_{\star})}{1 + \mathsf{T}_{0}(\sigma_{\star}) + \hat{\mathsf{T}}_{0}(b\rho_{2}^{2b})} \frac{\frac{\sigma_{\star}-b\rho_{1}^{2b}}{b\rho_{2}^{2b}-\sigma_{\star}}(\frac{\rho_{2}}{\rho_{1}})^{-2(\theta_{\star}+j)}}{1+\frac{\sigma_{\star}-b\rho_{1}^{2b}}{b\rho_{2}^{2b}-\sigma_{\star}}(\frac{\rho_{2}}{\rho_{1}})^{-2(\theta_{\star}+j)}} \bigg\} + \bigO\bigg(\frac{(\ln n)^{2}}{n}\bigg) \quad \mbox{as } n \to + \infty.
\end{align*}
The claim follows after a computation combining the asymptotics of $\mathcal{S}_{1},\mathcal{S}_{2},\mathcal{S}_{3}$.
\end{proof}
\begin{lemma}\label{lemma:S3p2p asymp hard}
We can choose $M'$ sufficiently large such that the following holds. For any fixed $x_{1},\dots,x_{2m} \in \mathbb{R}$, there exists $\delta > 0$ such that 
\begin{align}\nonumber
& S_{3}^{(2)} = n \int_{\sigma_{\star}}^{\frac{b\rho_{2}^{2b}}{1+\epsilon}} \hat{f}_{1}(x)dx + \bigg( \alpha - \frac{1}{2} + \theta_{\star} \bigg) \hat{f}_{1}( \sigma_{\star} ) + \bigg( \theta_{2,-}^{(n,\epsilon)}-\alpha-\frac{1}{2} \bigg) \hat{f}_{1}\Big(\frac{b\rho_{2}^{2b}}{1+\epsilon}\Big) + \int_{\sigma_{\star}}^{\frac{b\rho_{2}^{2b}}{1+\epsilon}}\hat{f}(x)dx 
	\\ \label{S32expansion}
& + \sum_{j=0}^{+\infty} \ln \bigg\{ 1 + \frac{\mathsf{T}_{0}(\sigma_{\star})+\hat{\mathsf{T}}_{0}(\sigma_{\star})}{1 - \hat{\mathsf{T}}_{0}(\sigma_{\star}) + \hat{\mathsf{T}}_{0}(b\rho_{2}^{2b})} \frac{\frac{b\rho_{2}^{2b}-\sigma_{\star}}{\sigma_{\star}-b\rho_{1}^{2b}}\big(\frac{\rho_{2}}{\rho_{1}}\big)^{-2(j+1-\theta_{\star})}}{1+\frac{b\rho_{2}^{2b}-\sigma_{\star}}{\sigma_{\star}-b\rho_{1}^{2b}}\big(\frac{\rho_{2}}{\rho_{1}}\big)^{-2(j+1-\theta_{\star})}} \bigg\} + \bigO\bigg(\frac{(\ln n)^{2}}{n}\bigg)
\end{align}
as $n \to +\infty$ uniformly for $u_{1} \in \{z \in \mathbb{C}: |z-x_{1}|\leq \delta\},\dots,u_{2m} \in \{z \in \mathbb{C}: |z-x_{2m}|\leq \delta\}$, where $\hat{f}_{1}(x) := \ln \big( 1-\hat{\mathsf{T}}_{0}(x) + \hat{\mathsf{T}}_{0}(b\rho_{2}^{2b}) \big)$ and $\hat{f}$ and $\mathsf{T}_{j}, \hat{\mathsf{T}}_{j}$ are defined in \eqref{def of fh hard} and \eqref{def of T and That}.
\end{lemma}
\begin{proof}
After a long but direct computation using \eqref{lol2} and \eqref{lol3}, we obtain
\begin{align*}
& F_{n,j,l} = e^{-\frac{t_{\ell}}{b}(j/n-b\rho_{1}^{2b})} \frac{\frac{b\rho_{2}^{2b}-j/n}{j/n-b\rho_{1}^{2b}}\big(\frac{\rho_{2}}{\rho_{1}}\big)^{-2(j-j_{\star})}}{1+ \frac{b\rho_{2}^{2b}-j/n}{j/n-b\rho_{1}^{2b}}\big(\frac{\rho_{2}}{\rho_{1}}\big)^{-2(j-j_{\star})}} + \bigO\bigg( \frac{\big(\frac{\rho_{2}}{\rho_{1}}\big)^{-2(j-j_{\star})}}{n} + \frac{1}{n^{2}} \bigg), \quad \ell=1,\dots,m, \\
& F_{n,j,l} = 1 - \frac{e^{-\frac{t_{\ell}}{b}(b\rho_{2}^{2b}-j/n)} }{1+\frac{b\rho_{2}^{2b}-j/n}{j/n-b\rho_{1}^{2b}}\big(\frac{\rho_{2}}{\rho_{1}}\big)^{-2(j-j_{\star})}} + \frac{t_{\ell}}{n}e^{-\frac{t_{\ell}}{b}(b\rho_{2}^{2b}-j/n)} \bigg( \frac{b\rho_{2}^{2b}}{b\rho_{2}^{2b}-j/n} - \frac{\alpha}{b} + \frac{t_{\ell}\, j/n}{2b} \bigg) \\
& + \bigO\bigg( \frac{\big(\frac{\rho_{2}}{\rho_{1}}\big)^{-2(j-j_{\star})}}{n} + \frac{1}{n^{2}} \bigg), \quad \ell=m+1,\dots,2m,
\end{align*}
as $n \to + \infty$ uniformly for $j \in \{\lfloor j_{\star} \rfloor+1,\dots,j_{2,-}-1\}$. Hence, 
\begin{align*}
& S_{3}^{(2)} = \sum_{j=\lfloor j_{\star} \rfloor+1}^{j_{2,-}-1} \ln \bigg\{ 1 + \mathsf{T}_{0}(j/n)\frac{ \frac{b\rho_{2}^{2b}-j/n}{j/n-b\rho_{1}^{2b}}\big(\frac{\rho_{2}}{\rho_{1}}\big)^{-2(j-j_{\star})}}{1+\frac{b\rho_{2}^{2b}-j/n}{j/n-b\rho_{1}^{2b}}\big(\frac{\rho_{2}}{\rho_{1}}\big)^{-2(j-j_{\star})}} + \hat{\mathsf{T}}_{0}(b\rho_{2}^{2b}) - \frac{\hat{\mathsf{T}}_{0}(j/n)}{1+\frac{b\rho_{2}^{2b}-j/n}{j/n-b\rho_{1}^{2b}}\big(\frac{\rho_{2}}{\rho_{1}}\big)^{-2(j-j_{\star})}} \\
& + \frac{1}{n} \bigg\{ \bigg( \frac{b\rho_{2}^{2b}}{b\rho_{2}^{2b}-j/n} - \frac{\alpha}{b} \bigg) \hat{\mathsf{T}}_{1}(j/n) + \frac{j/n}{2b} \hat{\mathsf{T}}_{2}(j/n) \bigg\} + \sum_{\ell=1}^{2m} \omega_{\ell} \bigO\bigg( \frac{\big(\frac{\rho_{2}}{\rho_{1}}\big)^{-2(j-j_{\star})}}{n} + \frac{1}{n^{2}} \bigg)  \bigg\}
\end{align*}
as $n \to + \infty$, where the above error terms are independent of $u_{1},\dots,u_{2m}$ and uniform for $j \in \{ \lfloor j_{\star} \rfloor+1,\dots,j_{2,-}-1 \}$. Expanding further, we obtain
\begin{align*}
S_{3}^{(2)} = \mathcal{S}_{4} + \mathcal{S}_{5} + \frac{1}{n}\mathcal{S}_{6} + \bigO\bigg( \frac{\ln n}{n} \bigg), \qquad \mbox{as } n \to + \infty,
\end{align*}
where
\begin{align*}
& \mathcal{S}_{4} = \sum_{j=\lfloor j_{\star} \rfloor+1}^{j_{2,-}-1} \ln \Big\{ 1 - \hat{\mathsf{T}}_{0}(j/n) + \hat{\mathsf{T}}_{0}(b\rho_{2}^{2b}) \Big\}, \\
& \mathcal{S}_{5} = \sum_{j=\lfloor j_{\star} \rfloor+1}^{j_{2,-}-1} \ln \bigg\{ 1 + \frac{\mathsf{T}_{0}(j/n)+\hat{\mathsf{T}}_{0}(j/n)}{1 - \hat{\mathsf{T}}_{0}(j/n) + \hat{\mathsf{T}}_{0}(b\rho_{2}^{2b})} \frac{\frac{b\rho_{2}^{2b}-j/n}{j/n-b\rho_{1}^{2b}}\big(\frac{\rho_{2}}{\rho_{1}}\big)^{-2(j-j_{\star})}}{1+\frac{b\rho_{2}^{2b}-j/n}{j/n-b\rho_{1}^{2b}}\big(\frac{\rho_{2}}{\rho_{1}}\big)^{-2(j-j_{\star})}} \bigg\}, \\
& \mathcal{S}_{6} = \sum_{j=\lfloor j_{\star} \rfloor+1}^{j_{2,-}-1} \frac{\big( \frac{b\rho_{2}^{2b}}{b\rho_{2}^{2b}-j/n} - \frac{\alpha}{b} \big) \hat{\mathsf{T}}_{1}(j/n) + \frac{j/n}{2b} \hat{\mathsf{T}}_{2}(j/n)}{1 - \hat{\mathsf{T}}_{0}(j/n) + \hat{\mathsf{T}}_{0}(b\rho_{2}^{2b})}.
\end{align*}
We can express $\mathcal{S}_4$ and $\mathcal{S}_6$ in terms of the functions $\hat{f}_1$ and $\hat{f}$ as follows:
$$\mathcal{S}_{4} = \sum_{j=j_{1,+}+1}^{\lfloor j_{\star} \rfloor} \hat{f}_1(j/n), \qquad
\mathcal{S}_{6} = \sum_{j=j_{1,+}+1}^{\lfloor j_{\star} \rfloor} \hat{f}(j/n).$$
These sums can be expanded using Lemma \ref{lemma:Riemann sum NEW} (with $A=\sigma_{\star}$, $a_{0} = 1-\alpha-\theta_{\star}$, $B= \frac{b\rho_{2}^{2b}}{1+\epsilon}$ and $b_{0} = \theta_{2,-}^{(n,\epsilon)}-\alpha-1$); this gives the terms involving $\hat{f}_1$ and $\hat{f}$ on the right-hand side of (\ref{S32expansion}). We now turn to the analysis of $\mathcal{S}_{5}$. Let $s_{5,j}$ be the summand of $\mathcal{S}_{5}$. Since $s_{5,j} = \bigO((\frac{\rho_{2}}{\rho_{1}})^{-2(j-j_{\star})})$ as $n \to + \infty$ uniformly for $j \in \{\lfloor j_{\star} \rfloor+1,\dots,j_{2,-}-1\}$, we have
\begin{align*}
\mathcal{S}_{5} = \sum_{j=\lfloor j_{\star} \rfloor+1}^{\lfloor j_{\star}+M' \ln n \rfloor} s_{5,j} + \bigO(n^{-100}), \qquad \mbox{as } n \to + \infty,
\end{align*}
provided $M'$ is chosen large enough. Furthermore, since $\mathsf{T}_{0}$ and $\hat{\mathsf{T}}_{0}$ are analytic at $\sigma_{\star}$,
\begin{align*}
\mathcal{S}_{5} & = \sum_{j=\lfloor j_{\star} \rfloor+1}^{\lfloor j_{\star}+M' \ln n \rfloor} \bigg[ \ln \bigg\{ 1 + \frac{\mathsf{T}_{0}(\sigma_{\star})+\hat{\mathsf{T}}_{0}(\sigma_{\star})}{1 - \hat{\mathsf{T}}_{0}(\sigma_{\star}) + \hat{\mathsf{T}}_{0}(b\rho_{2}^{2b})} \frac{\frac{b\rho_{2}^{2b}-\sigma_{\star}}{\sigma_{\star}-b\rho_{1}^{2b}}\big(\frac{\rho_{2}}{\rho_{1}}\big)^{-2(j-j_{\star})}}{1+\frac{b\rho_{2}^{2b}-\sigma_{\star}}{\sigma_{\star}-b\rho_{1}^{2b}}\big(\frac{\rho_{2}}{\rho_{1}}\big)^{-2(j-j_{\star})}} \bigg\} + \bigO\big(j/n-\sigma_{\star}\big) \bigg] + \bigO(n^{-100}) \\
& = \sum_{j=\lfloor j_{\star} \rfloor+1}^{\lfloor j_{\star}+M' \ln n \rfloor} \ln \bigg\{ 1 + \frac{\mathsf{T}_{0}(\sigma_{\star})+\hat{\mathsf{T}}_{0}(\sigma_{\star})}{1 - \hat{\mathsf{T}}_{0}(\sigma_{\star}) + \hat{\mathsf{T}}_{0}(b\rho_{2}^{2b})} \frac{\frac{b\rho_{2}^{2b}-\sigma_{\star}}{\sigma_{\star}-b\rho_{1}^{2b}}\big(\frac{\rho_{2}}{\rho_{1}}\big)^{-2(j-j_{\star})}}{1+\frac{b\rho_{2}^{2b}-\sigma_{\star}}{\sigma_{\star}-b\rho_{1}^{2b}}\big(\frac{\rho_{2}}{\rho_{1}}\big)^{-2(j-j_{\star})}} \bigg\} + \bigO\bigg(\frac{(\ln n)^{2}}{n}\bigg) \\
& = \sum_{j=\lfloor j_{\star} \rfloor+1}^{+\infty} \ln \bigg\{ 1 + \frac{\mathsf{T}_{0}(\sigma_{\star})+\hat{\mathsf{T}}_{0}(\sigma_{\star})}{1 - \hat{\mathsf{T}}_{0}(\sigma_{\star}) + \hat{\mathsf{T}}_{0}(b\rho_{2}^{2b})} \frac{\frac{b\rho_{2}^{2b}-\sigma_{\star}}{\sigma_{\star}-b\rho_{1}^{2b}}\big(\frac{\rho_{2}}{\rho_{1}}\big)^{-2(j-j_{\star})}}{1+\frac{b\rho_{2}^{2b}-\sigma_{\star}}{\sigma_{\star}-b\rho_{1}^{2b}}\big(\frac{\rho_{2}}{\rho_{1}}\big)^{-2(j-j_{\star})}} \bigg\} + \bigO\bigg(\frac{(\ln n)^{2}}{n}\bigg)
\end{align*}
as $n \to + \infty$. Changing indices, we get
\begin{align*}
\mathcal{S}_{5} = \sum_{j=0}^{+\infty} \ln \bigg\{ 1 + \frac{\mathsf{T}_{0}(\sigma_{\star})+\hat{\mathsf{T}}_{0}(\sigma_{\star})}{1 - \hat{\mathsf{T}}_{0}(\sigma_{\star}) + \hat{\mathsf{T}}_{0}(b\rho_{2}^{2b})} \frac{\frac{b\rho_{2}^{2b}-\sigma_{\star}}{\sigma_{\star}-b\rho_{1}^{2b}}\big(\frac{\rho_{2}}{\rho_{1}}\big)^{-2(j+1-\theta_{\star})}}{1+\frac{b\rho_{2}^{2b}-\sigma_{\star}}{\sigma_{\star}-b\rho_{1}^{2b}}\big(\frac{\rho_{2}}{\rho_{1}}\big)^{-2(j+1-\theta_{\star})}} \bigg\} + \bigO\bigg(\frac{(\ln n)^{2}}{n}\bigg), \quad \mbox{as } n \to + \infty.
\end{align*}
The claim follows after a computation combining the asymptotics of $\mathcal{S}_{4},\mathcal{S}_{5},\mathcal{S}_{6}$.
\end{proof}

Recall that $\sigma_{1}, \sigma_{2} > 0$ were defined in (\ref{def of taustar}), $\mathsf{Q} = \mathsf{Q}(\vec{t}, \vec{u})$ was defined in (\ref{mathsfQdef}), and $\mathcal{F}_n$ was defined in terms of the Jacobi theta function $\theta(z; \tau)$ in (\ref{mathcalFndef}).

\begin{lemma}\label{lemma:S3 asymp hard}
We can choose $M'$ sufficiently large such that the following holds. For any fixed $x_{1},\dots,x_{2m} \in \mathbb{R}$, there exists $\delta > 0$ such that 
\begin{align*}
 S_{3} = &\; n \bigg\{\int_{\frac{b\rho_{1}^{2b}}{1-\epsilon}}^{\sigma_{\star}} f_{1}(x)dx + \int_{\sigma_{\star}}^{\frac{b\rho_{2}^{2b}}{1+\epsilon}} \hat{f}_{1}(x)dx \bigg\} + \bigg( \alpha - \frac{1}{2} + \theta_{1,+}^{(n,\epsilon)} \bigg) f_{1}\Big( \frac{b \rho_{1}^{2b}}{1-\epsilon} \Big)  + \bigg( \theta_{2,-}^{(n,\epsilon)}-\alpha-\frac{1}{2} \bigg) \hat{f}_{1}\Big(\frac{b\rho_{2}^{2b}}{1+\epsilon}\Big)  \\
& + \bigg( \frac{1}{2} - \alpha - \theta_{\star} \bigg) \big(f_{1}(\sigma_{\star}) - \hat{f}_{1}( \sigma_{\star} ) \big) + \int_{\frac{b\rho_{1}^{2b}}{1-\epsilon}}^{\sigma_{\star}}f(x)dx + \int_{\sigma_{\star}}^{\frac{b\rho_{2}^{2b}}{1+\epsilon}}\hat{f}(x)dx \\
& + \theta_\star \ln \mathsf{Q} -\frac{\ln{\mathsf{Q}}}{2} \bigg(
1 - \frac{2\ln(\sigma_{2}/\sigma_{1}) + \ln{\mathsf{Q}}}{2 \ln(\rho_{2}/\rho_{1})}\bigg) 
+ \mathcal{F}_n + \bigO\bigg(\frac{(\ln n)^{2}}{n}\bigg)
\end{align*}
as $n \to +\infty$ uniformly for $u_{1} \in \{z \in \mathbb{C}: |z-x_{1}|\leq \delta\},\dots,u_{2m} \in \{z \in \mathbb{C}: |z-x_{2m}|\leq \delta\}$, where $f_{1}(x) := \ln \big( 1+\mathsf{T}_{0}(x) + \hat{\mathsf{T}}_{0}(b\rho_{2}^{2b}) \big)$, $\hat{f}_{1}(x) := \ln \big( 1-\hat{\mathsf{T}}_{0}(x) + \hat{\mathsf{T}}_{0}(b\rho_{2}^{2b}) \big)$ and $\mathsf{T}_{j}$, $\hat{\mathsf{T}}_{j}$, $f$ and $\hat{f}$ are defined in \eqref{def of T and That}--\eqref{def of fh hard}.
\end{lemma}
\begin{proof}
Let
\begin{align}\nonumber
\Sigma_n := &\;  \sum_{j=0}^{+\infty} \ln \bigg\{ 1 - \frac{\mathsf{T}_{0}(\sigma_{\star})+\hat{\mathsf{T}}_{0}(\sigma_{\star})}{1 + \mathsf{T}_{0}(\sigma_{\star}) + \hat{\mathsf{T}}_{0}(b\rho_{2}^{2b})} \frac{\frac{\sigma_{\star}-b\rho_{1}^{2b}}{b\rho_{2}^{2b}-\sigma_{\star}}(\frac{\rho_{2}}{\rho_{1}})^{-2(\theta_{\star}+j)}}{1+\frac{\sigma_{\star}-b\rho_{1}^{2b}}{b\rho_{2}^{2b}-\sigma_{\star}}(\frac{\rho_{2}}{\rho_{1}})^{-2(\theta_{\star}+j)}} \bigg\} 
	\\\nonumber
& + \sum_{j=0}^{+\infty} \ln \bigg\{ 1 + \frac{\mathsf{T}_{0}(\sigma_{\star})+\hat{\mathsf{T}}_{0}(\sigma_{\star})}{1 - \hat{\mathsf{T}}_{0}(\sigma_{\star}) + \hat{\mathsf{T}}_{0}(b\rho_{2}^{2b})} \frac{\frac{b\rho_{2}^{2b}-\sigma_{\star}}{\sigma_{\star}-b\rho_{1}^{2b}}\big(\frac{\rho_{2}}{\rho_{1}}\big)^{-2(j+1-\theta_{\star})}}{1+\frac{b\rho_{2}^{2b}-\sigma_{\star}}{\sigma_{\star}-b\rho_{1}^{2b}}\big(\frac{\rho_{2}}{\rho_{1}}\big)^{-2(j+1-\theta_{\star})}} \bigg\}.
\end{align}
The desired conclusion is a direct consequence of Lemmas \ref{lemma:S3p1p asymp hard} and \ref{lemma:S3p2p asymp hard} if we can show that
\begin{align}\label{twosumsFn}
\Sigma_n = \theta_\star \ln \mathsf{Q} -\frac{\ln{\mathsf{Q}}}{2} \bigg(
1 - \frac{2\ln(\sigma_{2}/\sigma_{1}) + \ln{\mathsf{Q}}}{2 \ln(\rho_{2}/\rho_{1})}\bigg) 
+  \mathcal{F}_n.
\end{align}
To show (\ref{twosumsFn}), we introduce the short-hand notation $v := \frac{b\rho_{2}^{2b}-\sigma_{\star}}{\sigma_{\star}-b\rho_{1}^{2b}} = \sigma_{2}/\sigma_{1} > 0$ and $w:= \rho_1/\rho_2 \in (0,1)$.
Then we can write
\begin{align*}
 \Sigma_n = & \sum_{j=0}^{+\infty} \ln \bigg\{ 1 - \frac{ \mathsf{Q} - 1}{ \mathsf{Q}} \frac{1}{v w^{-2(\theta_\star + j)} +1} \bigg\} 
+ \sum_{j=0}^{+\infty} \ln \bigg\{ 1 + (\mathsf{Q} - 1) \frac{1}{\frac{w^{-2(j+1-\theta_\star)}}{v} +1} \bigg\}.
\end{align*}
Combining the logarithms in the two sums and then letting $j = \ell-1$, this becomes
\begin{align*}
 \Sigma_n = &\; \sum_{j=0}^{+\infty} \ln \frac{(1 + w^{2j + 2} \frac{v  \mathsf{Q}}{w^{2\theta_\star}} )(1 + w^{2j} \frac{w^{2\theta_\star}}{v  \mathsf{Q}})}{(1 + w^{2j+2}\frac{v}{w^{2\theta_\star}})(1 + w^{2j} \frac{w^{2\theta_\star}}{v} )}
 =  \sum_{\ell =1}^{+\infty} \ln \frac{(1 + w^{2\ell -1} \frac{v  \mathsf{Q}}{w^{2\theta_\star-1}} )(1 + w^{2\ell -1} \frac{w^{2\theta_\star-1}}{v  \mathsf{Q}})}{(1 + w^{2\ell -1}\frac{v}{w^{2\theta_\star-1}})(1 + w^{2\ell -1} \frac{w^{2\theta_\star-1}}{v} )}.
\end{align*}
Using the Jacobi triple product identity
$$\theta(z; \tau) = \prod_{\ell =1}^\infty (1 - e^{2\ell  \pi i \tau})(1 + e^{(2\ell -1)\pi i\tau + 2\pi i z})(1 + e^{(2\ell -1)\pi i\tau - 2\pi i z}),$$
we obtain
\begin{align*}
\Sigma_n =  \ln \frac{\theta(\frac{1}{2\pi i} \ln(\frac{v  \mathsf{Q}}{w^{2\theta_\star - 1}}); \frac{\ln w}{\pi i})}{\theta(\frac{1}{2\pi i} \ln(\frac{v}{w^{2\theta_\star - 1}}); \frac{\ln w}{\pi i})},
\end{align*}
where $\im(\frac{\ln w}{\pi i}) > 0$ because $w \in (0,1)$.
With the help of the Jacobi identity
$$
\theta(z; \tau) = (-i\tau)^{-1/2} e^{-\frac{\pi i z^2}{\tau}} \theta\Big(\frac{z}{\tau}; -\frac{1}{\tau}\Big)$$
this can be written as
\begin{align*}
\Sigma_n & = \theta_\star \ln \mathsf{Q} - \frac{\ln \mathsf{Q}}{2}\bigg(1 + \frac{2\ln v + \ln \mathsf{Q}}{2 \ln w}\bigg) 
 + \ln \frac{\theta(\frac{\ln(v  \mathsf{Q}) }{2 \ln w} - \theta_\star + \frac{1}{2}; -\frac{\pi i}{\ln w})}{ \theta(\frac{\ln v}{2 \ln w}   - \theta_\star + \frac{1}{2}; -\frac{\pi i}{\ln w})}.
\end{align*}
Using the periodicity property $\theta(z + 1; \tau) = \theta(z; \tau)$, we can replace $\theta_\star = j_{\star}- \lfloor j_{\star} \rfloor$ by $j_\star = n\sigma_{\star} - \alpha$ and $+\frac{1}{2}$ by $-\frac{1}{2}$ inside $\theta$. Using also $\theta(-z;\tau)=\theta(z;\tau)$, $v = \sigma_{2}/\sigma_{1}$ and $w = \rho_1/\rho_2$, we arrive at (\ref{twosumsFn}).
\end{proof}

We now turn our attention to $S_2$ and $S_4$. Let $M:=n^{\frac{1}{10}}$. We split $S_{2k}$, $k = 1,2$, as follows:
\begin{align}\label{asymp prelim of S2kpvp hard}
& S_{2k}=S_{2k}^{(1)}+S_{2k}^{(2)}+S_{2k}^{(3)}, \qquad S_{2k}^{(v)} := \sum_{j:\lambda_{j,k}\in I_{v}}  \ln \bigg( 1 + \sum_{\ell=1}^{2m} \omega_{\ell} F_{n,j,\ell} \bigg), \quad v=1,2,3,
\end{align}
where
\begin{align*}
I_{1} = [1-\epsilon,1-\tfrac{M}{\sqrt{n}}), \qquad I_{2} = [1-\tfrac{M}{\sqrt{n}},1+\tfrac{M}{\sqrt{n}}], \qquad I_{3} = (1+\tfrac{M}{\sqrt{n}},1+\epsilon].
\end{align*}
From \eqref{asymp prelim of S2kpvp hard}, we see that the large $n$ asymptotics of $\{S_{2k}^{(v)}\}_{v=1,2,3}$ involve the asymptotics of $\gamma(a,z)$ when $a \to + \infty$, $z \to +\infty$ with $\lambda=\frac{z}{a} \in [1-\epsilon,1+\epsilon]$. These sums can also be rewritten using
\begin{align}\label{sums lambda j hard}
& \sum_{j:\lambda_{j,k}\in I_{3}} = \sum_{j=j_{k,-}}^{g_{k,-}-1}, \qquad \sum_{j:\lambda_{j,k}\in I_{2}} = \sum_{j= g_{k,-}}^{g_{k,+}}, \qquad \sum_{j:\lambda_{j,k}\in I_{1}} = \sum_{j= g_{k,+}+1}^{j_{k,+}},
\end{align}
where $g_{k,-} := \lceil \frac{bn\rho_{k}^{2b}}{1+\frac{M}{\sqrt{n}}}-\alpha \rceil$ and $g_{k,+} := \lfloor \frac{bn\rho_{k}^{2b}}{1-\frac{M}{\sqrt{n}}}-\alpha \rfloor$ for $k = 1,2$. Let us also define
\begin{align*}
& \theta_{k,-}^{(n,M)} := g_{k,-} - \bigg( \frac{bn \rho_{k}^{2b}}{1+\frac{M}{\sqrt{n}}} - \alpha \bigg) = \bigg\lceil \frac{bn \rho_{k}^{2b}}{1+\frac{M}{\sqrt{n}}} - \alpha \bigg\rceil - \bigg( \frac{bn \rho_{k}^{2b}}{1+\frac{M}{\sqrt{n}}} - \alpha \bigg), \\
& \theta_{k,+}^{(n,M)} := \bigg( \frac{bn \rho_{k}^{2b}}{1-\frac{M}{\sqrt{n}}} - \alpha \bigg) - g_{k,+} = \bigg( \frac{bn \rho_{k}^{2b}}{1-\frac{M}{\sqrt{n}}} - \alpha \bigg) - \bigg\lfloor \frac{bn \rho_{k}^{2b}}{1-\frac{M}{\sqrt{n}}} - \alpha \bigg\rfloor.
\end{align*}
Clearly, $\theta_{k,-}^{(n,M)},\theta_{k,+}^{(n,M)} \in [0,1)$.

\begin{lemma}\label{lemma:S2kp3p hard}
For any fixed $x_{1},\dots,x_{2m} \in \mathbb{R}$, there exists $\delta > 0$ such that
\begin{align*}
S_{2}^{(3)} = & \; \Big( b\rho_{1}^{2b}n - j_{1,-} - bM\rho_{1}^{2b}\sqrt{n} + bM^{2}\rho_{1}^{2b}-\alpha+\theta_{1,-}^{(n,M)} - bM^{3}\rho_{1}^{2b}n^{-\frac{1}{2}} \Big) \ln  \Omega + \bigO(M^{4}n^{-1}),
\end{align*}
as $n \to +\infty$ uniformly for $u_{1} \in \{z \in \mathbb{C}: |z-x_{1}|\leq \delta\},\dots,u_{2m} \in \{z \in \mathbb{C}: |z-x_{2m}|\leq \delta\}$.
\end{lemma}
\begin{proof}
For all sufficiently large $n$, and uniformly for $j\in \{j:\lambda_{j,1}\in I_{3}\}$, we have that $\eta_{j,2}, \hat{\eta}_{j,\ell}$, $\ell=m+1,\dots,2m$ are positive and bounded away from $0$, and furthermore
\begin{align*}
& \min_{\ell\in\{1,\dots,m\}} \{\eta_{j,1},\hat{\eta}_{j,\ell}\} \geq \tfrac{M}{\sqrt{n}} + \bigO(\tfrac{1}{\sqrt{n}}), & & - \sqrt{a_{j}/2} \min_{\ell\in\{1,\dots,m\}} \{\eta_{j,1},\hat{\eta}_{j,\ell}\} \leq - \tfrac{M r^{b}}{\sqrt{2}} + \bigO(1).
\end{align*}
Using Lemma \ref{lemma: asymp of gamma for lambda one over sqrt away from 1} and the fact that $M= n^{\frac{1}{10}}$, we infer that 
\begin{align*}
& F_{n,j,\ell} = \frac{1+\bigO(e^{-\frac{a_{j}\hat\eta_{j,\ell}^{2}}{2}})}{1+\bigO(e^{-\frac{a_{j}\eta_{j,2}^{2}}{2}}+e^{-\frac{a_{j}\eta_{j,1}^{2}}{2}})} = 1+\bigO(e^{-cn^{1/5}}), & & \ell=1,\dots,m, \\
& F_{n,j,\ell} = \frac{1+\bigO(e^{-\frac{a_{j}\hat\eta_{j,\ell}^{2}}{2}}+e^{-\frac{a_{j}\eta_{j,2}^{2}}{2}}+e^{-\frac{a_{j}\eta_{j,1}^{2}}{2}})}{1+\bigO(e^{-\frac{a_{j}\eta_{j,2}^{2}}{2}}+e^{-\frac{a_{j}\eta_{j,1}^{2}}{2}})} = 1+\bigO(e^{-cn^{1/5}}), & & \ell=m+1,\dots,2m,
\end{align*}
uniformly for $j \in \{j_{1,-},\dots,g_{1,-}-1\}$. Hence, by \eqref{sums lambda j hard},
\begin{align*}
S_{2}^{(3)} & = \sum_{j=j_{1,-}}^{g_{1,-}-1} \ln \bigg( 1+\sum_{\ell=1}^{2m} \omega_{\ell} \big[1 + \bigO(e^{-cn^{1/5}}) \big] \bigg) = (g_{1,-}-j_{1,-})\ln \Omega + \bigO(n e^{-cn^{1/5}}), \quad \mbox{as } n \to +\infty.
\end{align*}
In the above, the error terms before the second equality are independent of $u_{1},\dots,u_{2m}$, so the expansion (see \cite[Lemma 2.4]{Charlier 2d jumps})
\begin{align*}
g_{1,-}-j_{1,-} = b\rho_{1}^{2b}n - j_{1,-} - bM\rho_{1}^{2b}\sqrt{n} + bM^{2}\rho_{1}^{2b}-\alpha+\theta_{1,-}^{(n,M)} - bM^{3}\rho_{1}^{2b}n^{-\frac{1}{2}} + \bigO(M^{4}n^{-1})
\end{align*}
yields the claim.
\end{proof}

\begin{lemma}\label{lemma:S4p1p hard}
For any fixed $x_{1},\dots,x_{2m} \in \mathbb{R}$, there exists $\delta > 0$ such that
\begin{align*}
S_{4}^{(1)} = & \; \bigO(n^{-100}),
\end{align*}
as $n \to +\infty$ uniformly for $u_{1} \in \{z \in \mathbb{C}: |z-x_{1}|\leq \delta\},\dots,u_{2m} \in \{z \in \mathbb{C}: |z-x_{2m}|\leq \delta\}$.
\end{lemma}
\begin{proof}
For all sufficiently large $n$, and uniformly for $j\in \{j:\lambda_{j,2}\in I_{1}\}$, we have that $\eta_{j,1}, \hat{\eta}_{j,\ell}$, $\ell=1,\dots,m$ are negative and bounded away from $0$, and furthermore
\begin{align*}
& \max_{\ell\in\{m+1,\dots,2m\}} \{\eta_{j,2},\hat{\eta}_{j,\ell}\} \leq -\tfrac{M}{\sqrt{n}} + \bigO(\tfrac{1}{\sqrt{n}}), & & - \sqrt{a_{j}/2} \max_{\ell\in\{m+1,\dots,2m\}} \{\eta_{j,2},\hat{\eta}_{j,\ell}\} \geq \tfrac{M r^{b}}{\sqrt{2}} + \bigO(1).
\end{align*}
Using Lemma \ref{lemma: asymp of gamma for lambda one over sqrt away from 1} and the fact that $M= n^{\frac{1}{10}}$, we infer that
\begin{align*}
& F_{n,j,\ell} = \bigO(e^{-cn^{1/5}}) = \bigO(n^{-101}), & & \ell=1,\dots,2m, 
\end{align*}
uniformly for $j \in \{g_{2,+}+1,\dots,j_{2,+}\}$. Hence, by \eqref{sums lambda j hard},
\begin{align*}
S_{4}^{(1)} & = \sum_{j=g_{2,+}+1}^{j_{2,+}} \ln \bigg( 1+\sum_{\ell=1}^{2m} \omega_{\ell} \bigO(n^{-101}) \bigg) = \bigO(n^{-100}), \quad \mbox{as } n \to +\infty,
\end{align*}
and the claim follows.
\end{proof}

\begin{lemma}\label{lemma:S2p1p hard}
For any fixed $x_{1},\dots,x_{2m} \in \mathbb{R}$, there exists $\delta > 0$ such that 
\begin{align*}
& S_{2}^{(1)} = D_{1}^{(\epsilon)} n + D_{2}^{(M)} \sqrt{n} + D_{3} \ln n + D_{4}^{(n,\epsilon,M)} + \frac{D_{5}^{(n,M)}}{\sqrt{n}} + \bigO\bigg( \frac{M^{4}}{n} + \frac{1}{\sqrt{n} M} + \frac{1}{M^{6}} + \frac{\sqrt{n}}{M^{11}} \bigg),
\end{align*}
as $n \to +\infty$ uniformly for $u_{1} \in \{z \in \mathbb{C}: |z-x_{1}|\leq \delta\},\dots,u_{2m} \in \{z \in \mathbb{C}: |z-x_{2m}|\leq \delta\}$, where
\begin{align*}
& D_{1}^{(\epsilon)} = \int_{b\rho_{1}^{2b}}^{\frac{b\rho_{1}^{2b}}{1-\epsilon}}f_{1}(x)dx, \qquad D_{2}^{(M)} = -b\rho_{1}^{2b} f_{1}(b\rho_{1}^{2b})M, \qquad D_{3} =  \frac{-b\rho_{1}^{2b} \mathsf{T}_{1}(b\rho_{1}^{2b})}{2(1+\mathsf{T}_{0}(b\rho_{1}^{2b})+\hat{\mathsf{T}}_{0}(b\rho_{2}^{2b}))}, \\
& D_{4}^{(n,\epsilon,M)} = -b\rho_{1}^{2b} M^{2} \Big( f_{1}(b\rho_{1}^{2b}) + \frac{b\rho_{1}^{2b}}{2}f_{1}'(b\rho_{1}^{2b}) \Big) - \frac{b\rho_{1}^{2b} \mathsf{T}_{1}(b\rho_{1}^{2b})}{1+\mathsf{T}_{0}(b\rho_{1}^{2b})+\hat{\mathsf{T}}_{0}(b\rho_{2}^{2b})} \ln \bigg( \frac{\epsilon}{M(1-\epsilon)} \bigg) \\
& + \int_{b\rho_{1}^{2b}}^{\frac{b\rho_{1}^{2b}}{1-\epsilon}} \bigg\{ f(x) + \frac{b\rho_{1}^{2b}\mathsf{T}_{1}(b\rho_{1}^{2b})}{(1+\mathsf{T}_{0}(b\rho_{1}^{2b})+\hat{\mathsf{T}}_{0}(b\rho_{2}^{2b}))(x-b\rho_{1}^{2b})} \bigg\}dx + \bigg( \alpha - \frac{1}{2} + \theta_{1,+}^{(n,M)} \bigg)f_{1}(b \rho_{1}^{2b}) \\
& + \bigg( \frac{1}{2} - \alpha - \theta_{1,+}^{(n,\epsilon)} \bigg) f_{1}\bigg( \frac{b\rho_{1}^{2b}}{1-\epsilon} \bigg) + \frac{b \mathsf{T}_{1}(b\rho_{1}^{2b})}{M^{2}(1+\mathsf{T}_{0}(b\rho_{1}^{2b})+\hat{\mathsf{T}}_{0}(b\rho_{2}^{2b}))} + \frac{-5b \mathsf{T}_{1}(b\rho_{1}^{2b})}{2\rho_{1}^{2b}M^{4}(1+\mathsf{T}_{0}(b\rho_{1}^{2b})+\hat{\mathsf{T}}_{0}(b\rho_{2}^{2b}))}, \\
& D_{5}^{(n,M)} = -M^{3}b\rho_{1}^{2b} \bigg( f_{1}(b\rho_{1}^{2b}) + b\rho_{1}^{2b} f_{1}'(b\rho_{1}^{2b}) + \frac{(b\rho_{1}^{2b})^{2}}{6}f_{1}''(b\rho_{1}^{2b}) \bigg) + M b\rho_{1}^{2b} f_{1}'(b\rho_{1}^{2b}) \bigg( \alpha-\frac{1}{2}+\theta_{1,+}^{(n,M)} \bigg) \\
& + M \bigg( \frac{(b+\alpha)\rho_{1}^{2b}\mathsf{T}_{1}(b\rho_{1}^{2b})}{1+\mathsf{T}_{0}(b\rho_{1}^{2b})+\hat{\mathsf{T}}_{0}(b\rho_{2}^{2b})} - \frac{b\rho_{1}^{4b} \mathsf{T}_{2}(b\rho_{1}^{2b})}{2(1+\mathsf{T}_{0}(b\rho_{1}^{2b})+\hat{\mathsf{T}}_{0}(b\rho_{2}^{2b}))} + \frac{b\rho_{1}^{4b} \mathsf{T}_{1}(b\rho_{1}^{2b})^{2}}{(1+\mathsf{T}_{0}(b\rho_{1}^{2b})+\hat{\mathsf{T}}_{0}(b\rho_{2}^{2b}))^{2}} \bigg),
\end{align*}
where $f_{1}$ and $f$ are as in the statement of Lemma \ref{lemma:S3 asymp hard}.
\end{lemma}
\begin{proof}
Using (\ref{Fnjellexpression}) and Lemma \ref{lemma: uniform}, we obtain
\begin{align}
S_{2}^{(1)} & = \hspace{-0.15cm} \sum_{j= g_{1,+}+1}^{j_{1,+}} \hspace{-0.15cm} \ln  \hspace{-0.05cm} \bigg( \hspace{-0.05cm} 1 \hspace{-0.05cm} + \hspace{-0.05cm} \sum_{\ell=1}^{m} \omega_{\ell} \bigg\{ \frac{ \frac{1}{2}\mathrm{erfc}\big(-\hat{\eta}_{j,\ell} \sqrt{a_{j}/2}\big) - R_{a_{j}}(\hat\eta_{j,\ell}) }{ \frac{1}{2}\mathrm{erfc}\big(-\eta_{j,1} \sqrt{a_{j}/2}\big) - R_{a_{j}}(\eta_{j,1}) } + \bigO(e^{-cn^{\frac{1}{5}}})\bigg\} + \hspace{-0.15cm} \sum_{\ell=m+1}^{2m} \hspace{-0.15cm} \omega_{\ell}(1+\bigO(e^{-cn^{\frac{1}{5}}})) \bigg) \nonumber \\
& = \sum_{j= g_{1,+}+1}^{j_{1,+}} \ln \bigg( 1+\sum_{\ell=1}^{m} \omega_{\ell} \frac{ \frac{1}{2}\mathrm{erfc}\big(-\hat{\eta}_{j,\ell} \sqrt{a_{j}/2}\big) - R_{a_{j}}(\hat\eta_{j,\ell}) }{ \frac{1}{2}\mathrm{erfc}\big(-\eta_{j,1} \sqrt{a_{j}/2}\big) - R_{a_{j}}(\eta_{j,1}) } + \hat{\mathsf{T}}_{0}(b\rho_{2}^{2b}) \bigg) +\bigO(e^{-cn^{1/5}}). \label{lol4}
\end{align}
The computation and the resulting formulas appearing in the rest of this proof are the same as in the proof of \cite[Lemma 2.6]{ACCL2022}, except that $\rho$ in \cite{ACCL2022} corresponds to $\rho_1$ here and the extra term $\hat{\mathsf{T}}_{0}(b\rho_{2}^{2b})$ in the above equation implies that $1+ \mathsf{T}_{0}(x)$ in \cite[Lemma 2.6]{ACCL2022} is replaced here by $1+ \mathsf{T}_{0}(x) +\hat{\mathsf{T}}_{0}(b\rho_{2}^{2b})$. We nevertheless provide the details here for completeness. Using \eqref{asymp of Ra}, we find
\begin{align*}
S_{2}^{(1)} = &\; \sum_{j= g_{1,+}+1}^{j_{1,+}} \bigg(  f_{1}(j/n)+\frac{1}{n}f(j/n) + \frac{1}{n^{2}} \frac{2b^{3}\rho_{1}^{4b} \mathsf{T}_{1}(j/n)}{(1+\mathsf{T}_{0}(j/n)+\hat{\mathsf{T}}_{0}(b\rho_{2}^{2b}))(j/n-b\rho_{1}^{2b})^{3}} \\
& + \frac{1}{n^{3}} \frac{-10b^{5}\rho_{1}^{6b} \mathsf{T}_{1}(j/n)}{(1+\mathsf{T}_{0}(j/n)+\hat{\mathsf{T}}_{0}(b\rho_{2}^{2b}))(j/n-b\rho_{1}^{2b})^{5}} \\
& + \bigO\bigg(\frac{n^{-2}}{(j/n-b\rho_{1}^{2b})^{2}} + \frac{n^{-3}}{(j/n-b\rho_{1}^{2b})^{4}} + \frac{n^{-4}}{(j/n-b\rho_{1}^{2b})^{7}} + \frac{n^{-6}}{(j/n-b\rho_{1}^{2b})^{12}} \bigg)\bigg).
\end{align*}
Note that
\begin{align*}
& \sum_{j= g_{1,+}+1}^{j_{1,+}}\bigO\bigg(\frac{n^{-2}}{(j/n-b\rho_{1}^{2b})^{2}} + \frac{n^{-3}}{(j/n-b\rho_{1}^{2b})^{4}} + \frac{n^{-4}}{(j/n-b\rho_{1}^{2b})^{7}} + \frac{n^{-6}}{(j/n-b\rho_{1}^{2b})^{12}} \bigg) \\
& = \bigO\bigg( \frac{1}{M\sqrt{n} } + \frac{1}{M^{3}\sqrt{n}} + \frac{1}{M^{6}} + \frac{\sqrt{n}}{M^{11}} \bigg).
\end{align*}
Also, using Lemma \ref{lemma:Riemann sum NEW} (with $A=\frac{b\rho_{1}^{2b}}{1-\frac{M}{\sqrt{n}}}$, $a_{0}=1-\alpha-\theta_{1,+}^{(n,M)}$, $B=\frac{b\rho_{1}^{2b}}{1-\epsilon}$ and $b_{0} = -\alpha-\theta_{1,+}^{(n,\epsilon)}$), we get
\begin{align*}
& \sum_{j= g_{1,+}+1}^{j_{1,+}} f_{1}(j/n) = n \int_{\frac{b\rho_{1}^{2b}}{1-\frac{M}{\sqrt{n}}}}^{\frac{b\rho_{1}^{2b}}{1-\epsilon}}f_{1}(x)dx + \big(\alpha-\tfrac{1}{2}+\theta_{1,+}^{(n,M)}\big)f_{1}(\tfrac{b\rho_{1}^{2b}}{1-\frac{M}{\sqrt{n}}})+\big(\tfrac{1}{2}-\alpha-\theta_{1,+}^{(n,\epsilon)}\big)f_{1}(\tfrac{b\rho_{1}^{2b}}{1-\epsilon})+\bigO(n^{-1}), \\
& \frac{1}{n}\sum_{j= g_{1,+}+1}^{j_{1,+}} f(j/n) = \int_{\frac{b\rho_{1}^{2b}}{1-\frac{M}{\sqrt{n}}}}^{\frac{b\rho_{1}^{2b}}{1-\epsilon}}f(x)dx + \bigO\bigg(\frac{1}{M\sqrt{n}}\bigg), \\
& \frac{1}{n^{2}}\sum_{j= g_{1,+}+1}^{j_{1,+}}  \frac{2b^{3}\rho_{1}^{4b} \mathsf{T}_{1}(j/n)(j/n-b\rho^{2b})^{-3}}{1+\mathsf{T}_{0}(j/n)+\hat{\mathsf{T}}_{0}(b\rho_{2}^{2b})} = \frac{1}{n} \int_{\frac{b\rho_{1}^{2b}}{1-\frac{M}{\sqrt{n}}}}^{\frac{b\rho_{1}^{2b}}{1-\epsilon}} \frac{2b^{3}\rho_{1}^{4b} \mathsf{T}_{1}(x)(x-b\rho_{1}^{2b})^{-3}}{1+\mathsf{T}_{0}(x)+\hat{\mathsf{T}}_{0}(b\rho_{2}^{2b})} dx + \bigO\bigg(\frac{1}{M^{3}\sqrt{n}}\bigg), \\
& \frac{1}{n^{3}}\sum_{j= g_{1,+}+1}^{j_{1,+}} \frac{-10b^{5}\rho_{1}^{6b} \mathsf{T}_{1}(j/n) (j/n-b\rho^{2b})^{-5}}{1+\mathsf{T}_{0}(j/n)+\hat{\mathsf{T}}_{0}(b\rho_{2}^{2b})} = \frac{1}{n^{2}} \int_{\frac{b\rho_{1}^{2b}}{1-\frac{M}{\sqrt{n}}}}^{\frac{b\rho_{1}^{2b}}{1-\epsilon}} \frac{-10b^{5}\rho_{1}^{6b} \mathsf{T}_{1}(x)(x-b\rho^{2b})^{-5}}{1+\mathsf{T}_{0}(x)+\hat{\mathsf{T}}_{0}(b\rho_{2}^{2b})} dx + \bigO\bigg(\frac{1}{M^{5}\sqrt{n}}\bigg).
\end{align*}
Furthermore, by a direct analysis,
\begin{align*}
& n \int_{\frac{b\rho_{1}^{2b}}{1-\frac{M}{\sqrt{n}}}}^{\frac{b\rho_{1}^{2b}}{1-\epsilon}}f_{1}(x)dx = n \int_{b\rho_{1}^{2b}}^{\frac{b\rho_{1}^{2b}}{1-\epsilon}}f_{1}(x)dx - b\rho_{1}^{2b}f_{1}(b \rho_{1}^{2b})M\sqrt{n} - M^{2} b \rho_{1}^{2b}\big( f_{1}(b\rho_{1}^{2b})+ \frac{b \rho_{1}^{2b}}{2} f_{1}'(b \rho_{1}^{2b}) \big) \\
& \hspace{3cm} - \frac{M^{3}}{\sqrt{n}} b\rho_{1}^{2b} \Big( f_{1}(b\rho_{1}^{2b}) + b\rho_{1}^{2b}f_{1}'(b\rho_{1}^{2b}) + \frac{(b\rho_{1}^{2b})^{2}}{6}f_{1}''(b\rho_{1}^{2b})  \Big) + \bigO\bigg( \frac{M^{4}}{n} \bigg), \\
& f_{1}(\tfrac{b\rho_{1}^{2b}}{1-\frac{M}{\sqrt{n}}}) = f_{1}(b \rho_{1}^{2b}) + \frac{M}{\sqrt{n}}b\rho_{1}^{2b}f_{1}'(b \rho_{1}^{2b}) + \bigO\bigg( \frac{M^{2}}{n} \bigg), \\
& \int_{\frac{b\rho_{1}^{2b}}{1-\frac{M}{\sqrt{n}}}}^{\frac{b\rho_{1}^{2b}}{1-\epsilon}}f(x)dx = \int_{b\rho_{1}^{2b}}^{\frac{b\rho_{1}^{2b}}{1-\epsilon}} \bigg\{ f(x) + \frac{b\rho_{1}^{2b} \mathsf{T}_{1}(b\rho_{1}^{2b})}{(1+\mathsf{T}_{0}(b\rho_{1}^{2b})+\hat{\mathsf{T}}_{0}(b\rho_{2}^{2b}))(x-b\rho_{1}^{2b})} \bigg\}dx - \frac{b\rho_{1}^{2b} \mathsf{T}_{1}(b\rho_{1}^{2b})}{2(1+\mathsf{T}_{0}(b\rho_{1}^{2b})+\hat{\mathsf{T}}_{0}(b\rho_{2}^{2b}))} \ln n \\
& - \frac{b\rho_{1}^{2b} \mathsf{T}_{1}(b\rho_{1}^{2b})}{1+\mathsf{T}_{0}(b\rho_{1}^{2b})+\hat{\mathsf{T}}_{0}(b\rho_{2}^{2b})} \ln \frac{\epsilon}{M(1-\epsilon)} + \frac{M}{\sqrt{n}} \bigg\{ \frac{(b+\alpha) \rho_{1}^{2b} \mathsf{T}_{1}(b\rho_{1}^{2b})}{1+\mathsf{T}_{0}(b\rho_{1}^{2b})+\hat{\mathsf{T}}_{0}(b\rho_{2}^{2b})} - \frac{b \rho_{1}^{4b} \mathsf{T}_{2}(b\rho_{1}^{2b})}{2(1+\mathsf{T}_{0}(b\rho_{1}^{2b})+\hat{\mathsf{T}}_{0}(b\rho_{2}^{2b}))} \\
& + \frac{b\rho_{1}^{4b} \mathsf{T}_{1}(b\rho_{1}^{2b})^{2}}{(1+\mathsf{T}_{0}(b\rho_{1}^{2b})+\hat{\mathsf{T}}_{0}(b\rho_{2}^{2b}))^{2}} \bigg\} + \bigO\bigg( \frac{M^{2}}{n} \bigg), \\
& \frac{1}{n} \int_{\frac{b\rho_{1}^{2b}}{1-\frac{M}{\sqrt{n}}}}^{\frac{b\rho_{1}^{2b}}{1-\epsilon}} \frac{2b^{3}\rho_{1}^{4b} \mathsf{T}_{1}(x)}{(1+\mathsf{T}_{0}(x)+\hat{\mathsf{T}}_{0}(b\rho_{2}^{2b}))(x-b\rho_{1}^{2b})^{3}} dx = \frac{b \mathsf{T}_{1}(b \rho_{1}^{2b})}{M^{2}(1+\mathsf{T}_{0}(b \rho_{1}^{2b})+\hat{\mathsf{T}}_{0}(b \rho_{2}^{2b}))} + \bigO\bigg(\frac{1}{M \sqrt{n}}\bigg), \\
& \frac{1}{n^{2}} \int_{\frac{b\rho_{1}^{2b}}{1-\frac{M}{\sqrt{n}}}}^{\frac{b\rho_{1}^{2b}}{1-\epsilon}} \frac{-10b^{5}\rho_{1}^{6b} \mathsf{T}_{1}(x)}{(1+\mathsf{T}_{0}(x)+\hat{\mathsf{T}}_{0}(b\rho_{2}^{2b}))(x-b\rho_{1}^{2b})^{5}} dx = \frac{-5 b \mathsf{T}_{1}(b \rho_{1}^{2b})}{2\rho_{1}^{2b}M^{4}(1+\mathsf{T}_{0}(b \rho_{1}^{2b})+\hat{\mathsf{T}}_{0}(b \rho_{2}^{2b}))} + \bigO\bigg(\frac{1}{M^{3} \sqrt{n}}\bigg).
\end{align*}
The claim now follows after a computation.
\end{proof}

\begin{lemma}\label{lemma:S4p3p hard}
For any fixed $x_{1},\dots,x_{2m} \in \mathbb{R}$, there exists $\delta > 0$ such that 
\begin{align*}
& S_{4}^{(3)} = \hat{D}_{1}^{(\epsilon)} n + \hat{D}_{2}^{(M)} \sqrt{n} + \hat{D}_{3} \ln n + \hat{D}_{4}^{(n,\epsilon,M)} + \frac{\hat{D}_{5}^{(n,M)}}{\sqrt{n}} + \bigO\bigg( \frac{M^{4}}{n} + \frac{1}{\sqrt{n} M} + \frac{1}{M^{6}}  + \frac{\sqrt{n}}{M^{11}} \bigg),
\end{align*}
as $n \to +\infty$ uniformly for $u_{1} \in \{z \in \mathbb{C}: |z-x_{1}|\leq \delta\},\dots,u_{2m} \in \{z \in \mathbb{C}: |z-x_{2m}|\leq \delta\}$, where
\begin{align*}
& \hat{D}_{1}^{(\epsilon)} = \int_{\frac{b\rho_{2}^{2b}}{1+\epsilon}}^{b\rho_{2}^{2b}}\hat{f}_{1}(x)dx, \qquad \hat{D}_{2}^{(M)} = 0, \qquad \hat{D}_{3} =  \frac{b\rho_{2}^{2b} \hat{\mathsf{T}}_{1}(b\rho_{2}^{2b})}{2}, \\
& \hat{D}_{4}^{(n,\epsilon,M)} = M^{2} \frac{b^{2}\rho_{2}^{4b}}{2}\hat{f}_{1}'(b\rho_{2}^{2b}) + b\rho_{2}^{2b} \hat{\mathsf{T}}_{1}(b\rho_{2}^{2b}) \ln \bigg( \frac{\epsilon}{M(1+\epsilon)} \bigg) + \int_{\frac{b\rho_{2}^{2b}}{1+\epsilon}}^{b\rho_{2}^{2b}} \bigg\{ \hat{f}(x) - \frac{b\rho_{2}^{2b}\hat{\mathsf{T}}_{1}(b\rho_{2}^{2b})}{b\rho_{2}^{2b}-x} \bigg\}dx  \\
& + \bigg( \frac{1}{2} + \alpha - \theta_{2,-}^{(n,\epsilon)} \bigg)\hat{f}_{1}\bigg(\frac{b \rho_{2}^{2b}}{1+\epsilon}\bigg) - \frac{b \hat{\mathsf{T}}_{1}(b\rho_{2}^{2b})}{M^{2}} + \frac{5b \hat{\mathsf{T}}_{1}(b\rho_{2}^{2b})}{2\rho_{2}^{2b}M^{4}}, \\
& \hat{D}_{5}^{(n,M)} = -M^{3}b^{2}\rho_{2}^{4b} \bigg( \hat{f}_{1}'(b\rho_{2}^{2b}) + \frac{b\rho_{2}^{2b}}{6}\hat{f}_{1}''(b\rho_{2}^{2b}) \bigg) + M b\rho_{2}^{2b} \hat{f}_{1}'(b\rho_{2}^{2b}) \bigg( \frac{1}{2}+\alpha-\theta_{2,-}^{(n,M)} \bigg) \\
& + M \bigg( (b+\alpha)\rho_{2}^{2b}\hat{\mathsf{T}}_{1}(b\rho_{2}^{2b}) + \frac{b\rho_{2}^{4b}}{2} \hat{\mathsf{T}}_{2}(b\rho_{2}^{2b}) + b\rho_{2}^{4b} \hat{\mathsf{T}}_{1}(b\rho_{2}^{2b})^{2} \bigg),
\end{align*}
where $\hat{f}_{1}$ and $\hat{f}$ are as in the statement of Lemma \ref{lemma:S3 asymp hard}.
\end{lemma}
\begin{proof}
Using Lemma \ref{lemma: uniform}, we obtain
\begin{align}
S_{4}^{(3)} & = \sum_{j= j_{2,-}}^{g_{2,-}-1} \ln \bigg( 1+\sum_{\ell=1}^{m} \omega_{\ell}  \bigO(e^{-cn}) \nonumber \\
& + \sum_{\ell=m+1}^{2m} \omega_{\ell}\bigg\{\frac{ \frac{1}{2}\mathrm{erfc}\big(-\hat{\eta}_{j,\ell} \sqrt{a_{j}/2}\big) - R_{a_{j}}(\hat\eta_{j,\ell}) - \big[\frac{1}{2}\mathrm{erfc}\big(-\eta_{j,2} \sqrt{a_{j}/2}\big) - R_{a_{j}}(\eta_{j,2})\big] }{ 1 - \big[\frac{1}{2}\mathrm{erfc}\big(-\eta_{j,2} \sqrt{a_{j}/2}\big) - R_{a_{j}}(\eta_{j,2})\big] }+\bigO(e^{-cn})) \bigg\} \bigg) \nonumber \\
& = \bigO(e^{-cn}) + \sum_{j= j_{2,-}}^{g_{2,-}-1} \ln \bigg( 1 \nonumber \\
& +\sum_{\ell=m+1}^{2m} \omega_{\ell} \frac{ \frac{1}{2}\mathrm{erfc}\big(-\hat{\eta}_{j,\ell} \sqrt{a_{j}/2}\big) - R_{a_{j}}(\hat\eta_{j,\ell}) - \big[\frac{1}{2}\mathrm{erfc}\big(-\eta_{j,2} \sqrt{a_{j}/2}\big) - R_{a_{j}}(\eta_{j,2})\big] }{ 1 - \big[\frac{1}{2}\mathrm{erfc}\big(-\eta_{j,2} \sqrt{a_{j}/2}\big) - R_{a_{j}}(\eta_{j,2})\big] }  \bigg). \label{lol5}
\end{align}
Using then \eqref{asymp of Ra}, we find
\begin{align*}
& S_{4}^{(3)} = \sum_{j= j_{2,-}}^{g_{2,-}-1} \bigg(  \hat{f}_{1}(j/n)+\frac{1}{n}\hat{f}(j/n) + \frac{1}{n^{2}} \frac{-2b^{3}\rho_{2}^{4b} \hat{\mathsf{T}}_{1}(j/n)}{(1-\hat{\mathsf{T}}_{0}(j/n)+\hat{\mathsf{T}}_{0}(b\rho_{2}^{2b}))(b\rho_{2}^{2b}-j/n)^{3}} \\
& + \frac{1}{n^{3}} \frac{10b^{5}\rho_{2}^{6b} \hat{\mathsf{T}}_{1}(j/n)(b\rho_{2}^{2b}-j/n)^{-5}}{1-\hat{\mathsf{T}}_{0}(j/n)+\hat{\mathsf{T}}_{0}(b\rho_{2}^{2b})} \\
& + \bigO\bigg(\frac{n^{-2}}{(b\rho_{2}^{2b}-j/n)^{2}} + \frac{n^{-3}}{(b\rho_{2}^{2b}-j/n)^{4}} + \frac{n^{-4}}{(b\rho_{2}^{2b}-j/n)^{7}} + \frac{n^{-6}}{(b\rho_{2}^{2b}-j/n)^{12}} \bigg)\bigg).
\end{align*}
Note that
\begin{align*}
& \sum_{j= j_{2,-}}^{g_{2,-}-1} \bigO\bigg(\frac{n^{-2}}{(b\rho_{2}^{2b}-j/n)^{2}} + \frac{n^{-3}}{(b\rho_{2}^{2b}-j/n)^{4}} + \frac{n^{-4}}{(b\rho_{2}^{2b}-j/n)^{7}} + \frac{n^{-6}}{(b\rho_{2}^{2b}-j/n)^{12}} \bigg) \\
& = \bigO\bigg( \frac{1}{M\sqrt{n} } + \frac{1}{M^{3}\sqrt{n}} + \frac{1}{M^{6}} + \frac{\sqrt{n}}{M^{11}} \bigg).
\end{align*}
Also, using Lemma \ref{lemma:Riemann sum NEW} (with $A=\frac{b\rho_{2}^{2b}}{1+\epsilon}$, $a_{0}=\theta_{2,-}^{(n,\epsilon)}-\alpha$, $B=\frac{b\rho_{2}^{2b}}{1+\frac{M}{\sqrt{n}}}$ and $b_{0} = \theta_{2,-}^{(n,M)}-1-\alpha$), we get
\begin{align*}
& \sum_{j= j_{2,-}}^{g_{2,-}-1} \hat{f}_{1}(j/n) = n \int_{\frac{b\rho_{2}^{2b}}{1+\epsilon}}^{\frac{b\rho_{2}^{2b}}{1+\frac{M}{\sqrt{n}}}}\hat{f}_{1}(x)dx +(\tfrac{1}{2}+\alpha-\theta_{2,-}^{(n,\epsilon)})\hat{f}_{1}(\tfrac{b\rho_{2}^{2b}}{1+\epsilon}) + \big(\theta_{2,-}^{(n,M)}-\tfrac{1}{2}-\alpha\big)\hat{f}_{1}(\tfrac{b\rho_{2}^{2b}}{1+\frac{M}{\sqrt{n}}})+\bigO(n^{-1}), \\
& \frac{1}{n}\sum_{j= j_{2,-}}^{g_{2,-}-1} \hat{f}(j/n) = \int_{\frac{b\rho_{2}^{2b}}{1+\epsilon}}^{\frac{b\rho_{2}^{2b}}{1+\frac{M}{\sqrt{n}}}} \hat{f}(x)dx + \bigO\bigg(\frac{1}{M\sqrt{n}}\bigg), \\
& \frac{1}{n^{2}}\sum_{j= j_{2,-}}^{g_{2,-}-1}  \frac{-2b^{3}\rho_{2}^{4b} \hat{\mathsf{T}}_{1}(j/n)(b\rho_{2}^{2b}-j/n)^{-3}}{1-\hat{\mathsf{T}}_{0}(j/n)+\hat{\mathsf{T}}_{0}(b\rho_{2}^{2b})} = \frac{1}{n} \int_{\frac{b\rho_{2}^{2b}}{1+\epsilon}}^{\frac{b\rho_{2}^{2b}}{1+\frac{M}{\sqrt{n}}}} \frac{-2b^{3}\rho_{2}^{4b} \hat{\mathsf{T}}_{1}(x)(b\rho_{2}^{2b}-x)^{-3}}{1-\hat{\mathsf{T}}_{0}(x)+\hat{\mathsf{T}}_{0}(b\rho_{2}^{2b})} dx + \bigO\bigg(\frac{1}{M^{3}\sqrt{n}}\bigg), \\
& \frac{1}{n^{3}}\sum_{j= j_{2,-}}^{g_{2,-}-1} \frac{10b^{5}\rho_{2}^{6b} \hat{\mathsf{T}}_{1}(j/n)(b\rho_{2}^{2b}-j/n)^{-5}}{1-\hat{\mathsf{T}}_{0}(j/n)+\hat{\mathsf{T}}_{0}(b\rho_{2}^{2b})} = \frac{1}{n^{2}} \int_{\frac{b\rho_{2}^{2b}}{1+\epsilon}}^{\frac{b\rho_{2}^{2b}}{1+\frac{M}{\sqrt{n}}}} \frac{10b^{5}\rho_{2}^{6b} \hat{\mathsf{T}}_{1}(x)(b\rho_{2}^{2b}-x)^{-5}}{1-\hat{\mathsf{T}}_{0}(x)+\hat{\mathsf{T}}_{0}(b\rho_{2}^{2b})} dx + \bigO\bigg(\frac{1}{M^{5}\sqrt{n}}\bigg).
\end{align*}
Furthermore, by a direct analysis,
\begin{align*}
& n \int_{\frac{b\rho_{2}^{2b}}{1+\epsilon}}^{\frac{b\rho_{2}^{2b}}{1+\frac{M}{\sqrt{n}}}}\hat{f}_{1}(x)dx = n \int_{\frac{b\rho_{2}^{2b}}{1+\epsilon}}^{b\rho_{2}^{2b}}\hat{f}_{1}(x)dx + M^{2} \frac{b^{2} \rho_{2}^{4b}}{2} \hat{f}_{1}'(b \rho_{2}^{2b}) \big) \\
& \hspace{3cm} - \frac{M^{3}}{\sqrt{n}} b^{2} \rho_{2}^{4b} \Big( \hat{f}_{1}'(b\rho_{2}^{2b}) + \frac{b\rho_{2}^{2b}}{6}\hat{f}_{1}''(b\rho_{2}^{2b})  \Big) + \bigO\bigg( \frac{M^{4}}{n} \bigg), \\
& \hat{f}_{1}(\tfrac{b\rho_{2}^{2b}}{1+\frac{M}{\sqrt{n}}}) = \hat{f}_{1}(b \rho_{2}^{2b}) - \frac{M}{\sqrt{n}}b\rho_{2}^{2b}\hat{f}_{1}'(b \rho_{2}^{2b}) + \bigO\bigg( \frac{M^{2}}{n} \bigg), \\
& \int_{\frac{b\rho_{2}^{2b}}{1+\epsilon}}^{\frac{b\rho_{2}^{2b}}{1+\frac{M}{\sqrt{n}}}} \hat{f}(x)dx = \int_{\frac{b\rho_{2}^{2b}}{1+\epsilon}}^{b\rho_{2}^{2b}} \bigg\{ \hat{f}(x) - \frac{b\rho_{2}^{2b} \hat{\mathsf{T}}_{1}(b\rho_{2}^{2b})}{b\rho_{2}^{2b}-x} \bigg\}dx + \frac{b\rho_{2}^{2b} \hat{\mathsf{T}}_{1}(b\rho_{2}^{2b})}{2} \ln n - b\rho_{2}^{2b} \hat{\mathsf{T}}_{2}(b\rho_{2}^{2b}) \ln \frac{M(1+\epsilon)}{\epsilon} \\
&  + \frac{M}{\sqrt{n}} \bigg\{ (b+\alpha) \rho_{2}^{2b} \hat{\mathsf{T}}_{1}(b\rho_{2}^{2b}) + \frac{b \rho_{2}^{4b}}{2}\hat{\mathsf{T}}_{2}(b\rho_{2}^{2b}) + b\rho_{2}^{4b}  \hat{\mathsf{T}}_{1}(b\rho_{2}^{2b})^{2} \bigg\} + \bigO\bigg( \frac{M^{2}}{n} \bigg), \\
& \frac{1}{n} \int_{\frac{b\rho_{2}^{2b}}{1+\epsilon}}^{\frac{b\rho_{2}^{2b}}{1+\frac{M}{\sqrt{n}}}} \frac{-2b^{3}\rho_{2}^{4b} \hat{\mathsf{T}}_{1}(x)(b\rho_{2}^{2b}-x)^{-3}}{1-\hat{\mathsf{T}}_{0}(x)+\hat{\mathsf{T}}_{0}(b\rho_{2}^{2b})} dx = -\frac{b}{M^{2}}\hat{\mathsf{T}}_{1}(b \rho_{2}^{2b}) + \bigO\bigg(\frac{1}{M \sqrt{n}}\bigg), \\
& \frac{1}{n^{2}} \int_{\frac{b\rho_{2}^{2b}}{1+\epsilon}}^{\frac{b\rho_{2}^{2b}}{1+\frac{M}{\sqrt{n}}}} \frac{10b^{5}\rho_{2}^{6b} \hat{\mathsf{T}}_{1}(x)(b\rho_{2}^{2b}-x)^{-5}}{1-\hat{\mathsf{T}}_{0}(x)+\hat{\mathsf{T}}_{0}(b\rho_{2}^{2b})} dx = \frac{5 b \hat{\mathsf{T}}_{1}(b \rho_{2}^{2b})}{2\rho_{2}^{2b}M^{4}} + \bigO\bigg(\frac{1}{M^{3} \sqrt{n}}\bigg).
\end{align*}
The claim now follows after a computation.
\end{proof}

For $\ell \in \{1,\dots,2m\}$ and $j \in \{1,\dots,n\}$, we define $\hat{M}_{j,\ell} := \sqrt{n}(\hat{\lambda}_{j,\ell}-1)$, $M_{j,1} := \sqrt{n}(\lambda_{j,1}-1)$ and $M_{j,2} := \sqrt{n}(\lambda_{j,2}-1)$. For the large $n$ asymptotics of $\smash{S_{2}^{(2)}}$ we will need the following lemma.
\begin{lemma}[{Taken from \cite[Lemma 3.11]{Charlier 2d gap}}]\label{lemma:Riemann sum}
Let $h \in C^{3}(\mathbb{R})$ and $k \in \{1,2\}$. As $n \to + \infty$, we have
\begin{align}
& \sum_{j=g_{k,-}}^{g_{k,+}}h(M_{j,k}) = b\rho_{k}^{2b} \int_{-M}^{M} h(t) dt \; \sqrt{n} - 2 b \rho_{k}^{2b} \int_{-M}^{M} th(t) dt + \bigg( \frac{1}{2}-\theta_{k,-}^{(n,M)} \bigg)h(M)+ \bigg( \frac{1}{2}-\theta_{k,+}^{(n,M)} \bigg)h(-M) \nonumber \\
& + \frac{1}{\sqrt{n}}\bigg[ 3b\rho_{k}^{2b} \int_{-M}^{M}t^{2}h(t)dt + \bigg( \frac{1}{12}+\frac{\theta_{k,-}^{(n,M)}(\theta_{k,-}^{(n,M)}-1)}{2} \bigg)\frac{h'(M)}{b\rho_{k}^{2b}} - \bigg( \frac{1}{12}+\frac{\theta_{k,+}^{(n,M)}(\theta_{k,+}^{(n,M)}-1)}{2} \bigg)\frac{h'(-M)}{b\rho_{k}^{2b}} \bigg] \nonumber \\
& + \bigO\Bigg(  \frac{1}{n^{3/2}} \sum_{j=g_{k,-}+1}^{g_{k,+}} \bigg( (1+|M_{j,k}|^{3}) \tilde{\mathfrak{m}}_{j,n, k}(h) + (1+M_{j,k}^{2})\tilde{\mathfrak{m}}_{j,n, k}(h') + (1+|M_{j,k}|) \tilde{\mathfrak{m}}_{j,n, k}(h'') + \tilde{\mathfrak{m}}_{j,n, k}(h''') \bigg)   \Bigg), \label{sum f asymp 2}
\end{align}
where, for $\tilde{h} \in C(\mathbb{R})$ and $j \in \{g_{k,-}+1,\dots,g_{k,+}\}$, we define $\tilde{\mathfrak{m}}_{j,n, k}(\tilde{h}) := \max_{x \in [M_{j,k},M_{j-1,k}]}|\tilde{h}(x)|$.
\end{lemma}

\begin{lemma}\label{lemma:S2kp2p hard}
Let $x_{1},\dots,x_{2m} \in \mathbb{R}$ be fixed. There exists $\delta > 0$ such that
\begin{align*}
&  S_{2}^{(2)} = E_{2}^{(M)} \sqrt{n} + E_{4}^{(M)} + \frac{E_{5}^{(M)}}{\sqrt{n}} + \bigO\bigg( \frac{M^{4}}{n} + \frac{M^{14}}{n^{2}} \bigg), \\
& E_{2}^{(M)} = 2b\rho_{1}^{2b} M \ln(1+\mathsf{T}_{0}(b\rho_{1}^{2b})+\hat{\mathsf{T}}_{0}(b\rho_{2}^{2b})), \\
& E_{4}^{(M)} = \ln(1+\mathsf{T}_{0}(b\rho_{1}^{2b})+\hat{\mathsf{T}}_{0}(b\rho_{2}^{2b})) \big( 1-\theta_{1,-}^{(n,M)}-\theta_{1,+}^{(n,M)} \big) + b \rho_{1}^{2b} \int_{-M}^{M} h_{1}(t)dt, \\
& E_{5}^{(M)} = 2b\rho_{1}^{2b}M^{3} \ln(1+\mathsf{T}_{0}(b\rho_{1}^{2b})+\hat{\mathsf{T}}_{0}(b\rho_{2}^{2b})) + \bigg( \frac{1}{2}-\theta_{1,-}^{(n,M)} \bigg) h_{1}(M) + \bigg( \frac{1}{2}-\theta_{1,+}^{(n,M)} \bigg) h_{1}(-M) \\
& \hspace{1.15cm} + b \rho_{1}^{2b} \int_{-M}^{M} \big( h_{2}(t)-2th_{1}(t) \big)dt,
\end{align*}
as $n \to +\infty$ uniformly for $u_{1} \in \{z \in \mathbb{C}: |z-x_{1}|\leq \delta\},\dots,u_{2m} \in \{z \in \mathbb{C}: |z-x_{2m}|\leq \delta\}$, where $h_{1}$, $h_{2}$ are given by
\begin{align*}
& h_{1}(x) = -\frac{ 2\rho_{1}^{b} \mathsf{T}_{1}(b\rho_{1}^{2b})}{1+\mathsf{T}_{0}(b\rho_{1}^{2b})+\hat{\mathsf{T}}_{0}(b\rho_{2}^{2b})} \frac{e^{-\frac{1}{2}x^{2}\rho_{1}^{2b}}}{\sqrt{2\pi} \, \mathrm{erfc}(-\frac{x \rho_{1}^{b}}{\sqrt{2}})}, \\
& h_{2}(x) = -\frac{h_{1}(x)^{2}}{2} + \frac{1}{1+\mathsf{T}_{0}(b\rho_{1}^{2b})+\hat{\mathsf{T}}_{0}(b\rho_{2}^{2b})} \frac{e^{-\frac{1}{2}x^{2}\rho_{1}^{2b}}}{\sqrt{2\pi} \, \mathrm{erfc}(-\frac{x \rho_{1}^{b}}{\sqrt{2}})} \bigg\{ \Big( \rho_{1}^{b} x - \frac{5}{3}\rho_{1}^{3b}x^{3} \Big) \mathsf{T}_{1}(b\rho^{2b}) \\
& \hspace{1.25cm} -\rho_{1}^{3b}x \mathsf{T}_{2}(b\rho_{1}^{2b}) + \frac{4-10\rho_{1}^{2b}x^{2}}{3} \mathsf{T}_{1}(b\rho_{1}^{2b}) \frac{e^{-\frac{1}{2}x^{2}\rho_{1}^{2b}}}{\sqrt{2\pi}\, \mathrm{erfc}(-\frac{x\rho_{1}^{b}}{\sqrt{2}})} \bigg\}.
\end{align*}
\end{lemma}
\begin{proof}
In a similar way as in \eqref{lol4}, we infer that there exists $\delta > 0$ such that
\begin{align}\label{lol1 hard}
& S_{2}^{(2)} = \sum_{j:\lambda_{j,1}\in I_{2}} \ln \bigg( 1 + \hat{\mathsf{T}}_{0}(b\rho_{2}^{2b})+\sum_{\ell=1}^{m} \omega_{\ell} \frac{ \frac{1}{2}\mathrm{erfc}\big(-\hat{\eta}_{j,\ell} \sqrt{a_{j}/2}\big) - R_{a_{j}}(\hat\eta_{j,\ell}) }{ \frac{1}{2}\mathrm{erfc}\big(-\eta_{j,1} \sqrt{a_{j}/2}\big) - R_{a_{j}}(\eta_{j,1}) }  \bigg) +\bigO(e^{-cn^{1/5}}),
\end{align}
as $n \to + \infty$ uniformly for $u_{1} \in \{z \in \mathbb{C}: |z-x_{1}|\leq \delta\},\dots,u_{2m} \in \{z \in \mathbb{C}: |z-x_{2m}|\leq \delta\}$. For $j \in \{j:\lambda_{j,1}\in I_{2}\}$, we have $1-\frac{M}{\sqrt{n}} \leq \lambda_{j,1} = \frac{bn\rho_{1}^{2b}}{j+\alpha} \leq 1+\frac{M}{\sqrt{n}}$, $-M \leq M_{j,1} \leq M$, and 
\begin{align*}
\hat{M}_{j,\ell} = M_{j,1} - \frac{t_{\ell}}{\sqrt{n}} - \frac{t_{\ell}M_{j,1}}{n}, \qquad \ell=1,\dots,m.
\end{align*}
Furthermore, for $\ell\in \{1,\dots,m\}$, as $n \to + \infty$ we have 
\begin{align*}
\hat{\eta}_{j,\ell} = &\; \frac{M_{j,1}}{\sqrt{n}} - \frac{M_{j,1}^{2}+3t_{\ell}}{3n} + \frac{7M_{j,1}^{3}-12 t_{\ell} M_{j,1}}{36 n^{3/2}} + \bigO\bigg(\frac{1+M_{j,1}^{4}}{n^2}\bigg),
	 \\
-\hat{\eta}_{j,\ell} \sqrt{a_{j}/2} = & - \frac{M_{j,1}\rho_{1}^{b}}{\sqrt{2}} + \frac{(5M_{j,1}^{2}+6t_{\ell})\rho_{1}^{b}}{6\sqrt{2} \sqrt{n}} - \frac{\rho_{1}^{b} M_{j,1}(53M_{j,1}^{2}+12t_{\ell})}{72\sqrt{2} n} 
+ \bigO\bigg( \frac{1+M_{j,1}^{4}}{n^{3/2}} \bigg),
\end{align*}
uniformly for $j\in \{j:\lambda_{j,1}\in I_{2}\}$. Hence, for $\ell\in \{1,\dots,m\}$, by \eqref{asymp of Ra} as $n \to + \infty$ we have
\begin{align*}
& R_{a_{j}}(\hat{\eta}_{j,\ell}) = \frac{e^{-\frac{M_{j,1}^{2}\rho_{1}^{2b}}{2}}}{\sqrt{2\pi}} \bigg( \frac{-1}{3\rho_{1}^{b}\sqrt{n}} - \frac{M_{j,1}(3+10M_{j,1}^{2}\rho_{1}^{2b}+12t_{\ell}\rho_{1}^{2b})}{36\rho_{1}^{b}n} + \bigO((1+M_{j,1}^{6})n^{-\frac{3}{2}}) \bigg)
\end{align*}
and
\begin{align*}
& \frac{1}{2}\mathrm{erfc}\Big(-\hat{\eta}_{j,\ell} \sqrt{a_{j}/2}\Big) = \frac{1}{2}\mathrm{erfc}\Big(-\frac{\rho_{1}^{b}M_{j,1}}{\sqrt{2}}\Big) -\frac{e^{-\frac{M_{j,1}^{2}\rho_{1}^{2b}}{2}}\rho_{1}^{b}(5 M_{j,1}^{2} + 6 t_{\ell})}{6\sqrt{2\pi}\sqrt{n}} \\
& + \frac{e^{-\frac{M_{j,1}^{2}\rho_{1}^{2b}}{2}}M_{j,1} \rho_{1}^{b}}{72\sqrt{2\pi} \, n} \Big( 53M_{j,1}^{2} + 12 t_{\ell} - 25 M_{j,1}^{4} \rho_{1}^{2b} - 60 M_{j,1}^{2} t_{\ell} \rho_{1}^{2b} - 36 t_{\ell}^{2}\rho_{1}^{2b} \Big) + \bigO\Big(e^{-\frac{M_{j,1}^{2}\rho_{1}^{2b}}{2}}\frac{1+M_{j,1}^{8}}{n^{3/2}} \Big),
\end{align*}
uniformly for $j\in \{j:\lambda_{j,1}\in I_{2}\}$. Extending the above asymptotic formulas to higher order (see \cite[Lemma 2.8]{ACCL2022} for details in a similar situation), we deduce that the argument of $\ln$ in \eqref{lol1 hard} enjoys the following asymptotics
\begin{align}
& 1+\hat{\mathsf{T}}_{0}(b\rho_{2}^{2b})+\sum_{\ell=1}^{m} \omega_{\ell} \frac{ \frac{1}{2}\mathrm{erfc}\big(-\hat{\eta}_{j,\ell} \sqrt{a_{j}/2}\big) - R_{a_{j}}(\hat\eta_{j,\ell}) }{ \frac{1}{2}\mathrm{erfc}\big(-\eta_{j,1} \sqrt{a_{j}/2}\big) - R_{a_{j}}(\eta_{j,1}) }    \nonumber \\
& = g_{1}(M_{j,1}) + \frac{g_{2}(M_{j,1})}{\sqrt{n}} + \frac{g_{3}(M_{j,1})}{n} + \frac{g_{4}(M_{j,1})}{n^{3/2}} + \frac{g_{5}(M_{j,1})}{n^{2}} + \bigO\Big(\frac{1+|M_{j,1}|^{13}}{n^{5/2}}\Big), \label{asymp of S2kp2p in proof hard}
\end{align}
as $n \to + \infty$, where
\begin{align*}
& g_{1}(x) = 1+\mathsf{T}_{0}(b\rho_{1}^{2b})+\hat{\mathsf{T}}_{0}(b\rho_{2}^{2b}), \qquad g_{2}(x) = - \frac{e^{-\frac{1}{2}x^{2}\rho_{1}^{2b}}2\rho_{1}^{b} \mathsf{T}_{1}(b\rho_{1}^{2b})}{\sqrt{2\pi} \mathrm{erfc}(-\frac{x\rho_{1}^{b}}{\sqrt{2}})}, \\
& g_{3}(x) = \frac{e^{-\frac{1}{2}x^{2}\rho_{1}^{2b}}}{3\sqrt{2\pi} \, \mathrm{erfc}(-\frac{x\rho_{1}^{b}}{\sqrt{2}})} \bigg\{ \frac{e^{-\frac{1}{2}x^{2}\rho_{1}^{2b}} \mathsf{T}_{1}(b\rho_{1}^{2b})}{\sqrt{2\pi} \, \mathrm{erfc}(-\frac{x\rho_{1}^{b}}{\sqrt{2}})} (4-10x^{2}\rho_{1}^{2b}) + \mathsf{T}_{1}(b\rho_{1}^{2b}) \big( 3x\rho_{1}^{b}-5x^{3}\rho_{1}^{3b} \big) \\
& \hspace{4.1cm} -3 \rho_{1}^{3b}x \mathsf{T}_{2}(b\rho_{1}^{2b}) \bigg\}.
\end{align*}
We remark that equation (\ref{asymp of S2kp2p in proof hard}) coincides with \cite[eq (2.29)]{ACCL2022} except that (\ref{asymp of S2kp2p in proof hard}) has the extra term $\hat{\mathsf{T}}_{0}(b\rho_{2}^{2b})$ on each side and $(\rho, M_j, \eta_j, \eta_{j,\ell})$ in \cite[eq (2.29)]{ACCL2022} are given by $(\rho_1, M_{j,1}, \eta_{j,1}, \hat{\eta}_{j,\ell})$ here.
The functions $g_{4}$ and $g_{5}$ can also be computed explicitly, but we do not write them down. Since $\frac{e^{-\frac{1}{2}x^{2}\rho_{1}^{2b}}}{\sqrt{2\pi} \, \mathrm{erfc}(-\frac{x\rho_{1}^{b}}{\sqrt{2}})} \sim -\frac{\rho_{1}^{b} x}{2}$ as $x \to - \infty$, $g_{2}(x) = \bigO(x)$ as $x \to -\infty$. At first glance, it seems that $g_{3}(x) = \bigO(x^{4})$ as $x \to -\infty$. However, some cancellations take place and in fact $g_{3}(x) = \bigO(x^{2})$ as $x \to -\infty$. Likewise, the expressions for $g_{4}$ and $g_{5}$ suggest a priori that $g_{4}(x) = \bigO(x^{7})$ and $g_{5}(x) = \bigO(x^{10})$ as $x \to -\infty$, but here too, cancellations take place and in fact  $g_{4}(x) = \bigO(x^{3})$ and $g_{5}(x) = \bigO(x^{4})$ as $x \to -\infty$. Using the above large $x$ estimates for $\{g_{j}(x)\}_{j=2}^{5}$ and \eqref{asymp of S2kp2p in proof hard}, we obtain after a computation that
\begin{align*}
& S_{2}^{(2)} = \sum_{j=g_{1,-}}^{g_{1,+}} \bigg\{ \ln(1+\mathsf{T}_{0}(b\rho_{1}^{2b})+\hat{\mathsf{T}}_{0}(b\rho_{2}^{2b})) + \frac{h_{1}(M_{j,1})}{\sqrt{n}} + \frac{h_{2}(M_{j,1})}{n} + \bigO\bigg( \frac{1+|M_{j,1}|^{3}}{n^{3/2}} + \frac{1+|M_{j,1}|^{13}}{n^{5/2}} \bigg) \bigg\}.
\end{align*}
as $n \to + \infty$. The above error can be estimated as follows:
\begin{align*}
\sum_{j=g_{1,-}}^{g_{1,+}}\bigO\bigg( \frac{1+|M_{j,1}|^{3}}{n^{3/2}} + \frac{1+|M_{j,1}|^{13}}{n^{5/2}} \bigg) = \bigO\bigg( \frac{M^{4}}{n} + \frac{M^{14}}{n^{2}} \bigg), \qquad \mbox{as } n \to + \infty.
\end{align*}
Using then Lemma \ref{lemma:Riemann sum}, we obtain the claim.
\end{proof}

\begin{lemma}\label{lemma:S4p2p hard}
Let $x_{1},\dots,x_{2m} \in \mathbb{R}$ be fixed. There exists $\delta > 0$ such that
\begin{align*}
&  S_{4}^{(2)} = \hat{E}_{2}^{(M)} \sqrt{n} + \hat{E}_{4}^{(M)} + \frac{\hat{E}_{5}^{(M)}}{\sqrt{n}} + \bigO\bigg( \frac{M^{4}}{n} + \frac{M^{14}}{n^{2}} \bigg), \\
& \hat{E}_{2}^{(M)} = 0, \qquad \hat{E}_{4}^{(M)} = b \rho_{2}^{2b} \int_{-M}^{M} \hat{h}_{1}(t)dt, \\
& \hat{E}_{5}^{(M)} = \bigg( \frac{1}{2}-\theta_{2,-}^{(n,M)} \bigg) \hat{h}_{1}(M) + \bigg( \frac{1}{2}-\theta_{2,+}^{(n,M)} \bigg) \hat{h}_{1}(-M) + b \rho_{2}^{2b} \int_{-M}^{M} \big( \hat{h}_{2}(t) - 2 t \hat{h}_{1}(t) \big)dt,
\end{align*}
as $n \to +\infty$ uniformly for $u_{1} \in \{z \in \mathbb{C}: |z-x_{1}|\leq \delta\},\dots,u_{2m} \in \{z \in \mathbb{C}: |z-x_{2m}|\leq \delta\}$, where $\hat{h}_{1}$, $\hat{h}_{2}$ are given by
\begin{align*}
& \hat{h}_{1}(x) = 2\rho_{2}^{b} \hat{\mathsf{T}}_{1}(b\rho_{2}^{2b}) \frac{e^{-\frac{1}{2}x^{2}\rho_{2}^{2b}}}{\sqrt{2\pi} \, \big(2 - \mathrm{erfc}(-\frac{x \rho_{2}^{b}}{\sqrt{2}})\big)}, \\
& \hat{h}_{2}(x) = -\frac{\hat{h}_{1}(x)^{2}}{2} + \frac{e^{-\frac{1}{2}x^{2}\rho_{2}^{2b}}}{\sqrt{2\pi} \, \big(2 - \mathrm{erfc}(-\frac{x \rho_{2}^{b}}{\sqrt{2}})\big)} \bigg\{ \Big( -\rho_{2}^{b} x + \frac{5}{3}\rho_{2}^{3b}x^{3} \Big) \hat{\mathsf{T}}_{1}(b\rho_{2}^{2b}) \\
& \hspace{1.25cm} -\rho_{2}^{3b}x \hat{\mathsf{T}}_{2}(b\rho_{2}^{2b}) + \frac{4-10\rho_{2}^{2b}x^{2}}{3} \hat{\mathsf{T}}_{1}(b\rho_{2}^{2b}) \frac{e^{-\frac{1}{2}x^{2}\rho_{2}^{2b}}}{\sqrt{2\pi}\, \big(2-\mathrm{erfc}(-\frac{x\rho_{2}^{b}}{\sqrt{2}})\big)} \bigg\}.
\end{align*}
\end{lemma}
\begin{proof}
In a similar way as in \eqref{lol5}, we infer that there exists $\delta > 0$ such that
\begin{align}\label{lol1 hard hat}
& S_{4}^{(2)} = \bigO(e^{-cn}) + \sum_{j:\lambda_{j,2}\in I_{2}}  \ln \bigg( 1  \nonumber \\
& +\sum_{\ell=m+1}^{2m} \omega_{\ell} \frac{ \frac{1}{2}\mathrm{erfc}\big(-\hat{\eta}_{j,\ell} \sqrt{a_{j}/2}\big) - R_{a_{j}}(\hat\eta_{j,\ell}) - \big[\frac{1}{2}\mathrm{erfc}\big(-\eta_{j,2} \sqrt{a_{j}/2}\big) - R_{a_{j}}(\eta_{j,2})\big] }{ 1 - \big[\frac{1}{2}\mathrm{erfc}\big(-\eta_{j,2} \sqrt{a_{j}/2}\big) - R_{a_{j}}(\eta_{j,2})\big] }  \bigg)
\end{align}
as $n \to + \infty$ uniformly for $u_{1} \in \{z \in \mathbb{C}: |z-x_{1}|\leq \delta\},\dots,u_{2m} \in \{z \in \mathbb{C}: |z-x_{2m}|\leq \delta\}$. For $j \in \{j:\lambda_{j,2}\in I_{2}\}$, we have $1-\frac{M}{\sqrt{n}} \leq \lambda_{j,2} = \frac{bn\rho_{2}^{2b}}{j+\alpha} \leq 1+\frac{M}{\sqrt{n}}$, $-M \leq M_{j,2} \leq M$, and 
\begin{align*}
\hat{M}_{j,k} = M_{j,2} + \frac{t_{k}}{\sqrt{n}} + \frac{t_{k}M_{j,2}}{n}, \qquad k=m+1,\dots,2m.
\end{align*}
Furthermore, for $\ell\in \{m+1,\dots,2m\}$, as $n \to + \infty$ we have 
\begin{align}
\hat{\eta}_{j,\ell} & = \frac{M_{j,2}}{\sqrt{n}} - \frac{M_{j,2}^{2}-3t_{\ell}}{3n} + \frac{7M_{j,2}^{3}+12 t_{\ell} M_{j,2}}{36 n^{3/2}} + \bigO\bigg(\frac{1+M_{j,2}^{4}}{n^{2}}\bigg)
	 \\
-\hat{\eta}_{j,\ell} \sqrt{a_{j}/2} & = - \frac{M_{j,2}\rho_{2}^{b}}{\sqrt{2}} + \frac{(5M_{j,2}^{2}-6t_{\ell})\rho_{2}^{b}}{6\sqrt{2} \sqrt{n}} - \frac{\rho_{2}^{b} M_{j,2}(53M_{j,2}^{2}-12t_{\ell})}{72\sqrt{2} n} + \bigO\bigg( \frac{1+M_{j,2}^{4}}{n^{3/2}} \bigg)
\end{align}
uniformly for $j\in \{j:\lambda_{j,2}\in I_{2}\}$. Hence, for $\ell\in \{m+1,\dots,2m\}$, by \eqref{asymp of Ra} as $n \to + \infty$ we have
\begin{align*}
& R_{a_{j}}(\hat{\eta}_{j,\ell}) = \frac{e^{-\frac{M_{j,2}^{2}\rho_{2}^{2b}}{2}}}{\sqrt{2\pi}} \bigg( \frac{-1}{3\rho_{2}^{b}\sqrt{n}} - \frac{M_{j,2}(3+10M_{j,2}^{2}\rho_{2}^{2b}-12t_{\ell}\rho_{2}^{2b})}{36\rho_{2}^{b}n} + \bigO((1+M_{j,2}^{6})n^{-\frac{3}{2}}) \bigg)
\end{align*}
and
\begin{align*}
& \frac{1}{2}\mathrm{erfc}\Big(-\hat{\eta}_{j,\ell} \sqrt{a_{j}/2}\Big) = \frac{1}{2}\mathrm{erfc}\Big(-\frac{\rho_{2}^{b}M_{j,2}}{\sqrt{2}}\Big) -\frac{e^{-\frac{M_{j,2}^{2}\rho_{2}^{2b}}{2}}\rho_{2}^{b}(5 M_{j,2}^{2} - 6 t_{\ell})}{6\sqrt{2\pi}\sqrt{n}} \\
& + \frac{e^{-\frac{M_{j,2}^{2}\rho_{2}^{2b}}{2}}M_{j,2} \rho_{2}^{b}}{72\sqrt{2\pi} \, n} \Big( 53M_{j,2}^{2} - 12 t_{\ell} - 25 M_{j,2}^{4} \rho_{2}^{2b} + 60 M_{j,2}^{2} t_{\ell} \rho_{2}^{2b} - 36 t_{\ell}^{2}\rho_{2}^{2b} \Big) + \bigO\Big(e^{-\frac{M_{j,2}^{2}\rho_{2}^{2b}}{2}}\frac{1+M_{j,2}^{8}}{n^{3/2}} \Big),
\end{align*}
uniformly for $j\in \{j:\lambda_{j,2}\in I_{2}\}$. Extending the above asymptotic formulas to higher order (see \cite[Lemma 2.8]{ACCL2022} for details in a similar situation), we deduce from \eqref{lol1 hard hat} that
\begin{align}
& 1 + \sum_{\ell=m+1}^{2m} \omega_{\ell} \frac{ \frac{1}{2}\mathrm{erfc}\big(-\hat{\eta}_{j,\ell} \sqrt{a_{j}/2}\big) - R_{a_{j}}(\hat\eta_{j,\ell}) - \big[\frac{1}{2}\mathrm{erfc}\big(-\eta_{j,2} \sqrt{a_{j}/2}\big) - R_{a_{j}}(\eta_{j,2})\big] }{ 1 - \big[\frac{1}{2}\mathrm{erfc}\big(-\eta_{j,2} \sqrt{a_{j}/2}\big) - R_{a_{j}}(\eta_{j,2})\big] } \nonumber \\
& = 1 + \frac{\hat{g}_{2}(M_{j,2})}{\sqrt{n}} + \frac{\hat{g}_{3}(M_{j,2})}{n} + \frac{\hat{g}_{4}(M_{j,2})}{n^{3/2}} + \frac{\hat{g}_{5}(M_{j,2})}{n^{2}} + \bigO\Big(\frac{1+|M_{j,2}|^{13}}{n^{5/2}}\Big)\bigg), \label{asymp of S2kp2p in proof hard hat}
\end{align}
as $n \to + \infty$, where
\begin{align*}
& \hat{g}_{2}(x) = \frac{e^{-\frac{1}{2}x^{2}\rho_{2}^{2b}}2\rho_{2}^{b} \hat{\mathsf{T}}_{1}(b\rho_{2}^{2b})}{\sqrt{2\pi} (2-\mathrm{erfc}(-\frac{x\rho_{2}^{b}}{\sqrt{2}}))}, \\
& \hat{g}_{3}(x) = \frac{e^{-\frac{1}{2}x^{2}\rho_{2}^{2b}}}{\sqrt{2\pi} \, \big(2 - \mathrm{erfc}(-\frac{x \rho_{2}^{b}}{\sqrt{2}})\big)} \bigg\{ \Big( -\rho_{2}^{b} x + \frac{5}{3}\rho_{2}^{3b}x^{3} \Big) \hat{\mathsf{T}}_{1}(b\rho_{2}^{2b}) \\
& \hspace{1.25cm} -\rho_{2}^{3b}x \hat{\mathsf{T}}_{2}(b\rho_{2}^{2b}) + \frac{4-10\rho_{2}^{2b}x^{2}}{3} \hat{\mathsf{T}}_{1}(b\rho_{2}^{2b}) \frac{e^{-\frac{1}{2}x^{2}\rho_{2}^{2b}}}{\sqrt{2\pi}\, \big(2-\mathrm{erfc}(-\frac{x\rho_{2}^{b}}{\sqrt{2}})\big)} \bigg\}.
\end{align*}
The functions $\hat{g}_{4}$ and $\hat{g}_{5}$ can also be computed explicitly, but we do not write them down. Since $\frac{e^{-\frac{1}{2}x^{2}\rho_{2}^{2b}}}{\sqrt{2\pi} \, (2-\mathrm{erfc}(-\frac{x\rho_{2}^{b}}{\sqrt{2}}))} \sim \frac{\rho_{2}^{b} x}{2}$ as $x \to + \infty$, $\hat{g}_{2}(x) = \bigO(x)$ as $x \to +\infty$. At first glance, it seems that $\hat{g}_{3}(x) = \bigO(x^{4})$ as $x \to +\infty$. However, some cancellations take place and in fact $\hat{g}_{3}(x) = \bigO(x^{2})$ as $x \to +\infty$. Likewise, the expressions for $\hat{g}_{4}$ and $\hat{g}_{5}$ suggest a priori that $\hat{g}_{4}(x) = \bigO(x^{7})$ and $\hat{g}_{5}(x) = \bigO(x^{10})$ as $x \to +\infty$, but here too, cancellations take place and in fact $\hat{g}_{4}(x) = \bigO(x^{3})$ and $\hat{g}_{5}(x) = \bigO(x^{4})$ as $x \to +\infty$. Using the above large $x$ estimates for $\{\hat{g}_{j}(x)\}_{j=2}^{5}$ and \eqref{asymp of S2kp2p in proof hard hat}, we obtain after a computation that
\begin{align*}
& S_{4}^{(2)} = \sum_{j=g_{2,-}}^{g_{2,+}} \bigg\{ \frac{\hat{h}_{1}(M_{j,2})}{\sqrt{n}} + \frac{\hat{h}_{2}(M_{j,2})}{n} + \bigO\bigg( \frac{1+|M_{j,2}|^{3}}{n^{3/2}} + \frac{1+|M_{j,2}|^{13}}{n^{5/2}} \bigg) \bigg\}.
\end{align*}
as $n \to + \infty$, where $\hat{h}_{1} = \hat{g}_{2}$ and $\hat{h}_{2} = - \frac{\hat{h}_{1}^{2}}{2}+\hat{g}_3$. Note that
\begin{align*}
\sum_{j=g_{2,-}}^{g_{2,+}}\bigO\bigg( \frac{1+|M_{j,2}|^{3}}{n^{3/2}} + \frac{1+|M_{j,2}|^{13}}{n^{5/2}} \bigg) = \bigO\bigg( \frac{M^{4}}{n} + \frac{M^{14}}{n^{2}} \bigg), \qquad \mbox{as } n \to + \infty.
\end{align*}
Using then Lemma \ref{lemma:Riemann sum}, we find the claim.
\end{proof}

Define the real constants $\{\mathcal{I}_j\}_1^4 \subset \mathbb{R}$ by
\begin{align}
& \mathcal{I}_1 = \int_{-\infty}^{+\infty} \bigg\{ \frac{e^{-y^{2}}}{\sqrt{\pi}\, \mathrm{erfc}(y)} - \chi_{(0,+\infty)}(y) \bigg[ y + \frac{y}{2(1+y^{2})} \bigg] \bigg\}dy, \label{def of I1}
	\\
& \mathcal{I}_{2} = \int_{-\infty}^{+\infty} \bigg\{ \frac{y^{3}e^{-y^{2}}}{\sqrt{\pi} \, \mathrm{erfc}(y)} - \chi_{(0,+\infty)}(y) \bigg[ y^{4}+\frac{y^{2}}{2}-\frac{1}{2} \bigg] \bigg\}dy, \label{def of I3} 
	\\
& \mathcal{I}_{3} = \int_{-\infty}^{+\infty} \bigg\{ \bigg( \frac{e^{-y^{2}}}{\sqrt{\pi} \, \mathrm{erfc}(y)} \bigg)^{2} - \chi_{(0,+\infty)}(y) \bigg[ y^{2}+1 \bigg] \bigg\} dy, \label{def of I4}
	\\
& \mathcal{I}_{4} = \int_{-\infty}^{+\infty} \bigg\{ \bigg( \frac{y \, e^{-y^{2}}}{\sqrt{\pi} \, \mathrm{erfc}(y)} \bigg)^{2} - \chi_{(0,+\infty)}(y)\bigg[ y^{4}+y^{2}-\frac{3}{4} \bigg] \bigg\} dy. \label{def of I5}
\end{align}
Recall also that $\mathcal{I}$ is given by \eqref{def of I}.

\begin{lemma}\label{lemma: asymp of S2k final hard}
For any fixed $x_{1},\dots,x_{2m} \in \mathbb{R}$, there exists $\delta > 0$ such that
\begin{align*}
& S_{2} =  - j_{1,-} \ln  \Omega + F_{1}^{(\epsilon)}n + F_{2} \ln n + F_{3}^{(n,\epsilon)} + \frac{F_{4}}{\sqrt{n}}  + \bigO\bigg(\frac{\sqrt{n}}{M^{11}} + \frac{1}{M^{6}} + \frac{1}{\sqrt{n} M} + \frac{M^{4}}{n} + \frac{M^{14}}{n^{2}} \bigg),
\end{align*}
as $n \to +\infty$ uniformly for $u_{1} \in \{z \in \mathbb{C}: |z-x_{1}|\leq \delta\},\dots,u_{2m} \in \{z \in \mathbb{C}: |z-x_{2m}|\leq \delta\}$, where
\begin{align*}
F_{1}^{(\epsilon)} = &\; b \rho_{1}^{2b} \ln \Omega + \int_{b\rho_{1}^{2b}}^{\frac{b\rho_{1}^{2b}}{1-\epsilon}} f_{1}(x)dx, 
	\\
F_{2} = & -\frac{b\rho_{1}^{2b}}{2} \frac{\mathsf{T}_{1}(b\rho_{1}^{2b})}{\Omega}, 
	\\
F_{3}^{(n,\epsilon)} = &\; \frac{1}{2}\ln \Omega + \int_{b\rho_{1}^{2b}}^{\frac{b\rho_{1}^{2b}}{1-\epsilon}} \bigg\{ f(x) + \frac{b\rho_{1}^{2b} \mathsf{T}_{1}(b\rho_{1}^{2b})}{\Omega(x-b\rho_{1}^{2b})} \bigg\}dx + \bigg( \frac{1}{2}-\alpha - \theta_{1,+}^{(n,\epsilon)}  \bigg) f_{1}\bigg(\frac{b\rho_{1}^{2b}}{1-\epsilon}\bigg) 
	\\
&  - 2b\rho_{1}^{2b} \frac{\mathsf{T}_{1}(b\rho_{1}^{2b})}{\Omega} \mathcal{I}_1 
 + \frac{b\rho_{1}^{2b}}{2\Omega}\mathsf{T}_{1}(b \rho_{1}^{2b}) \big( \ln 2 - 2b \ln(\rho_{1}) \big) - \frac{\mathsf{T}_{1}(b\rho_{1}^{2b})}{\Omega}b\rho_{1}^{2b} \ln \bigg( \frac{\epsilon}{1-\epsilon} \bigg), 
 	\\
 F_{4} = &\; \sqrt{2}b\rho_{1}^{b}\frac{ \rho_{1}^{2b} \mathsf{T}_{2}(b\rho_{1}^{2b}) -5 \mathsf{T}_{1}(b\rho_{1}^{2b}) }{ \Omega } \mathcal{I}
 + \frac{10\sqrt{2} b \rho_{1}^{b}}{3} \frac{\mathsf{T}_{1}(b\rho_{1}^{2b})}{\Omega} \mathcal{I}_2 
 	\\
& + \sqrt{2} b \rho_{1}^{b} \frac{\mathsf{T}_{1}(b\rho_{1}^{2b})}{\Omega} \bigg( \frac{2}{3} -  \rho_{1}^{2b} \frac{\mathsf{T}_{1}(b\rho_{1}^{2b})}{\Omega} \bigg) \mathcal{I}_3 
- \frac{10 \sqrt{2} b \rho_{1}^{b}}{3} \frac{\mathsf{T}_{1}(b\rho_{1}^{2b})}{\Omega} \mathcal{I}_4,
\end{align*}
and $f_{1}$ and $f$ are as in the statement of Lemma \ref{lemma:S3 asymp hard}.
\end{lemma}
\begin{proof}
By combining Lemmas \ref{lemma:S2kp3p hard}, \ref{lemma:S2p1p hard} and \ref{lemma:S2kp2p hard}, we have
\begin{align*}
& S_{2} = - j_{1,-} \ln  \Omega + F_{1}^{(\epsilon)}n + \widetilde{F}_{2}\sqrt{n} + F_{2} \ln n + F_{3}^{(n,\epsilon,M)} + \frac{F_{4}^{(M)}}{\sqrt{n}}  + \bigO\bigg(\frac{\sqrt{n}}{M^{11}} + \frac{1}{M^{6}} + \frac{1}{\sqrt{n} M} + \frac{M^{4}}{n} + \frac{M^{14}}{n^{2}} \bigg),
\end{align*}
as $n \to +\infty$ uniformly for $u_{1} \in \{z \in \mathbb{C}: |z-x_{1}|\leq \delta\},\dots,u_{2m} \in \{z \in \mathbb{C}: |z-x_{2m}|\leq \delta\}$, where $F_{1}^{(\epsilon)}$ and $F_{2}$ are as in the statement, and
\begin{align*}
& \widetilde{F}_{2} = - bM \rho_{1}^{2b} \ln \Omega + D_{2}^{(M)}+E_{2}^{(M)}, \\
& F_{3}^{(n,\epsilon,M)} = \big( bM^{2}\rho_{1}^{2b} - \alpha + \theta_{1,-}^{(n,M)} \big) \ln \Omega + D_{4}^{(n,\epsilon,M)} + E_{4}^{(M)}, \\
& F_{4}^{(M)} = -bM^{3}\rho_{1}^{2b} \ln \Omega + D_{5}^{(n,M)} + E_{5}^{(M)}.
\end{align*}
Using that $f_{1}(b \rho_{1}^{2b}) = \ln(1+\mathsf{T}_{0}(b\rho_{1}^{2b})+\hat{\mathsf{T}}_{0}(b\rho_{2}^{2b})) = \ln \Omega$, we readily verify that $\widetilde{F}_{2}=0$. Furthermore, a long but direct computation shows that (see \cite[Lemma 2.9]{ACCL2022} for a similar situation with more details provided)
\begin{align*}
& F_{3}^{(n,\epsilon,M)} = F_{3}^{(n,\epsilon)} + \bigO(M^{-6}), \qquad F_{4}^{(M)} = F_{4} + \bigO(M^{-1}),
\end{align*}
and the claim follows.
\end{proof}

\begin{lemma}\label{lemma: asymp of S4 final hard}
For any fixed $x_{1},\dots,x_{2m} \in \mathbb{R}$, there exists $\delta > 0$ such that
\begin{align*}
& S_{4} = \hat{F}_{1}^{(\epsilon)}n + \hat{F}_{2} \ln n + \hat{F}_{3}^{(n,\epsilon)} + \frac{\hat{F}_{4}}{\sqrt{n}}  + \bigO\bigg(\frac{\sqrt{n}}{M^{11}} +\frac{1}{M^{6}} + \frac{1}{\sqrt{n} M} + \frac{M^{4}}{n} + \frac{M^{14}}{n^{2}} \bigg),
\end{align*}
as $n \to +\infty$ uniformly for $u_{1} \in \{z \in \mathbb{C}: |z-x_{1}|\leq \delta\},\dots,u_{2m} \in \{z \in \mathbb{C}: |z-x_{2m}|\leq \delta\}$, where
\begin{align*}
\hat{F}_{1}^{(\epsilon)} = &\; \int_{\frac{b\rho_{2}^{2b}}{1+\epsilon}}^{b\rho_{2}^{2b}} \hat{f}_{1}(x)dx, 
	\\
\hat{F}_{2} = &\; \frac{b\rho_{2}^{2b}}{2} \hat{\mathsf{T}}_{1}(b\rho_{2}^{2b}), 
	\\
\hat{F}_{3}^{(n,\epsilon)} = &\; \int_{\frac{b\rho_{2}^{2b}}{1+\epsilon}}^{b\rho_{2}^{2b}} \bigg\{ \hat{f}(x) - \frac{b\rho_{2}^{2b} \hat{\mathsf{T}}_{1}(b\rho_{2}^{2b})}{b\rho_{2}^{2b}-x} \bigg\}dx + \bigg( \frac{1}{2}+\alpha - \theta_{2,-}^{(n,\epsilon)}  \bigg) \hat{f}_{1}\bigg(\frac{b\rho_{2}^{2b}}{1+\epsilon}\bigg) 
	\\
&  + 2b\rho_{2}^{2b} \hat{\mathsf{T}}_{1}(b\rho_{2}^{2b}) \mathcal{I}_1 
 - \frac{b\rho_{2}^{2b}}{2}\hat{\mathsf{T}}_{1}(b \rho_{2}^{2b}) \big( \ln 2 - 2b \ln(\rho_{2}) \big) + \hat{\mathsf{T}}_{1}(b\rho_{2}^{2b}) b\rho_{2}^{2b} \ln \bigg( \frac{\epsilon}{1+\epsilon} \bigg), 
	\\
\hat{F}_{4} = & -\sqrt{2}b\rho_{2}^{b} \big( \rho_{2}^{2b} \hat{\mathsf{T}}_{2}(b\rho_{2}^{2b}) + 5 \hat{\mathsf{T}}_{1}(b\rho_{2}^{2b}) \big) \mathcal{I} 
+ \frac{10\sqrt{2} b \rho_{2}^{b}}{3} \hat{\mathsf{T}}_{1}(b\rho_{2}^{2b}) \mathcal{I}_2 
	\\
& + \sqrt{2} b \rho_{2}^{b} \hat{\mathsf{T}}_{1}(b\rho_{2}^{2b}) \bigg( \frac{2}{3} - \rho_{2}^{2b} \hat{\mathsf{T}}_{1}(b\rho_{2}^{2b}) \bigg) \mathcal{I}_3 
 - \frac{10 \sqrt{2} b \rho_{2}^{b}}{3} \hat{\mathsf{T}}_{1}(b\rho_{2}^{2b}) \mathcal{I}_4,
\end{align*}
and $\hat{f}_{1}$ and $\hat{f}$ are as in the statement of Lemma \ref{lemma:S3 asymp hard}.
\end{lemma}
\begin{proof}
By combining Lemmas \ref{lemma:S4p1p hard}, \ref{lemma:S4p3p hard} and \ref{lemma:S4p2p hard}, we have
\begin{align*}
& S_{4} = \hat{F}_{1}^{(\epsilon)}n + \widetilde{\hat{F}}_{2}\sqrt{n} + \hat{F}_{2} \ln n + \hat{F}_{3}^{(n,\epsilon,M)} + \frac{\hat{F}_{4}^{(M)}}{\sqrt{n}}  + \bigO\bigg(\frac{\sqrt{n}}{M^{11}} + \frac{1}{M^{6}} + \frac{1}{\sqrt{n} M} + \frac{M^{4}}{n} + \frac{M^{14}}{n^{2}} \bigg),
\end{align*}
as $n \to +\infty$ uniformly for $u_{1} \in \{z \in \mathbb{C}: |z-x_{1}|\leq \delta\},\dots,u_{2m} \in \{z \in \mathbb{C}: |z-x_{2m}|\leq \delta\}$, where $\hat{F}_{1}^{(\epsilon)}$ and $\hat{F}_{2}$ are as in the statement, and
\begin{align*}
& \widetilde{\hat{F}}_{2} = \hat{D}_{2}^{(M)}+\hat{E}_{2}^{(M)}, \\
& \hat{F}_{3}^{(n,\epsilon,M)} = \hat{D}_{4}^{(n,\epsilon,M)} + \hat{E}_{4}^{(M)}, \\
& \hat{F}_{4}^{(M)} = \hat{D}_{5}^{(n,M)} + \hat{E}_{5}^{(M)}.
\end{align*}
It is easy to check that $\widetilde{\hat{F}}_{2}=0$. Furthermore, a long but direct computation shows that (see e.g. \cite[Proof of Lemma 2.9]{ACCL2022} for a similar analysis with more details provided)
\begin{align*}
& \hat{F}_{3}^{(n,\epsilon,M)} = \hat{F}_{3}^{(n,\epsilon)} + \bigO(M^{-6}), \qquad \hat{F}_{4}^{(M)} = \hat{F}_{4} + \bigO(M^{-1}),
\end{align*}
and the claim follows. 
\end{proof}

\begin{proof}[End of the proof of Theorem \ref{thm:main thm hard}]
Let $M' > 0$ be sufficiently large such that Lemmas \ref{lemma: S2km1 hard} and \ref{lemma: asymp of S2k final hard} hold. Using \eqref{log Dn as a sum of sums hard} and Lemmas \ref{lemma: S0 hard}, \ref{lemma: S2km1 hard}, \ref{lemma: asymp of S2k final hard}, \ref{lemma:S3 asymp hard}, \ref{lemma: asymp of S4 final hard} and \ref{lemma: S5 hard}, we conclude that for any $x_{1},\dots,x_{m} \in \mathbb{R}$, there exists $\delta > 0$ such that
\begin{align*}
& \ln \mathcal{E}_{n} = S_{0}+S_{1}+S_{2}+S_{3} +S_4 + S_5 \\
& = M' \ln \Omega + (j_{-}-M'-1) \ln \Omega - j_{-} \ln  \Omega + C_{1} n + C_{2} \ln n + (C_{3} + \ln \Omega) + \mathcal{F}_n + \frac{C_{4}}{\sqrt{n}} \\
&  + \bigO\bigg(\frac{\sqrt{n}}{M^{11}} + \frac{1}{M^{6}} + \frac{1}{\sqrt{n} M} + \frac{M^{4}}{n} + \frac{M^{14}}{n^{2}} \bigg),
\end{align*}
as $n \to +\infty$ uniformly for $u_{1} \in \{z \in \mathbb{C}: |z-x_{1}|\leq \delta\},\dots,u_{m} \in \{z \in \mathbb{C}: |z-x_{m}|\leq \delta\}$, where
\begin{align*}
  C_1 = &\; F_1^{(\epsilon)} + \int_{\frac{b\rho_{1}^{2b}}{1-\epsilon}}^{\sigma_{\star}} f_{1}(x)dx + \int_{\sigma_{\star}}^{\frac{b\rho_{2}^{2b}}{1+\epsilon}} \hat{f}_{1}(x)dx + \hat{F}_{1}^{(\epsilon)},
  \qquad C_2 = F_2 + \hat{F}_2, \qquad   C_4 = F_4 + \hat{F}_4,
  	\\
  C_3 + \ln \Omega = &\; F_{3}^{(n,\epsilon)} + \bigg( \alpha - \frac{1}{2} + \theta_{1,+}^{(n,\epsilon)} \bigg) f_{1}\Big( \frac{b \rho_{1}^{2b}}{1-\epsilon} \Big)  + \bigg( \theta_{2,-}^{(n,\epsilon)}-\alpha-\frac{1}{2} \bigg) \hat{f}_{1}\Big(\frac{b\rho_{2}^{2b}}{1+\epsilon}\Big)  
  	\\
& + \bigg( \frac{1}{2} - \alpha - \theta_{\star} \bigg) \big(f_{1}(\sigma_{\star}) - \hat{f}_{1}( \sigma_{\star} ) \big) + \int_{\frac{b\rho_{1}^{2b}}{1-\epsilon}}^{\sigma_{\star}}f(x)dx + \int_{\sigma_{\star}}^{\frac{b\rho_{2}^{2b}}{1+\epsilon}}\hat{f}(x)dx 
	\\
& + \theta_\star \ln \mathsf{Q} -\frac{\ln{\mathsf{Q}}}{2} \bigg(
1 - \frac{2\ln(\sigma_{2}/\sigma_{1}) + \ln{\mathsf{Q}}}{2 \ln(\rho_{2}/\rho_{1})}\bigg) 
 + \hat{F}_{3}^{(n,\epsilon)}.
\end{align*}
It only remains to show that the constants $\{C_j\}_1^4$ can be expressed as in the statement of Theorem \ref{thm:main thm hard}. This is easily verified for $C_1$ and $C_2$.
Using the identities
\begin{align*}
f_{1}(\sigma_{\star}) - \hat{f}_{1}( \sigma_{\star} ) & = \ln \mathsf{Q}, \qquad \mathcal{I}_{1} = \frac{\ln (2\sqrt{\pi})}{2},
	\\
\int_{\frac{b\rho_{1}^{2b}}{1-\epsilon}}^{\sigma_{\star}} \frac{b \rho_{1}^{2b} \mathsf{T}_{1}(b\rho_{1}^{2b};\vec{t},\vec{u})}{\Omega (x-b\rho_{1}^{2b})}dx & = \frac{b \rho_{1}^{2b} \mathsf{T}_{1}(b\rho_{1}^{2b};\vec{t},\vec{u})}{\Omega} \bigg( \ln \frac{\sigma_{\star}-b\rho_{1}^{2b}}{b\rho_{1}^{2b}} - \ln \frac{\epsilon}{1-\epsilon}\bigg),
	\\
-\int_{\sigma_{\star}}^{\frac{b\rho_{2}^{2b}}{1 + \epsilon}} \frac{b \rho_{2}^{2b} \hat{\mathsf{T}}_{1}(b\rho_{2}^{2b};\vec{t},\vec{u})}{b\rho_{2}^{2b}-x} dx & =  b \rho_{2}^{2b} \hat{\mathsf{T}}_{1}(b\rho_{2}^{2b};\vec{t},\vec{u}) \bigg(\ln\frac{b\rho_2^{2b}}{b\rho_2^{2b} - \sigma_{\star}}  + \ln \frac{\epsilon}{1+\epsilon}\bigg).
\end{align*}
long but straightforward calculations show that $C_{3}$ also can be written as in Theorem \ref{thm:main thm hard}. Moreover, substituting the expressions for $F_4$ and $\hat{F}_4$ obtained in Lemmas \ref{lemma: asymp of S2k final hard} and 
\ref{lemma: asymp of S4 final hard} into the relation $C_4 = F_4 + \hat{F}_4$, we infer that
\begin{align*}
& C_{4} = \sqrt{2}b\bigg(\rho_{1}^{b}\frac{ \rho_{1}^{2b} \mathsf{T}_{2}(b\rho_{1}^{2b};\vec{t},\vec{u}) -5 \mathsf{T}_{1}(b\rho_{1}^{2b};\vec{t},\vec{u}) }{ \Omega } - \rho_{2}^{b} \big( \rho_{2}^{2b} \hat{\mathsf{T}}_{2}(b\rho_{2}^{2b};\vec{t},\vec{u}) + 5 \hat{\mathsf{T}}_{1}(b\rho_{2}^{2b};\vec{t},\vec{u}) \big) \bigg)  \mathcal{I}
	\\ \nonumber
& + \frac{10\sqrt{2} b }{3} \bigg( \rho_{1}^{b} \frac{\mathsf{T}_{1}(b\rho_{1}^{2b};\vec{t},\vec{u})}{\Omega} + \rho_{2}^{b} \hat{\mathsf{T}}_{1}(b\rho_{2}^{2b};\vec{t},\vec{u}) \bigg) ( \mathcal{I}_2  - \mathcal{I}_4 )
	\\ \nonumber
& + \sqrt{2} b \bigg[ \rho_{1}^{b} \frac{\mathsf{T}_{1}(b\rho_{1}^{2b};\vec{t},\vec{u})}{\Omega} \bigg( \frac{2}{3} -  \rho_{1}^{2b} \frac{\mathsf{T}_{1}(b\rho_{1}^{2b};\vec{t},\vec{u})}{\Omega} \bigg) + \rho_{2}^{b} \hat{\mathsf{T}}_{1}(b\rho_{2}^{2b};\vec{t},\vec{u}) \bigg( \frac{2}{3} -  \rho_{2}^{2b} \hat{\mathsf{T}}_{1}(b\rho_{2}^{2b};\vec{t},\vec{u}) \bigg) \bigg] 
\mathcal{I}_3.
\end{align*}
Using the identities $\mathcal{I}_{3} = \mathcal{I}$ and $\mathcal{I}_{2}-\mathcal{I}_{4} = \mathcal{I}$, which can be obtained via integration by parts (see \cite[Lemma 2.10]{ACCL2022} for details), we conclude that $C_{4}$ also can be expressed as in Theorem \ref{thm:main thm hard}. This finishes the proof of Theorem \ref{thm:main thm hard}.
\end{proof}

\appendix

\section{Balayage of radially symmetric measures}\label{appendix:eq measure}
The measure $\mu_{h}$ is defined as the unique minimizer of the energy functional
\begin{align*}
I[\nu] = \iint \ln \frac{1}{|z-w|}\nu(d^{2}z)\nu(d^{2}w) + \int \hat{Q}(z) \nu(d^{2}z), \quad \hat{Q}(z) := \begin{cases}
|z|^{2b}, & \mbox{if } |z| \in [0,\rho_{1}]\cup[\rho_{2},+\infty), \\
+\infty, & \mbox{otherwise},
\end{cases}
\end{align*}
among all Borel probability measures $\nu$ on $\mathbb{C}$. In this appendix, we present a derivation of the expression \eqref{def of muh} for $\mu_h$. 

\medskip Let $C=\{z\in \mathbb{C}: |z| \in [0,\rho_{1}]\cup [\rho_{2},+\infty)\}$ and $G=\mathbb{C}\setminus C =\{z\in \mathbb{C}: |z| \in (\rho_{1},\rho_{2})\}$. The unique minimizer of 
\begin{align*}
\tilde{I}[\nu] = \iint \ln \frac{1}{|z-w|}\nu(d^{2}z)\nu(d^{2}w) + \int |z|^{2b} \nu(d^{2}z)
\end{align*}
among all Borel probability measures $\nu$ on $\mathbb{C}$ is given by $\mu(d^{2}z) = \chi_{[0,b^{-\frac{1}{2b}}]}(|z|) \frac{1}{4\pi} \Delta |z|^{2b}d^{\smash{2}}z = \smash{\chi_{[0,b^{-\frac{1}{2b}}]}(|z|)\frac{b^{\smash{2}}}{\pi}}|z|^{\smash{2b-2}}d^{\smash{2}}z$ (see \cite[Example IV.6.2]{SaTo}). We start with the following lemma.
\begin{lemma}
We have
\begin{align}\label{muhmuBal}
\mu_{h} = \mu \cdot \chi_{C} + \hat{\mu},
\end{align}
where $\hat{\mu} := \mbox{Bal}(\mu \cdot \chi_{G}\, , \partial G)$ is the balayage of $\mu \cdot \chi_{G}$ onto $\partial G = \{|z|=\rho_{1}\}\cup \{|z|=\rho_{2}\}$.
\end{lemma}
\begin{proof}
Let us first introduce some notation and background from \cite{SaTo}. Given a Borel probability measure $\nu$, define
\begin{align*}
U^{\nu}(z) := \int_\C\ln\frac 1 {|z-w|}\, \nu (d^{2}w), \qquad z \in \mathbb{C}.
\end{align*}
By \cite[Theorem I.1.3 (d) and (f)]{SaTo}, $\mu$ and $\mu_{h}$ satisfy the following conditions: there exist $F,F_{h} \in \mathbb{R}$ such that
\begin{align}
& \begin{cases}
U^{\mu}(z)+\frac{1}{2}|z|^{2b} \geq F, & z \in \mathbb{C}, \\
U^{\mu}(z)+\frac{1}{2}|z|^{2b} = F, & |z| \leq b^{-\frac{1}{2b}},
\end{cases} \label{cond for mu} \\
& \begin{cases}
U^{\mu_{h}}(z)+\frac{1}{2}\hat{Q}(z) \geq F_{h}, & z \in \mathbb{C}, \\
U^{\mu_{h}}(z)+\frac{1}{2}\hat{Q}(z) = F_{h}, & z \in \mathrm{supp}(\mu_{h}).
\end{cases} \label{cond for muh}
\end{align}
Furthermore, by \cite[Theorem I.3.3]{SaTo}, the conditions \eqref{cond for muh} uniquely characterize $\mu_{h}$ in the sense that if $\nu$ is a Borel probability measure with compact support which satisfies $U^{\nu}(z)+\frac{1}{2}\hat{Q}(z) \geq F_{\nu}$ for some constant $F_{\nu}$ and all $z \in \mathbb{C}$ and if $U^{\nu}(z)+\frac{1}{2}\hat{Q}(z) = F_{\nu}$ for all $z \in \mathrm{supp}(\nu)$, then $\nu = \mu_{h}$.

Let $\tilde{\mu}_{h}$ denote the right-hand side of \eqref{muhmuBal}. We will show that $\tilde{\mu}_{h}=\mu_{h}$ using the aforementioned uniqueness theorem, namely \cite[Theorem I.3.3]{SaTo}. Since $\hat{\mu}$ is the balayage of $\mu \cdot \chi_{G}$ onto $\partial G$, \cite[Theorem II.4.1]{SaTo} shows that
\begin{align*}
\begin{cases}
U^{\hat{\mu}}(z) \leq U^{\mu \cdot \chi_{G}}(z), & z \in \mathbb{C}, \\
U^{\hat{\mu}}(z) = U^{\mu \cdot \chi_{G}}(z), & z \in C.
\end{cases}
\end{align*}
Combining the above with \eqref{cond for mu} and using that $U^{\mu}=U^{\mu \cdot \chi_{C}}+U^{\mu \cdot \chi_{G}}$, we obtain (recall that $\hat{Q}(z) = + \infty$ for $z \in G$)
\begin{align*}
& \begin{cases}
U^{\mu \cdot \chi_{C}}(z)+U^{\hat{\mu}}(z)+\frac{1}{2}\hat{Q}(z) \geq F, & |z| \notin [0,\rho_{1}]\cup[\rho_{2},b^{-\frac{1}{2b}}], \\
U^{\mu \cdot \chi_{C}}(z)+U^{\hat{\mu}}(z)+\frac{1}{2}\hat{Q}(z) = F, & |z| \in [0,\rho_{1}]\cup[\rho_{2},b^{-\frac{1}{2b}}].
\end{cases}
\end{align*}
Since $U^{\mu \cdot \chi_{C}}+U^{\hat{\mu}}=U^{\tilde{\mu}_{h}}$, \cite[Theorem I.3.3]{SaTo} applies and gives $\tilde{\mu}_{h}=\mu_{h}$.
\end{proof}

It only remains to compute $\hat{\mu}:=\mbox{Bal}(\mu \cdot \chi_{G} \, , \partial G)$ explicitly. 
Recall that 
\begin{align*}
\mu(d^{2}z) =  \smash{\frac{b^{\smash{2}}}{\pi}}|z|^{\smash{2b-2}}\chi_{[0,b^{-\frac{1}{2b}}]}(r)d^{\smash{2}}z, \qquad z=re^{i\theta}, \; r\geq 0, \; \theta \in (-\pi,\pi],
\end{align*}
and that $0<\rho_1<\rho_2<b^{-\frac{1}{2b}}$. Hence $\mu \cdot \chi_{G}$ is the radially symmetric measure given by
\begin{align*}
(\mu \cdot \chi_{G})(d^{2}z) = f(r)\,\chi_{(\rho_{1},\rho_{2})}(r)\, rdr \frac{d\theta}{\pi},
\end{align*}
where $f(r):=b^{2}r^{2b-2}$. Fix $z\in \C\setminus \overline{G}$ and compute the logarithmic potential
\begin{align*}
U^{\mu\cdot \chi_{G}}(z)&=\int_\C\ln\frac 1 {|z-w|}\, (\mu \cdot \chi_{G})(d^{2}w) = 2\int_{\rho_1}^{\rho_2}rf(r)\, dr\frac{1}{2\pi}\int_0^{2\pi}\ln\frac 1 {|z-re^{i\theta}|}\, d\theta.
\end{align*}
Now we use the familiar integral (e.g. \cite[Chapter 0, page 22]{SaTo})
\begin{align}\label{useful simple integral}
\frac 1 {2\pi} \int_0^{2\pi}\ln\frac 1 {|z-re^{i\theta}|}\, d\theta=\begin{cases}
\ln\frac 1 r, & |z|\le r, \\
\ln \frac 1 {|z|}, & |z| \ge r.
\end{cases}
\end{align}
It follows that
\begin{align*}
U^{\mu\cdot \chi_{G}}(z) = \begin{cases}
\mathcal{C}_1, & \mbox{if } |z| \leq \rho_{1}, \\
\mathcal{C}_2\ln \frac 1 {|z|}, & \mbox{if } |z| \geq \rho_{2}, 
\end{cases} \quad \mbox{where } \quad \mathcal{C}_1=2\int_{\rho_1}^{\rho_2}rf(r)\ln\frac 1 r\, dr, \quad \mathcal{C}_2=2\int_{\rho_1}^{\rho_2}rf(r)\, dr.
\end{align*}
We make the ansatz that $\hat{\mu}$ has the form (see \cite[Theorem II.4.1]{SaTo})
\begin{align}\label{hatmu}
\hat{\mu}(d^{2}z)=\sigma_1\,\delta_{\rho_1}(r)\, dr \frac{d\theta}{2\pi}+\sigma_2\,\delta_{\rho_2}(r)\, dr \frac{d\theta}{2\pi}, \qquad z=re^{i\theta}, \; r\geq 0, \; \theta \in (-\pi,\pi],
\end{align}
for some constants $\sigma_1$ and $\sigma_2$. Using again \eqref{useful simple integral} we obtain
\begin{align*}
U^{\hat{\mu}}(z)=  \begin{cases}
\sigma_1 \ln \frac{1}{\rho_1}+\sigma_2\ln\frac{1}{\rho_2}, & \mbox{if } |z| \leq \rho_{1}, \\
(\sigma_1+\sigma_2)\ln\frac 1{|z|}, & \mbox{if } |z| \geq \rho_{2}.
\end{cases}
\end{align*}
By definition of $\hat{\mu}$, we must have $U^{\mu\cdot \chi_{G}}(z) = U^{\hat{\mu}}(z)$ for all $z \notin \overline{G}$, so we have the system
\begin{align*}
\begin{cases} 
\sigma_1\ln \frac 1 {\rho_1}+\sigma_2\ln\frac 1 {\rho_2}&\hspace{-.25cm}=\mathcal{C}_1, \\
\sigma_1+\sigma_2 &\hspace{-.25cm}=\mathcal{C}_2.
\end{cases}
\end{align*}
The solution is
\begin{align}\label{sigma12 in terms of C12}
\sigma_1= \frac{\mathcal{C}_1-\mathcal{C}_2\ln\frac 1 {\rho_2}}{\ln \frac {\rho_2}{\rho_1}},\qquad \sigma_2= \frac{\mathcal{C}_2\ln\frac 1{\rho_1}-\mathcal{C}_1}{\ln\frac {\rho_2}{\rho_1}}.
\end{align}
Recalling that $f(r)= b^2r^{2b-2}$, we get
\begin{align*}
& \mathcal{C}_1=2\int_{\rho_1}^{\rho_2}rf(r)\ln\frac 1 r\, dr =-b\rho_2^{2b}\ln(\rho_2)+b\rho_1^{2b} \ln(\rho_1) + \tfrac{1}{2} (\rho_2^{2b}-\rho_1^{2b}), \\
& \mathcal{C}_2=2\int_{\rho_1}^{\rho_2}rf(r)\, dr = b(\rho_2^{2b}-\rho_1^{2b}).
\end{align*}
Substituting the above expressions into \eqref{sigma12 in terms of C12} and recalling (\ref{muhmuBal}) and (\ref{hatmu}), we arrive at \eqref{def of muh}.

\section{Proof of Lemma \ref{T0lemma}}\label{T0lemmaapp}
Define $\phi_1$ and $\phi_2$ by
\begin{align*}
& \phi_1(x;\vec{t},\vec{u}) :=1 + \mathsf{T}_{0}(x;\vec{t},\vec{u}) + \hat{\mathsf{T}}_{0}(b\rho_{2}^{2b};\vec{u})
= 1 + \sum_{\ell=1}^{m} \omega_{\ell}e^{-\frac{t_{\ell}}{b}(x-b\rho_{1}^{2b})} + \sum_{\ell=m+1}^{2m} \omega_{\ell},
	\\
& \phi_2(x;\vec{t},\vec{u}) :=1 - \hat{\mathsf{T}}_{0}(x;\vec{t},\vec{u}) + \hat{\mathsf{T}}_{0}(b\rho_{2}^{2b};\vec{u})
= 1 - \sum_{\ell=m+1}^{2m} \omega_{\ell} e^{-\frac{t_{\ell}}{b}(b\rho_{2}^{2b}-x)} + \sum_{\ell=m+1}^{2m} \omega_{\ell}.
\end{align*}
A calculation gives
\begin{align*}
\partial_{u_{2m}} \phi_1 =&\; e^{u_1 + \dots + u_{2m}} e^{-\frac{t_{1}}{b}(x-b\rho_{1}^{2b})} + \sum_{\ell=2}^{m} e^{u_\ell + \dots + u_{2m}} (e^{-\frac{t_{\ell}}{b}(x-b\rho_{1}^{2b})} - e^{-\frac{t_{\ell-1}}{b}(x-b\rho_{1}^{2b})})
	\\
& + e^{u_{m+1} + \dots + u_{2m}} (1 - e^{-\frac{t_{m}}{b}(x-b\rho_{1}^{2b})}),
	\\
\partial_{u_{2m}} \phi_2 =&\; e^{u_{m+1} + \dots + u_{2m}} (1 - e^{-\frac{t_{m+1}}{b}(b\rho_{2}^{2b}-x)}) + \sum_{\ell=m+2}^{2m} e^{u_\ell + \dots + u_{2m}} (e^{-\frac{t_{\ell-1}}{b}(b\rho_{2}^{2b}-x)} - e^{-\frac{t_{\ell}}{b}(b\rho_{2}^{2b}-x)}).
\end{align*}
Using that $t_{1}>\dots>t_{m} \geq 0$, we see that $\partial_{u_{2m}} \phi_1(x;\vec{t},\vec{u}) > 0$ for any $x \geq b\rho_{1}^{2b}$, and using that $0 \leq t_{m+1} < \dots < t_{2m}$, we see that $\partial_{u_{2m}} \phi_2(x;\vec{t},\vec{u}) \geq 0$ for any $x \leq b\rho_{2}^{2b}$. In view of the limits
$$\lim_{u_{2m} \to -\infty} \phi_1(x;\vec{t},\vec{u}) = 0, \qquad \lim_{u_{2m} \to -\infty} \phi_2(x;\vec{t},\vec{u}) = e^{-\frac{t_{2m}}{b}(b\rho_{2}^{2b}-x)} > 0,$$
the desired inequalities follow.

\section{Uniform asymptotics of the incomplete gamma function}\label{section:uniform asymp gamma}
\begin{lemma}[{From \cite[formula 8.11.2]{NIST}}]\label{lemma:various regime of gamma}
Let $a>0$ be fixed. As $z \to +\infty$,
\begin{align*}
\gamma(a,z) = \Gamma(a) + \bigO(e^{-\frac{z}{2}}).
\end{align*}
\end{lemma}

\begin{lemma}[{From \cite[Section 11.2.4]{Temme}}]\label{lemma: uniform}
We have
\begin{align*}
& \frac{\gamma(a,z)}{\Gamma(a)} = \frac{1}{2}\mathrm{erfc}(-\eta \sqrt{a/2}) - R_{a}(\eta), \qquad R_{a}(\eta) = \frac{e^{-\frac{1}{2}a \eta^{2}}}{2\pi i}\int_{-\infty}^{\infty}e^{-\frac{1}{2}a u^{2}}g(u)du,
\end{align*}
where $\mathrm{erfc}$ is the complementary error function,
\begin{align}\label{lol8}
& \lambda = \frac{z}{a}, \qquad \eta = (\lambda-1)\sqrt{\frac{2 (\lambda-1-\ln \lambda)}{(\lambda-1)^{2}}}, \qquad g(u) := \frac{dt}{du} \frac{1}{\lambda -t} + \frac{1}{u +i\eta}.
\end{align}
The variables $t$ and $u$ are related by the bijection $t \mapsto u$ from $\mathcal{L} := \{\frac{\theta}{\sin \theta} e^{i\theta} : - \pi < \theta < \pi\}$ to $\R$ given by
\begin{align}\label{u t temme}
u = -i(t-1) \sqrt{\frac{2(t-1-\ln t)}{(t-1)^2}}, \qquad  t \in \mathcal{L}.
\end{align}
The principal branch is taken for the roots in \eqref{lol8} and \eqref{u t temme}. In addition, 
\begin{align}\label{asymp of Ra}
& R_{a}(\eta) \sim \frac{e^{-\frac{1}{2}a \eta^{2}}}{\sqrt{2\pi a}}\sum_{j=0}^{\infty} \frac{c_{j}(\eta)}{a^{j}}, \qquad \mbox{as } a \to + \infty
\end{align}
uniformly for $z \in [0,\infty)$ where all coefficients $c_{j}(\eta)$ are bounded functions of $\eta \in \mathbb{R}$ (i.e. bounded for $\lambda \in (0,+\infty)$). The first two coefficients are given by (see \cite[p. 312]{Temme})
\begin{align*}
c_{0}(\eta) = \frac{1}{\lambda-1}-\frac{1}{\eta}, \qquad c_{1}(\eta) = \frac{1}{\eta^{3}}-\frac{1}{(\lambda-1)^{3}}-\frac{1}{(\lambda-1)^{2}}-\frac{1}{12(\lambda-1)}.
\end{align*}
More generally, we have
\begin{align}\label{recursive def of the cj}
c_{j}(\eta) = \frac{1}{\eta} \frac{d}{d\eta}c_{j-1}(\eta) + \frac{\gamma_{j}}{\lambda-1}, \; \qquad j \geq 1,
\end{align}
where the $\gamma_{j}$ are the Stirling coefficients
\begin{align*}
\gamma_{j} = \frac{(-1)^{j}}{2^{j} \, j!} \bigg[ \frac{d^{2j}}{dx^{2j}} \bigg( \frac{1}{2}\frac{x^{2}}{x-\ln(1+x)} \bigg)^{j+\frac{1}{2}} \bigg]_{x=0}.
\end{align*}
In particular, the following hold:
\item[(i)] As $a \to +\infty$, $\gamma(a,\lambda a) = \Gamma(a)\big(1 + \bigO(e^{-\frac{a \eta^{2}}{2}})\big)$ uniformly for $\lambda \geq 1+\delta$ for each fixed $\delta > 0$.

\item[(ii)] As $a \to +\infty$, $\gamma(a,\lambda a) = \Gamma(a)\bigO(e^{-\frac{a \eta^{2}}{2}})$ uniformly for $\lambda$ in compact subsets of $(0,1)$,
\end{lemma} 

The following lemma is essentially a result of Tricomi \cite{Tricomi}; however, the coefficients appearing in Lemma \ref{lemma: asymp of gamma for lambda one over sqrt away from 1} below are written in a non-recursive way using \cite[Lemma A.4]{ACCL2022}.
\begin{lemma}[{From \cite[Lemma A.4]{ACCL2022}}]\label{lemma: asymp of gamma for lambda one over sqrt away from 1}
Let $N \geq 0$ be an integer, let $\eta$ be as in \eqref{lol8}, let $\varphi_j(\lambda) := \frac{(-1)^{j+1} (2j-1)!!}{\eta^{2j+1}}$, and let $S(\varphi_j(\lambda))$ denote the singular part of $\varphi_j(\lambda)$ at $\lambda = 1$, i.e., $S(\varphi_j(\lambda))$ is the sum of the singular terms in the Laurent expansion of $\varphi_j(\lambda)$ at $\lambda = 1$. The first $S(\varphi_j(\lambda))$ are given by
  \begin{align*}
& S(\varphi_0(\lambda)) = -\frac{1}{\lambda-1}, \qquad
S(\varphi_1(\lambda)) = \frac{1}{(\lambda-1)^{3}}+\frac{1}{(\lambda-1)^{2}}+\frac{1}{12(\lambda-1)},
	\\
& S(\varphi_2(\lambda)) = -\frac{3}{(\lambda -1)^5}-\frac{5}{(\lambda -1)^4} - \frac{25}{12 (\lambda
   -1)^3} - \frac{1}{12 (\lambda -1)^2} - \frac{1}{288 (\lambda -1)}.
\end{align*}

\item[(i)] As $a \to +\infty$, uniformly for $\lambda \geq 1+ \frac{1}{\sqrt{a}}$,
\begin{align*}
\frac{\gamma(a,\lambda a)}{\Gamma(a)} & = 1 + \frac{e^{-\frac{a}{2}\eta^2}}{\sqrt{2\pi}} \bigg\{\sum_{j=0}^{N-1} \frac{S(\varphi_j(\lambda))}{a^{j+\frac{1}{2}}} + \bigO\bigg(\frac{1}{a^{N+\frac{1}{2}}}\bigg) +  \bigO\bigg(\frac{1}{(a \eta^2)^{N+\frac{1}{2}}}\bigg)\bigg\}.
\end{align*}

\item[(ii)] As $a \to +\infty$, uniformly for $\lambda \in [\epsilon, 1-\frac{1}{\sqrt{a}}]$ for any fixed $\epsilon > 0$,
\begin{align*}
\frac{\gamma(a,\lambda a)}{\Gamma(a)} = \frac{e^{-\frac{a}{2}\eta^2}}{\sqrt{2\pi}} \bigg\{\sum_{j=0}^{N-1} \frac{S(\varphi_j(\lambda))}{a^{j+\frac{1}{2}}} + \bigO\bigg(\frac{1}{a^{N+\frac{1}{2}}}\bigg) +  \bigO\bigg(\frac{1}{(a \eta^2)^{N+\frac{1}{2}}}\bigg)\bigg\}.
\end{align*}
\end{lemma}

\paragraph{Acknowledgements.} CC acknowledges support from the Swedish Research Council, Grant No. 2021-04626. JL acknowledges support from the Swedish Research Council, Grant No. 2021-03877, and the Ruth and Nils-Erik Stenb\"ack Foundation.

\end{document}